\documentclass[reprint,amsfonts, amssymb, amsmath,  showkeys, nofootinbib,pra, superscriptaddress, twocolumn,longbibliography]{revtex4-1}
\usepackage{float}
\makeatletter
\let\newfloat\newfloat@ltx
\makeatother
\usepackage[english]{babel}
\usepackage[utf8]{inputenc}
\usepackage{graphics}
\usepackage{selinput}

\usepackage[normalem]{ulem}
\usepackage{braket}
\usepackage{amsthm}
\usepackage{mathtools}
\usepackage{physics}
\usepackage[dvipsnames]{xcolor}
\usepackage{graphicx}
\usepackage[left=16mm,right=16mm,top=35mm,columnsep=15pt]{geometry} 
\usepackage{adjustbox}
\usepackage{placeins}
\usepackage[T1]{fontenc}
\usepackage{lipsum}
\usepackage{csquotes}
\usepackage{hyphenat}
\usepackage{mathbbol} 
\usepackage{multirow}
\usepackage[linesnumbered,ruled,vlined]{algorithm2e}

\usepackage{color}
\usepackage{tabularray}
\UseTblrLibrary{diagbox}
\definecolor{WildSand}{rgb}{0.968,0.968,0.968}
\definecolor{NewYorkPink}{rgb}{0.925,0.792,0.792}
\definecolor{Feijoa}{rgb}{0.537,0.85,0.76}

\definecolor{RED}{rgb}{1.,0,0}
\SetKwInput{kwInit}{Init}

\makeatletter 
\renewcommand\onecolumngrid{
\do@columngrid{one}{\@ne}
\def\set@footnotewidth{\onecolumngrid}
\def\footnoterule{\kern-6pt\hrule width 1.5in\kern6pt}
}
\renewcommand\twocolumngrid{
        \def\footnoterule{
        \dimen@\skip\footins\divide\dimen@\thr@@
        \kern-\dimen@\hrule width.5in\kern\dimen@}
        \do@columngrid{mlt}{\tw@}
}
\makeatother

\def\KC{\mathcal{K}}
\def\LC{\mathcal{L}}

\usepackage[makeroom]{cancel}
\usepackage[toc,page]{appendix}
\usepackage[colorlinks=true,citecolor=blue,linkcolor=magenta]{hyperref}

\usepackage{tikz}
\tikzset{every picture/.style=remember picture}

\usepackage{MnSymbol}

\newcommand{\poly}{\operatorname{poly}}

\newcommand{\Ebb}{\mathbb{E}}

\newcommand{\DC}{\mathcal{D}}

\newcommand{\MC}{\mathcal{M}}
\newcommand{\NC}{\mathcal{N}}
\newcommand{\OC}{\mathcal{O}}
\newcommand{\PC}{\mathcal{P}}
\newcommand{\QC}{\mathcal{Q}}

\newcommand{\UC}{\mathcal{U}}

\newcommand{\XC}{\mathcal{X}}

\newcommand{\Var}{{\rm Var}}
\newcommand{\Cov}{{\rm Cov}}

\renewcommand{\geq}{\geqslant}
\renewcommand{\leq}{\leqslant}

\DeclareMathOperator*{\argmax}{arg\,max}

\renewcommand{\vec}[1]{\boldsymbol{#1}}  

\newcommand*{\g}{\text{guess}}

\newcommand{\bs}{\textsf{BS}}

\newcommand{\lm}{\lambda }

\newcommand{\id}{\mathbb{1}}
\newcommand{\qt}{\tilde{q}}

\newcommand{\qth}{q_{\thv}}
\newcommand{\qtth}{\tilde{q}_{\thv}}
\newcommand{\pt}{\tilde{p}}
\newcommand{\thv}{\vec{\theta}}
\newcommand{\xv}{\vec{x}}
\newcommand{\yv}{\vec{y}}
 
\newcommand{\alv}{\vec{\alpha} }
\newcommand{\btv}{\vec{\beta} }

\def\be{\begin{equation}}
\def\ee{\end{equation}}
\def\bs{\begin{split}}
\def\e{\end{split}}
\def\ba{\begin{eqnarray}}
\def\bea{\begin{eqnarray}}

\def\tea{\end{eqnarray}}
\def\ea{\end{eqnarray}}
\def\eea{\end{eqnarray}}

\def\d{\delta}

\def\l{\lambda}
\def\d{\delta}

\newtheorem{theorem}{Theorem}
\newtheorem{lemma}{Lemma}
\newtheorem{corollary}{Corollary}

\newtheorem{proposition}{Proposition}
\newtheorem*{proposition*}{Proposition}
\newtheorem{supplemental_proposition}{Supplemental Proposition}

\newtheorem{definition}{Definition}

\usepackage{amssymb}
\usepackage{dsfont}

\usepackage[normalem]{ulem}
\usepackage{hyperref}

\def\be{\begin{equation}}
\def\te{\end{equation}}
\def\ee{\end{equation}}
\def\ba{\begin{eqnarray}}
\def\bea{\begin{eqnarray}}

\def\tea{\end{eqnarray}}
\def\ea{\end{eqnarray}}
\def\eea{\end{eqnarray}}

\newcommand{\vol}{\boldsymbol{\mathcal{V}}} 
\newcommand{\pauli}{\mathfrak{p}_n}

\begin{document}

\title{Trainability barriers and opportunities in quantum generative modeling}

\author{Manuel S.~Rudolph}
\thanks{The first three authors contributed equally to this work.}
\affiliation{Institute of Physics, Ecole Polytechnique F\'{e}d\'{e}rale de Lausanne (EPFL), CH-1015 Lausanne, Switzerland}

\author{Sacha Lerch}
\thanks{The first three authors contributed equally to this work.}
\affiliation{Institute of Physics, Ecole Polytechnique F\'{e}d\'{e}rale de Lausanne (EPFL), CH-1015 Lausanne, Switzerland}

\author{Supanut Thanasilp$^*$}
\thanks{supanut.thanasilp@gmail.com}

\affiliation{Institute of Physics, Ecole Polytechnique F\'{e}d\'{e}rale de Lausanne (EPFL), CH-1015 Lausanne, Switzerland}
\affiliation{Chula Intelligent and Complex Systems, Department of Physics, Faculty of Science, Chulalongkorn University, Bangkok, Thailand, 10330}

\author{Oriel Kiss}
\affiliation{European Organization for Nuclear Research (CERN), Geneva 1211, Switzerland}
\affiliation{Department of Nuclear and Particle Physics, University of Geneva, Geneva 1211, Switzerland}

\author{Oxana Shaya}
\affiliation{Institute of Physics, Ecole Polytechnique F\'{e}d\'{e}rale de Lausanne (EPFL), CH-1015 Lausanne, Switzerland}

\author{Sofia Vallecorsa}
\affiliation{European Organization for Nuclear Research (CERN), Geneva 1211, Switzerland}

\author{Michele Grossi}
\affiliation{European Organization for Nuclear Research (CERN), Geneva 1211, Switzerland}

\author{Zo\"{e} Holmes}
\affiliation{Institute of Physics, Ecole Polytechnique F\'{e}d\'{e}rale de Lausanne (EPFL), CH-1015 Lausanne, Switzerland}

\date{\today}

\begin{abstract}
Quantum generative models provide inherently efficient sampling strategies and thus show promise for achieving an advantage using quantum hardware. In this work, we investigate the barriers to the trainability of quantum generative models posed by barren plateaus and exponential loss concentration. We explore the interplay between explicit and implicit models and losses, and show that using quantum generative models with explicit losses such as the KL divergence leads to a new flavour of barren plateaus. In contrast, the implicit Maximum Mean Discrepancy loss can be viewed as the expectation value of an observable that is either low-bodied and provably trainable, or global and untrainable depending on the choice of kernel. In parallel, we find that solely low-bodied implicit losses cannot in general distinguish high-order correlations in the target data, while some quantum loss estimation strategies can. We validate our findings by comparing different loss functions for modelling data from High-Energy-Physics.
\end{abstract}

\maketitle

\section{Introduction}
The advent of quantum computing has opened up new avenues for solving classically intractable problems~\cite{harrow2017quantum,QML_Lloyd,huang2021quantum, daley2022practical}. 
Naturally, researchers gravitate towards finding the first high-value applications that could be tackled with near- and mid-term quantum devices~\cite{preskill2018quantum}. This includes not only speed-ups~\cite{harrow2009quantum,lloyd2014quantum,huang2021power,huang2021quantum}, but potentially superior memory efficiency~\cite{anschuetz2022interpretable} or concrete qualitative improvements~\cite{alcazar2020classical,QCBM_generalisation}. Quantum machine learning (QML) is one of the domains that attracts this attention~\cite{QML_Lloyd}. Quantum systems, in being inherently probabilistic, are particularly well suited to generative modelling tasks~\cite{PerdomoOrtiz2017}. Generative models aim to learn the underlying distribution of a dataset and thereby provide a means of generating new data samples that are similar to the original data. As well as providing a naturally efficient means of generating samples, quantum generative models can provably encode probability distributions that are out of reach for classical models~\cite{coyle2020born, sweke2020learnability, gao2021enhancing}, and have been proposed for various applications, such as handwritten digits~\cite{rudolph2022generation}, finance~\cite{coyle2020generativeFinance} or High\hyp Energy\hyp Physics~\cite{Kiss_PRA,Delgado_QCBM}.

Despite the excitement surrounding the potential of generative QML, there remain substantial questions concerning its scalability. This is non-trivial to assess since implementations are constrained by hardware limitations to small-scale proof-of-principle problems~\cite{Hamilton2018,leyton2019robust,coyle2020generativeFinance, Zhu2018,rudolph2022generation}. Thus analytic results are essential to guide the successful development of this field. Of particular concern is the growing body of literature on cost function concentration and barren plateaus~\cite{mcclean2018barren,arrasmith2021equivalence, larocca2021diagnosing,cerezo2020impact,arrasmith2020effect,holmes2021barren,zhao2021analyzing, thanasilp2022exponential}, where loss function values can exponentially concentrate around a fixed value and loss gradients vanish exponentially with growing problem size. This phenomenon, which exponentially increases the resources required for training, originates from different sources~\cite{mcclean2018barren, holmes2021connecting, larocca2021diagnosing, cerezo2020cost, marrero2020entanglement, patti2020entanglement, wang2020noise, wang2021can, thanasilp2021subtleties, leone2022practical, li2022concentration}, and has been studied in a number of architectures~\cite{mcclean2018barren, larocca2021diagnosing, napp2022quantifying, pesah2020absence, larocca2022group, tangpanitanon2020expressibility, sharma2020trainability,rudolph2021orqviz,thanasilp2022exponential} as well as classes of cost function~\cite{cerezo2020cost, napp2022quantifying, thanasilp2021subtleties}. However, its impact on quantum generative modelling thus far has, except for the odd notable exception~\cite{kieferova2021quantum}, and very recent developments~\cite{coopmans2023sample}, been largely overlooked.

In this work, we provide a thorough study of trainability barriers and opportunities in quantum generative modelling. Critical to our analysis is the distinction between explicit and implicit models and losses. Explicit models provide efficient access directly to the model probabilities, whereas implicit models only provide samples drawn from their distribution~\cite{mohamed2016learning}. Quantum circuit Born machines (QCBMs)~\cite{benedetti2019generative}, the focus of this work, encode a probability distribution in an $n$-qubit pure state and thus are a paradigmatic example of an implicit model. Mirroring the capabilities of the models, explicit losses are those that are formulated explicitly in terms of the model and target probabilities, whereas implicit losses compare samples from the model and the training distribution. The most commonly used explicit loss for quantum generative models is the Kullbach-Leibler (KL) divergence~\cite{KLD}. Other examples include the Jensen-Shannon divergence (JSD), the total variation distance (TVD) and the classical fidelity. The Maximum Mean Discrepancy (MMD)~\cite{Gretton2012mmd} on the other hand is one of the leading examples of an implicit loss.

\begin{table}[t]
\centering
\begin{tblr}{width = \columnwidth,
  colspec = {Q[150]Q[223]Q[308]Q[240]},
  row{even} = {c},
  row{3} = {c},
  row{5} = {c},
  cells = {WildSand,c},
  cell{1}{1} = {r=2}{c},
  cell{1}{2} = {c=2}{c},
  cell{1}{4} = {r=2}{c},
  cell{3}{2} = {r=3}{NewYorkPink},
  cell{3}{3} = {r=2}{Feijoa},
  cell{3}{4} = {Feijoa},
  cell{4}{4} = {Feijoa},
  cell{5}{3} = {NewYorkPink},
  cell{5}{4} = {NewYorkPink},
  vlines,
  hline{1,6} = {-}{0.08em},
  hline{2} = {2-3}{},
  hline{3} = {-}{},
  hline{4} = {1,4}{},
  hline{5} = {1,3-4}{},
}
\textbf{Circuit depth}& \textbf{Explicit loss (pairwise)} &                           & {\textbf{Implicit loss}\\ \textbf{(MMD)}} \\
         & Conventional strategy            & Quantum strategy                  &                         \\
Product & {\textbf{No}\\(Corollary~\ref{coro:untrain})}        & {\textbf{Yes}\\(Local Quantum Fidelity~\cite{cerezo2020cost})} & {\textbf{Yes}\\($\sigma\in\Theta(n)$, Theorem~\ref{thm:mmd-sigma})}      \\
Shallow &                          &                           & {\textbf{Yes}\\($\sigma\in\Theta(n)$, Theorem~\ref{thm:mmd-train-general})}   \\
Deep     &                          & \textbf{No}~\cite{mcclean2018barren, holmes2021connecting}                        & \textbf{No}~\cite{mcclean2018barren, holmes2021connecting}                      
\end{tblr}
\caption{\textbf{Summary of our main results.}
This table summarizes our key analytical results on the trainability of different loss functions in quantum generative modelling tasks. 
Without a strong inductive bias, pairwise explicit losses are untrainable for all circuit depths with the conventional sampling strategy. 
A quantum strategy could be utilised to efficiently estimate the local quantum fidelity, Eq.~\eqref{eq:local-quantum-fidelity}, which is trainable for a shallow-depth circuit. 
The MMD using a classical Gaussian kernel with a linearly-scaled bandwidth ($\sigma \in \Theta(n)$) is expected to be trainable for a shallow-depth circuits. Note that `Yes' here indicates the existence of regimes with trainability guarantees- it does not preclude untrainable regimes including, for example, the use of global quantum fidelity or the MMD with a fixed bandwidth.}
\label{table:summary-of-results}
\end{table}

Here we argue that the tension between using an implicit generative model (providing only samples) with an explicit loss (requiring access to probabilities) leads to a new flavour of barren plateau. This result disqualifies all before-mentioned explicit losses, and crucially the KL divergence, for efficient training of QCBMs without a strong inductive bias towards the target distribution. In contrast, the MMD as an implicit loss exhibits more nuanced behaviour and can be either trainable or untrainable. By viewing the classical MMD loss as the expectation value of a quantum observable, we show that varying the bandwidth parameter of a Gaussian kernel interpolates the MMD loss between a loss composed of predominantly global terms and one composed of low-bodied terms with either exponentially or polynomially decaying loss variances in the number of qubits. In particular, we derive a polynomial lower bound on the loss for a wide family of different classes of structured and unstructured models that depends only on the effective entanglement light cone of the circuit.  These results are summarised in Table~\ref{table:summary-of-results}.

In parallel, we provide insights into how the globality of a generative loss affects the types of correlations in a dataset that can reliably be learned. In particular, we show that a $k$-bodied loss (see Fig.~\ref{fig:weight_mmd_operator}) cannot distinguish between distributions that agree on all $k$-marginals but disagree about higher-order correlations. Hence we argue that in the context of quantum generative modelling it is advantageous to train on \textit{full-bodied} losses, that is losses containing both low and high-bodied terms, rather than the purely local losses advocated elsewhere in quantum machine learning.
The MMD is then a promising candidate choice for the training of QCBMs as its bodyness can be controlled via the bandwidth parameter.

We additionally expand the pool of viable loss functions by proposing a new local quantum fidelity-type loss which leverages what we call a quantum strategy for evaluating losses. This is to be contrasted with the conventional measurement strategy which simply uses samples from the model distribution in the computational basis. We provide an efficient training protocol using the local quantum fidelity loss with provable trainability guarantees.

Finally, we support our analysis with a comparison of the performance of the KL divergence, MMD and local quantum fidelity losses for modelling High\hyp Energy\hyp Physics (HEP) data. Specifically, we consider electron energy depositions in the electromagnetic calorimeter (ECAL) part of detectors involved in a typical proton\hyp proton collision experiment at the LHC. We learn to generate hits in the detector as black and white images of various sizes, with up to 16 qubits. We confirm that the properly-tuned MMD and the local quantum fidelity losses remain trainable using a restrictive shot budget, while training with the KL divergence becomes increasingly futile. 

\begin{figure*}
    \centering
    \includegraphics[width=0.8\linewidth]{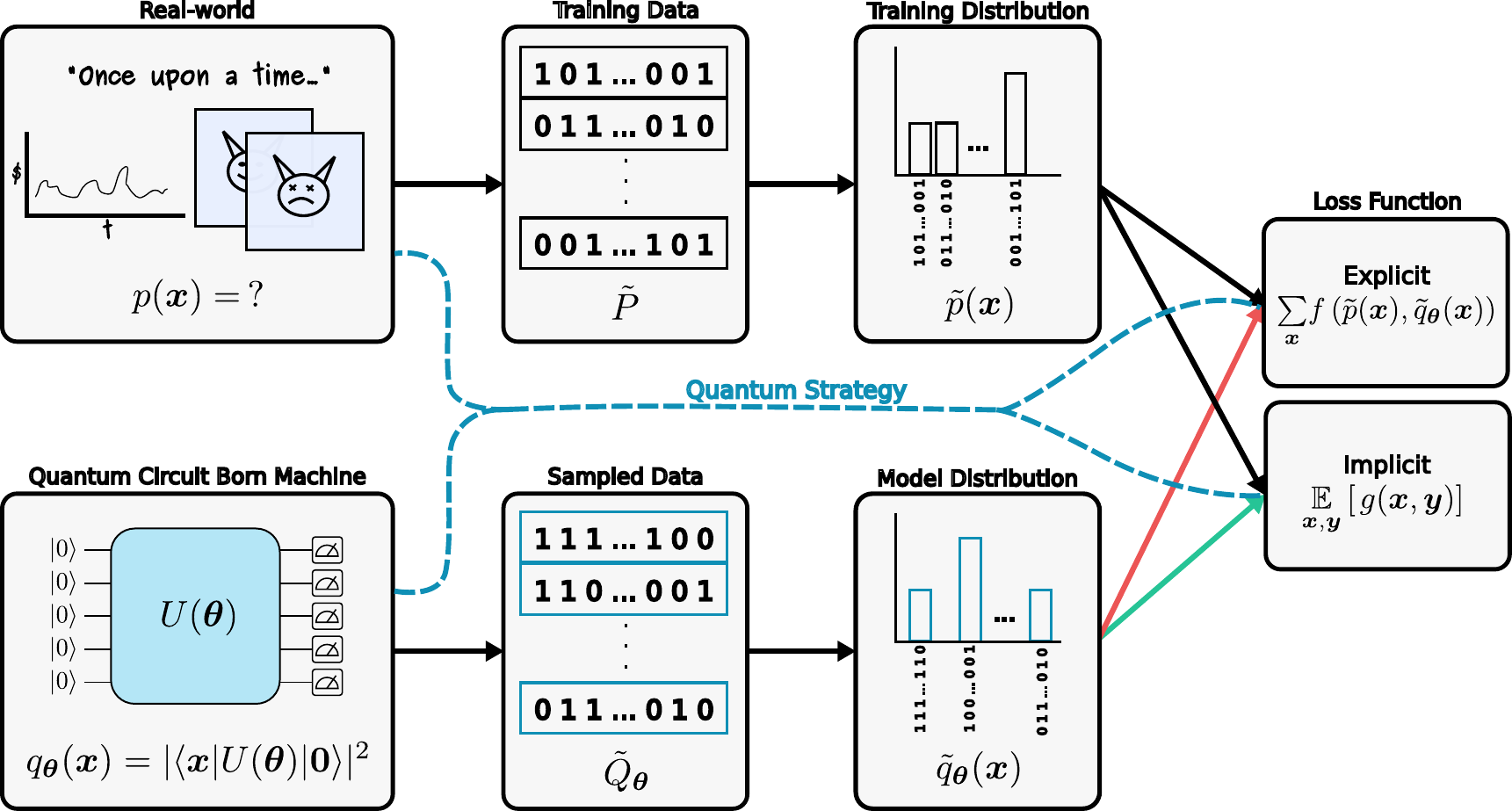}
    \caption{\textbf{The generative modelling framework using quantum circuit Born machines}. Given a training dataset $\tilde{P}$ with distribution $\pt(\xv)$ over discrete data samples $\xv$, the goal of a QCBM is to learn a distribution $\qth(\xv)$ which models the real-world distribution $p(\xv)$ from which the training data itself was sampled. This is done by tuning the parameters $\thv$ of a parametrized quantum circuit such that the QCBM minimizes a loss function that estimates the distance between the model and the training distribution. The QCBM is an implicit model and can thus in general not be paired with an explicit loss function, but it may be trainable using an implicit loss. In contrast to the conventional loss estimation strategy (solid lines) of generating a set of samples $\tilde{Q}_{\thv}$ and forming an empirical distribution $\qtth(\xv)$, strategies that are `more quantum' (dashed lines) can be employed with the aim of allowing QCBMs to be trained with loss functions which conventionally appear explicit. }
    \label{fig:generative_schematic}
\end{figure*}

\section{Results}

\subsection{Framework}

The goal of generative modelling is to use samples from a target distribution $p(\xv)$ to learn a model of $p(\xv)$ which can be used to generate new samples.
More concretely, as sketched in Fig.~\ref{fig:generative_schematic}, a generative model takes as input a training dataset $\Tilde{P}$ consisting of $M = |\Tilde{P}|$ samples drawn from the target distribution $p(\xv)$. This training set can be used to construct the empirical probability distribution $\tilde{p}(\xv)$ for all samples $\xv \in \Tilde{P}$. The training dataset, or the training distribution, is then used to train the variational parameters $\thv$ of a parameterized probability distribution $\qth(\xv)$. If successful, the output of the algorithm is a set of optimized parameters $\vec{\theta_{\rm opt}}$ such that the trained model $q_{\vec{\theta_{\rm opt}}}(\xv)$ well-approximates the unknown target distribution $p(\xv)$. The trained model  $q_{\vec{\theta_{\rm opt}}}(\xv)$ can then be used to generate new and previously unseen data. For compactness, we use the notation $p$ and $q_{\thv}$ to denote the target and model distributions respectively.

The process of training requires a \textit{loss function} $\LC(\thv)$ which estimates the distance between the model distribution $q_{\thv} (\xv)$ and the training distribution $\pt(\xv)$. 
For typical choices in loss function (detailed further in Section~\ref{sec:losses}), the loss is minimised when the model parameters $\thv$ are tuned such that the model distribution perfectly matches the empirical distribution obtained from the training data. That is, $\LC(\thv)=0$ if and only if $\qth(\xv) = \pt(\xv)$ over the entire data space $\XC$. 
Thus, by perfectly minimizing the loss, one perfectly learns the empirical distribution $\tilde{p}(\xv)$ but not the true target distribution $p(\xv)$. This scenario is commonly called \textit{overfitting}\footnote{In contrast, discriminative machine learning models can be perfectly minimized on the training data and not be overfitted.}.
To allow for \textit{generalization}~\cite{gili2022evaluating}, whereby the model can generate novel data with similar properties to the training data, one seeks to significantly reduce (but not perfectly minimize) the training loss. While generalization is the end-all goal of generative models, it is not the focus of this work. Instead, we focus on the training component of the generative framework, as failing to train also prohibits generalization.

\subsubsection{Quantum circuit models}

One prototypical quantum generative model is the \textit{quantum circuit Born machine} (QCBM)~\cite{cheng2018information,benedetti2019generative,Born_machine_Liu,coyle2020born}. Owing its name to the Born rule of quantum mechanics, a QCBM encodes a probability distribution over discrete data (here bitstrings) in an $n$-qubit pure quantum state that depends on a parameterized unitary $U(\thv)$,
\begin{align}
    \qth(\xv)  & = |\langle \xv| U(\thv)|\vec{0}\rangle|^2  \label{eq:qcbm_distribution} \, . 
\end{align}
Here $\ket{\xv}$ is a computational basis state corresponding to a bitstring $\xv$ and, without loss of generality, an initial state can be chosen as $\ket{\vec{0}} = \ket{0}^{\otimes n}$. We note that estimating $\qth(\xv)$ is equivalent to finding the expectation value of a global projector $|\xv \rangle \langle \xv |$.
More fundamentally, QCBMs enable the encoded distribution to be efficiently sampled simply by measuring in a chosen computational basis. That is, every measurement of the quantum state provides an unbiased sample from the encoded distribution (in an ideal noise-free setting). This is a very desirable property in generative models that many (classical) generative models do not share with the QCBM. Sampling techniques for classical generative models are often unreliable and may break down for certain distributions, as is the case for \textit{restricted Boltzmann machines} (RBMs)~\cite{RBM,hinton2012practical}. Born machines represent an effort to create a powerful, flexible and efficient generative model for classical discrete data, and as well as numerous `standard' digital quantum implemenations~\cite{Hamilton2018,leyton2019robust,coyle2020generativeFinance, Zhu2018,rudolph2022generation}, they have been widely implemented using tensor networks~\cite{MPS_born_machine,Cheng2019TTNBM,vieijra2022PEPSBM,TNBM}, continuous variable hardware~\cite{Born_machine_CV}, in a conditional setting~\cite{Kiss_PRA,Benedetti2021inference}, with non\hyp linearities~\cite{non-linear_QCBM}.

 \medskip 

An important, but rather subtle, distinction in generative modelling is that between \textit{explicit and implicit generative models}~\cite{mohamed2016learning,jerbi2021quantum}. Explicit generative models are ones that allow efficient access to the model probability $\qth(\xv)$ for any data sample $\xv$. Here, ``efficient'' means that the probabilities can be computed in a time and memory that are polynomial in the size of the data samples, i.e., $\OC(\poly(n))$ resources. Explicit (classical) generative models include for example auto-regressive models~\cite{Pixel_RNN}, RNNs~\cite{RNN}, tensor networks without loops (which includes tensor network Born machines)~\cite{MPS_born_machine,Cheng2019TTNBM}, and many forms of density estimators. In contrast, \textit{implicit} models lack this property and instead offer efficient access to samples from $\qth(\xv)$, which some forms of explicit models may struggle with.
A popular example of an implicit generative model are \textit{Generative Adversarial Networks} (GANs) \cite{GAN_Goodfellow} that leverage an implicit training scheme to learn powerful generators.

In the case of QCBMs implemented on quantum devices, it becomes evident that we do not have (efficient) explicit access to $q_{\thv}(\xv)$, but only to samples of the distribution in the computational basis. Consequently, QCBMs can be classified as implicit generative models.
In this work, we study the trainability issues that QCBMs suffer from as a result.

\subsubsection{Loss functions}\label{sec:losses}
Similarly to the distinction between explicit and implicit generative models, we draw a distinction between \textit{explicit and implicit loss functions}. In broad terms, explicit losses are those that can only be formulated explicitly in terms of the target and model \textit{probabilities}, whereas implicit losses are those that can be formulated in terms of an average over model and training data \textit{samples}. This distinction at the level of loss functions thus mirrors the capabilities and limitations of explicit and implicit generative models.

More concretely, we define an \textit{explicit loss} as a loss function $\LC$ that can be written solely as a function of the probabilities of the target and model distributions, without any dependence on the data itself. 
Explicit losses thus take the general form 
\begin{equation}\label{eq:explicitcost}
   \LC_{\rm expl}(\thv) := \sum_{\vec{x}_1 ... \vec{x}_{r}} f\Big(p(\xv_1), ..., p(\xv_r), \qth(\xv_{1}), ..., \qth(\xv_{r}) \Big) 
   \, ,
\end{equation}
where $f(\cdot)$ is a function that depends on the target probabilities $p(\xv_i)$ and model probabilities $ \qth(\xv_{i})$ for  data variables $\xv_i \in \XC$ with $i = 1, \, ... \, , r$. For this loss to be useful, the function $f$ should be chosen such that it measures the distance between the probability distributions $p$ and $\qth$.
Crucially, the function $f$ does not take the data values $\xv$ themselves as arguments. 

While in full generality explicit losses could compare multiple copies of the target and model probabilities (i.e., we can have $r > 1$), in practice, they usually take the simpler form
\begin{align}
    \LC(\thv) = \sum_{\xv \in \XC} f(p(\xv),\qth(\xv) ) \;. \label{eq:pairwiseexplicitloss}
\end{align}
We call such losses \textit{pairwise explicit losses} since they compare the model and target probabilities on the same data samples, or in our case, bitstrings. The pairwise explicit loss covers all so-called $f$-divergences~\cite{csiszar1967information}, including the commonly encountered KL divergence (KLD)~\cite{kullback1951KLD},
\begin{align}\label{eq:KLD-loss}
    \LC^{\rm KLD}(\thv) =  \sum_{\xv \in \XC} p(\xv) \log\left( \frac{p(\xv)}{\qth(\xv) }\right) \;,
\end{align}
the reverse-KLD,
\begin{align}
    \LC^{\rm rev-KLD}(\thv) = \sum_{\xv \in \XC} \qth(\xv) \log\left( \frac{\qth(\xv)}{p(\xv)}\right) \;,
\end{align}
the Jensen-Shannon divergence (JSD)~\cite{lin1991divergence},
\begin{align}
    \LC^{\rm JSD}(\thv) =  \sum_{\xv \in \XC}\Big[ &p(\xv) \log\left( \frac{p(\xv)}{p(\xv)+\qth(\xv) }\right) + \nonumber\\
    &\qth(\xv) \log\left( \frac{\qth(\xv)}{p(\xv)+\qth(\xv) }\right) \Big] \;,
\end{align}
and the total variation distance (TVD),
\begin{align}
\label{eq:loss_TV}
    \LC^{\rm TVD}(\thv) =  \sum_{\xv \in \XC} |p(\xv) - \qth(\xv)| \;.
\end{align}
Another example of loss function that can be written in this form is the classical fidelity,
\begin{align}\label{eq:classical_fidelity}
    \LC^{\rm CF}(\thv) = 1 -  \sum_{\xv \in \XC} \sqrt{p(\xv)\qth(\xv)} \;.
\end{align}

Notably, any non-data dependent post-processing of an explicit loss retains its explicit character. Thus, any non-data dependent function of an explicit loss (Eq.~\eqref{eq:explicitcost}) may also be considered an explicit loss. 
For example, the Rényi divergence~\cite{renyi1961measures}
\begin{align}\label{eq:renyi_divergence}
    \LC_{R,\alpha}(\thv) = \frac{1}{\alpha - 1} \log \left( \sum_{\xv} \frac{p^\alpha(\xv)}{\qth^{\alpha - 1}(\xv)}\right) \;,
\end{align}
with $0 < \alpha < \infty$ and $\alpha \neq 1$ can be classified as an explicit loss function.

\medskip

On the other hand, we define an \textit{implicit loss} as one that can be written as an average over samples drawn from the target and model distributions. That is, an implicit loss function can be expressed as
\begin{equation}\label{eq:implicitcost}
    \LC_{\rm impl}(\thv) := \mathbb{E}_{\xv_1, ..., \xv_r \sim \{p, q_{\thv}\}} \, g(\xv_1, ..., \xv_r) \, ,
\end{equation}
where $g(\xv_1, ..., \xv_r)$ is some function that depends on the data (but not probabilities), and the expectation is taken over data variables $\xv_1, ..., \xv_r$ sampled either from the data distribution $p$ or the model distribution $\qth$.

As a key example of an implicit loss, we focus on the commonly used \textit{Maximum Mean Discrepancy} (MMD)~\cite{Gretton2012mmd} loss. The MMD takes the form
\begin{align}
    \LC_{\rm MMD}(\thv)
    = \, & \Ebb_{\vec{x},\vec{y}\sim \qth }[K(\vec{x},\vec{y})] - 2  \Ebb_{\vec{x} \sim \qth,\vec{y}\sim p }[K(\vec{x},\vec{y})] \nonumber \\
    & + \Ebb_{\vec{x},\vec{y}\sim p }[K(\vec{x},\vec{y})] \, ,\label{eq:mmd-loss-implicit}
\end{align}
where $K(\xv,\yv)$ is a freely chosen kernel function. We consider the popular choice of a classical \textit{Gaussian kernel}, which is defined as
\begin{align}\label{eq:gaussian-kernel}
    K_{\sigma}(\vec{x},\vec{y})  & = e^{-\frac{\| \vec{x} - \vec{y}\|^2_2}{2\sigma}} = \prod_{i=1}^n e^{-\frac{(x_i - y_i)^2}{2\sigma}}   \;.
\end{align}
Here, $\|. \|_2$ is the 2-norm, $\sigma > 0$ is the so-called \textit{bandwidth} parameter, and $x_i, y_i$ are the values of bit $i$ in bitstring $\xv, \yv$, respectively. 
This kernel in effect provides a continuous measure of the distance between target and model bitstrings.

Interestingly, an implicit loss can always additionally be expressed in a form where it contains the target and model probabilities. 
Taking the MMD loss in Eq.~\eqref{eq:mmd-loss-implicit} as a concrete example, the loss can be re-written as
\begin{align}
    \LC_{\rm MMD}(\thv)
    =& \sum_{\vec{x},\vec{y} \in \XC} \qth(\vec{x})\qth(\vec{y}) K(\vec{x},\vec{y}) \nonumber \\
    & - 2 \sum_{\vec{x},\vec{y}  \in \XC} \qth(\vec{x})p(\vec{y}) K(\vec{x},\vec{y}) \nonumber \\
    & + \sum_{\vec{x},\vec{y} \in \XC} p(\vec{x})p(\vec{y}) K(\vec{x},\vec{y}) \;. \label{eq:mmd-loss-general}  
\end{align}
However, we stress that due to the data-dependence in the kernel $K(\xv, \yv)$, the MMD loss function can in general not be classified as an explicit loss. 

Nonetheless, this brings us to the subtle point that explicitness and implicitness are in fact not strictly mutually exclusive, i.e., one may be able to find a loss function that satisfies both Eq.~\eqref{eq:explicitcost} and Eq.~\eqref{eq:implicitcost} in specific cases.
For example, for the MMD this occurs if the kernel is chosen to be a Kronecker delta function, $K(\vec{x},\vec{y}) = \delta_{\xv\yv}$. However, such hybrid losses are very much rare edge cases, and the overwhelming majority of losses are either explicit or implicit. 
A more detailed discussion of the technical nuances of the explicit and implicit loss  distinction is provided in Supplementary Note~\ref{sec:technical_nuances}.

\subsubsection{Loss measurement strategies}

Central to the trainability of quantum generative models is the measurement strategy used to estimate the loss. Here we draw a distinction between \textit{conventional and quantum measurement strategies}. For simplicity we now restrict our discussion to implicit quantum generative models such as the QCBM.

The \textit{conventional} measurement strategy, which can be employed by both classical and quantum implicit models, starts by collecting sample data from the target and model distributions in the bases in which the data distribution is modelled, e.g., the computational basis for the case of classical data. For an implicit loss these samples can then be directly used to evaluate the loss function in Eq.~\eqref{eq:implicitcost}. For an explicit loss, this is not possible, and instead one needs to use the collected samples to recreate an empirical estimate $\qtth$ of the true model distributions $\qth$. 

More formally, as sketched in Fig.~\ref{fig:generative_schematic}, consider the set of bitstrings $\Tilde{Q}_{\thv}$ obtained after collecting $N$ samples from the model and the empirical model distribution~$\qtth(\xv)$ constructed from these samples. Then, the statistical estimate of the pairwise explicit loss function $\Tilde{\LC}(\thv)$ in Eq.~\eqref{eq:pairwiseexplicitloss} can be expressed as
\begin{align}
    \Tilde{\LC}(\thv) = \sum_{\xv \in \XC}
    f(\pt(\xv),\qtth(\xv)) \; .
\end{align} 
Crucially, since this proxy is all we have access to, the properties of this statistical estimate are what determine the trainability of an explicit loss function when evaluated via the conventional strategy. We note that zero-estimates of the model probabilities with $\qtth(\xv)=0$ are often `clipped' with a small regularization parameter $\epsilon\ll 1$ in order to avoid numerical instabilities in the loss computation.

This conventional strategy is somewhat classical in the sense that after sampling is performed on the quantum model, the post processing required to compute the cost is entirely classical. However, `more quantum' measurement strategies are also possible. In this case, a quantum circuit is used to compute functions of the probabilities, potentially more directly and/or collectively. 

For example, rather than computing the classical fidelity in Eq.~\eqref{eq:classical_fidelity} by explicitly computing the probabilities $\qth(\xv)$, one could encode the target distribution in a quantum state
$\ket{\phi} = \sum_{\xv} \sqrt{\pt (\xv)} \ket{\xv}$ and compute the quantum fidelity
\begin{align}\label{eq:quantumfid}
    \LC_{QF}(\thv) &:= 1 - |\bra{\phi} \psi(\thv)\rangle|^2 \\ &\sim 1 - \left| \sum_x  \sqrt{\tilde{p}(\xv)\qth(\xv)} \right|^2\;.
\end{align}
Up to arbitrary global phase factors (and a mod-square) this is equivalent to the classical fidelity. However, it can be computed via coherent strategies - namely a Loschmidt echo circuit~\cite{gibbs2021long, gibbs2022dynamical, caro2022outofdistribution, Volkoff2021Universal} or a SWAP test~\cite{barenco1997stabilization, garcia2013swap}. 
We note that in this case quantum generative modelling is equivalent to a state learning problem.
While this expression seemingly requires the entire training dataset to be loaded into a wavefunction, we present an approach in Sec.~\ref{sec:quantumfid} to estimate this cost using pairwise Hadamard tests.

More generally, it remains an open question if/when commonly encountered losses for generative modelling can be computed using quantum strategies and whether or not this brings any advantages\footnote{Beyond QCBMs, \textit{Quantum Generative Adversarial Networks} (QGANs)~\cite{QGAN_Loyd} trained with classical discriminators~\cite{QGAN_Zoufal,QGAN_classical,style_qgan} in effect use a conventional measurement strategy, whereas their variant with quantum discriminators~\cite{entangling_QGAN} use a quantum strategy.}. Nonetheless, we suggest that this is an interesting avenue for future research.

\subsubsection{Exponential concentration and barren plateaus}

For a quantum generative model to be trained successfully, the loss landscape must be sufficiently featured to enable a solution to be found. 
There is a growing awareness of the importance of barren plateaus, and its sister phenomenon \textit{exponential concentration}, for quantum machine learning~\cite{mcclean2018barren,arrasmith2021equivalence, larocca2021diagnosing,cerezo2020impact,arrasmith2020effect,holmes2021barren,zhao2021analyzing, thanasilp2022exponential}. 
A barren plateau (BP) is a loss landscape where the magnitudes
of gradients vanish exponentially with growing problem size~\cite{mcclean2018barren,cerezo2020cost,larocca2021diagnosing,marrero2020entanglement,patti2020entanglement,holmes2021connecting,holmes2021barren,zhao2021analyzing,wang2020noise,thanasilp2021subtleties,cerezo2020impact,arrasmith2020effect,wang2021can}. Closely related and equally problematic is exponential concentration where the loss is shown to concentrate with high probability to a single fixed value~\cite{arrasmith2021equivalence}. This, with high probability, results in poorly trained models using a polynomial number of measurement shots (regardless of the optimization method employed)~\cite{arrasmith2020effect}. More precisely, exponential concentration can be formally defined as follows. 

\begin{definition} [Exponential concentration]\label{def:exp-concentration}
Consider a quantity $X(\vec{\alpha})$ that depends on a set of variables $\vec{\alpha}$ and can be measured from a quantum computer as the expectation of some observable. $X(\vec{\alpha})$ is said to be deterministically exponentially concentrated in the number of qubits $n$ towards a certain fixed value $\mu$ if
\begin{align}
    |X(\vec{\alpha}) - \mu |\leq \beta \in O(1/b^n) \;,
\end{align}
for some $b>1$ and all $\vec{\alpha}$. Analogously, $X(\vec{\alpha})$ is probabilistically exponentially concentrated if
\begin{align} \label{eq:def-prob-concentration}
    {\rm Pr}_{\vec{\alpha}}[|X(\vec{\alpha}) - \mu| \geq \delta] \leq \frac{\beta}{\delta^2} \;\; , \; \beta \in O(1/b^n) \;,
\end{align}
for $b> 1$. That is, the probability that $X(\vec{\alpha})$ deviates from $\mu$ by a small amount $\delta$ is exponentially small for all $\vec{\alpha}$.
\end{definition}

A number of causes of exponential concentration and barren plateaus have been identified including using parameterized circuits that are too expressive~\cite{mcclean2018barren, holmes2021connecting, larocca2021diagnosing, tangpanitanon2020expressibility} or too entangling~\cite{marrero2020entanglement, sharma2020trainability, patti2020entanglement}. Hardware noise~\cite{wang2020noise, franca2020limitations, wang2021can} has also been shown to exponentially flatten the loss landscapes, which strongly hinders the potential of current noisy quantum devices. The exponential concentration can also happen due to randomness in the training dataset~\cite{thanasilp2021subtleties, leone2022practical,li2022concentration}.
In addition, there are studies on the exponential concentration in different QML models including dissipative parametrized quantum circuits~\cite{sharma2020trainability} as well as quantum kernel-based models~\cite{thanasilp2022exponential}.

Finally, the choice of loss function can also induce these phenomena. 
Thus far, loss concentration has predominantly been studied in the context of losses of the form
\begin{equation}\label{eq:VQEcost}
    C(\thv) = \Tr[ O U(\thv) \rho U(\thv)^\dagger] \;, 
\end{equation}
where $\rho$ is an $n$-qubit input state and $O$ is a Hermitian operator.
In particular, it has been shown that `global'~\cite{cerezo2020cost} losses, i.e., those where $O$ acts non-trivially on $\OC(n)$ qubits, induce loss concentration even for very shallow random circuits. Conversely, local losses where $O$ acts non-trivially on at most $log(n)$ \textit{adjacent} qubits (and more generally low-body losses where the adjacency constraint is lifted - see panel a) of Fig.~\ref{fig:weight_mmd_operator}) have been shown to enjoy trainability guarantees~\cite{cerezo2020cost,napp2022quantifying} with shallow unstructured circuits.
Furthermore, we note that how barren plateaus affect parametrized quantum circuits with a non-linear loss in the discriminative QML setting has been studied in Ref.~\cite{thanasilp2021subtleties}.

Here we study exponential concentration for generative modelling tasks on classical discrete data using implicit quantum generative models, and use our insights to establish guidelines of how best to train such models. 
Crucially, in this generative modeling context, the fixed points of the model probabilities tend to be exponentially small and the loss function contains the sum over exponentially many terms. These two together render previously used tools not directly applicable for studying the trainability of quantum generative models. 

\subsubsection{Large gradient variances are not enough}\label{sec:exactvariance} 

The presence or absence of barren plateaus is usually diagnosed by computing the variance of the loss over a given parameter distribution. Crucially this is usually computed for the \textit{exact loss}, i.e., not including the effect of shot noise. Here we argue that this approach can fail in the context of quantum generative modelling. In particular, if one computes the variance of the KLD loss then the loss variance can be non-exponentially vanishing even for very deep circuits. However, as we will argue in this section, the KLD loss is untrainable for both deep and shallow unstructured circuits if the model is implicit (i.e., only gives efficient access to samples and not to the probabilities). 

We now show that the variance of the exact KL divergence depends directly on the support of the target distribution and hence can be polynomially large. This is quantified by the following proposition which we prove in Supplementary Note~\ref{ap:KLD_exact}. 

\begin{proposition}\label{prop:exact-kl-var-main}
Consider the KLD loss as defined in Eq.~\eqref{eq:KLD-loss}. Assume access to the exact target distribution $p(\xv)$ and the model distribution $\qth(\xv)$. Then, we have
\begin{itemize}
    \item For deep (Haar random) parametrized circuit $U(\thv)$, the variance of the loss scales asymptotically ($2^n \gg 1$) as
    \begin{align}\label{eq:exact_kld_var}
        \Var_{\thv}[\LC^{\rm KLD}(\thv)] = \frac{\pi^2}{6}\sum_{\xv} p^2(\xv) \;.
    \end{align}
    \item  For a random tensor product circuit $U(\thv) = \bigotimes_{i=1}^n U_i(\theta_i)$ where $U_i(\theta_i)$ is a random single-qubit unitary, the variance of the loss scales as
    \begin{align}
        \Var_{\thv}[\LC^{\rm KLD}(\thv)] = n-\frac{\pi^2}{6}\sum_{\xv, \xv'} p(\xv) p(\xv')\norm{\xv-\xv'}_H \;, 
    \end{align}
    where $\| \cdot \|_{\rm H}$ is a Hamming distance.
\end{itemize}
\end{proposition}

\begin{figure*}
    \centering
    \includegraphics[width=2\columnwidth]{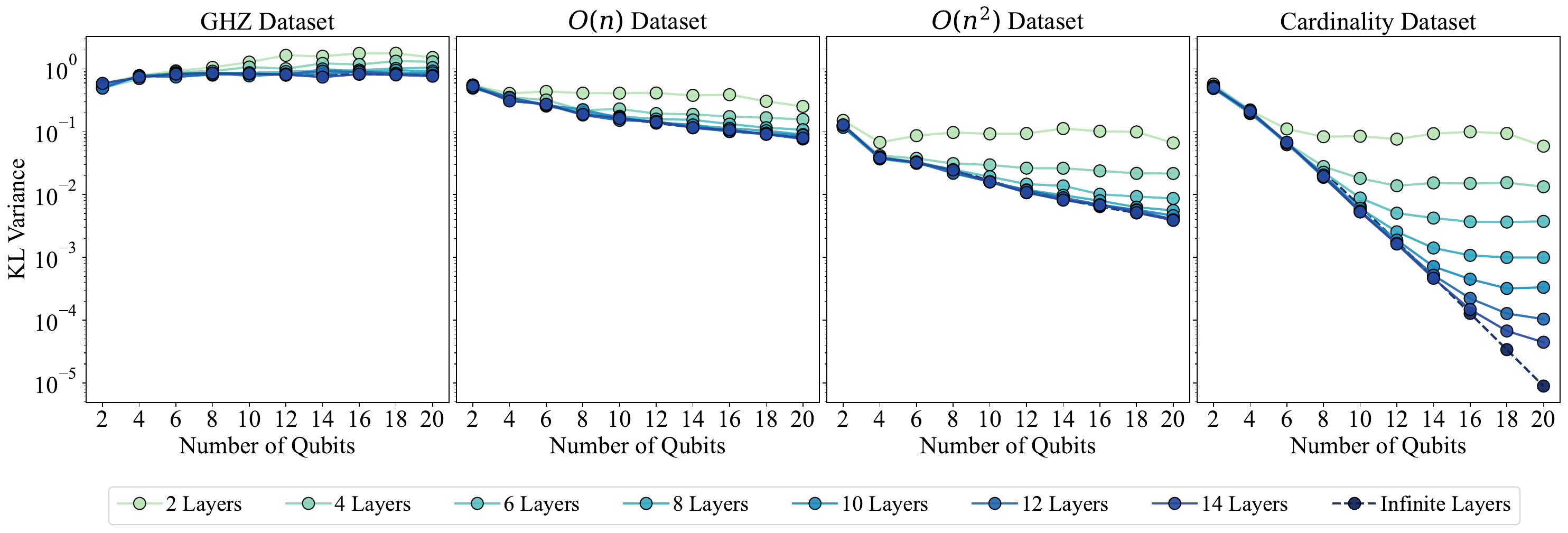}
    \caption{\textbf{Study of loss concentration with the \textit{exact} KLD loss function.} Numerical evidence that the \textit{exact} KLD loss can have a non-vanishing loss variance even when model probabilities exhibit exponential concentration. We study the loss concentration in randomly initialized line-topology circuits for various datasets, and increasing the number of qubits $n$ and circuit depth. We emphasize that the model probabilities $\qth(\xv)$ where evaluated exactly and in the absence of shot noise. We also show the infinite layer results beyond 6 qubits that are generated using Eq.~\eqref{eq:exact_kld_var}. The GHZ dataset consists of the all-0 and all-1 bitstrings ($\OC(1)$ support), the $\OC(n)$ and $\OC(n^2)$ datasets consist of $n$ and $n^2$ random bitstrings, respectively, and the cardinality dataset contains all bitstrings with $\frac{n}{2}$ cardinality ($\OC(2^n)$ support). There appears to be a strong data-dependence for the magnitude of the loss variance, which could lead to exponential concentration.}
    \label{fig:var-kld-exact}
\end{figure*}

It follows that the variance of the exact KLD can be non-exponentially vanishing even for a deep circuit, where one would generally expect a barren plateau~\cite{mcclean2018barren}, if the purity $\sum_{\xv} p^2(\xv)$ of the target distribution is non-exponentially vanishing. For this condition to be met, all we need is that at least one probability $p(\vec{x})$ of the target distribution is non-exponentially vanishing. This is captured by the following corollary. 

\begin{corollary}
    Under the same assumption as in Proposition~\ref{prop:exact-kl-var-main}, for the target distribution, at least one probability is at least polynomially large. Then, the variance of the KLD loss function does not vanish exponentially with the system's size. That is, $\exists \xv: \; p(\vec{x}) \in \Omega\left(\frac{1}{\poly(n)}\right)$, we have
    \begin{align}
        \Var_{\thv}[\LC^{\rm KLD}(\thv)] \notin \OC\left( \frac{1}{b^n} \right) \;,
    \end{align}
    for some constant $b > 1$. 
\end{corollary}

We note that any distribution with support on at most $D$ bit strings necessarily has at least one probability that is $1/D$ large. Thus the support of a distribution lower bounds the variance of the exact KLD. 
This is reflected in Fig~\ref{fig:var-kld-exact}. For the GHZ dataset, which has $\OC(1)$ support, we observe a strong evidence for non-vanishing variance for all circuit depths. For linear and quadratic support datasets, the variances moderately decrease as the number of qubits increases for deep circuits.

Thus we see that for certain target probability distributions, the KLD does not exhibit a barren plateau for \textit{explicit} models. This suggests that quantum-inspired models that can provide direct access to probabilities (e.g. tensor network Born machines~\cite{MPS_born_machine,Cheng2019TTNBM}) might be trainable with the KLD. However, current generative models running on quantum devices only provide access to samples from a distribution via measurements and, as we will argue in the next section, the large variance of the exact loss, in contrast to standard VQE-style losses, does not translate to substantial loss gradients in practise.

\subsection{Trainability analysis on loss functions}
In this section, we analyse the trainability of different loss functions used in quantum generative modelling.  

\subsubsection{Pairwise explicit losses}\label{sec:explictlosstrain}

\begin{figure}[t]
\includegraphics[width=0.99\columnwidth]{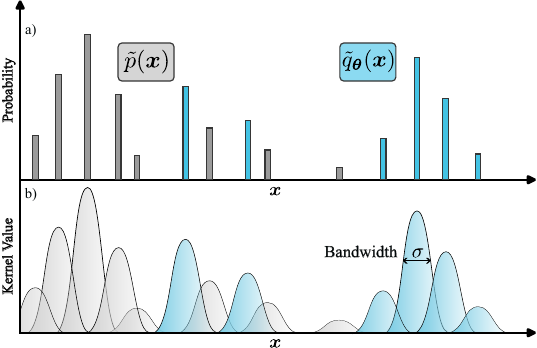}
\caption{\textbf{The problem with pairwise explicit losses.} In the space with $2^n$ unique  $n$-bit bitstrings, samples $\xv$ generated from an uninformed model with high probability do not coincide with any of the training bitstrings. In other words, the empirical model distribution $\qtth(\xv)$ and the training distribution $\pt(\xv)$ do not both have non-zero probabilities for any bitstring $\xv$. On the other hand, an implicit loss function such as the MMD provides a continuous measure of distance between the distributions by use of a Gaussian kernel with bandwidth $\sigma$.} 
\label{fig:nooverlap}
\end{figure}

Part of the power of \textit{quantum} generative models is that they can be used to continuously parameterise and express distributions over discrete data with exponential support. That is, an $n$-qubit model can be used to model distributions over $2^n$ different $n$-bitstrings. However, while the true target distribution may have exponential support, the amount of training data $\Tilde{P}$ is in practise restricted. More precisely, for large $n$ (e.g., $n>50$), it is reasonable to assume that the number of bitstrings in the training dataset scales at most polynomially in $n$. 
Similarly, the number of bitstrings samples obtained from the model must also scale at most polynomially in $n$. That is, $|\Tilde{P}|, |\Tilde{Q}_{\thv}| \in \OC(\poly(n))$. 

This discrepancy between the polynomial support of the training data and the exponential support of the model, can make it highly challenging to train implicit models using pairwise explicit loss functions. In loose terms, the problem is that the only bitstrings that contribute to the evaluation of a statistical estimate of an explicit cost are those corresponding to bitstrings $\tilde{P}$ in the training data. To estimate the loss one thus needs good estimates of the model distributions over the support of $\tilde{P}$. However, for an implicit model these estimates are obtained via sampling and the set $\tilde{P}$ contains an exponentially small proportion of the total number of bitstrings. As such, \textit{for generic models} (i.e., those using no information about the particular dataset at hand), the probability of measuring any bitstring in the training set will also be exponentially small (as sketched in Fig.~\ref{fig:nooverlap}), leading to a poor statistical estimate of the loss. This observation was in fact one of the original motivations for moving away from the KLD and introducing the MMD loss in a quantum context in Ref.~\cite{liu2018differentiable} or follow-up works such as Ref.~\cite{coyle2020born}. 

\subsubsection{Concentration of Pairwise Explicit Losses}

To make this line of argument more concrete, the first family of models we will consider are those where the individual model probabilities $\qth(\xv)$ are exponentially concentrated over different values of $\thv$. This is the case for a large family of unstructured parameterised quantum circuits.
Since estimating $\qth(\xv)$ is equivalent to computing the expectation value of the global projector $|\xv \rangle \langle \xv |$, the concentration of $\qth(\xv)$ can be viewed as resulting from the global-measurement induced barren plateau phenomenon~\cite{cerezo2020cost}. In this case, concentration is observed even for an ansatz that is comprised of only a single layer of single-qubit rotations. However, alternative phenomena (e.g. noise~\cite{wang2020noise} or expressibility~\cite{holmes2021connecting}) can also lead to the exponential concentration of $\qth(\xv)$. More formally, the following proposition holds. 

\begin{proposition}[Concentration of model]\label{prop:circuit-example}
    For all possible bitstrings $\xv \in \XC$, the underlying probability $\qth(\xv)$ of the quantum model exponentially concentrates towards some exponentially small fixed point $\mu \in O(1/b^n)$ for $b>1$ if the quantum generative model is constructed with:
    \begin{itemize}
        \item A single layer of random single qubit gates $U(\thv) = \bigotimes_{i = 1}^n U_i(\thv_i)$. Or, more precisely, if $\{U_i(\thv_i)\}_{\thv_i}$ forms a local 2-design on qubit $i$~\cite{cerezo2020cost}.
        \item $L$ layers of random $k$-local 2-designs, i.e., $U(\thv) = \prod_{l=1}^L\bigotimes_{j = 1}^{n/k} U_{l,j}(\thv_{l,j})$ with each $U_{l,j}(\thv_{l,j})$ acting on $k$ qubits and $\{ U_{l,j}(\thv_{l,j})\}_{\thv_{l,j}}$ forming a $k$-local 2-design over $\thv_{l,j}$~\cite{cerezo2020cost}.
        \item A parameterised quantum circuit $U(\thv)$ such that its ensemble over $\thv$ i.e., $\{ U(\thv)\}_{\thv}$ forms an approximate 2-design on $n$ qubits~\cite{mcclean2018barren, holmes2021connecting}. This holds even for the problem-inspired circuits~\cite{larocca2021diagnosing}.
        \item A linear-depth quantum circuit subject to local Pauli noise between each layer~\cite{wang2020noise}.
    \end{itemize}
\end{proposition}

Proposition~\ref{prop:circuit-example} provides examples of cases where the model probabilities exponentially concentrate over \textit{all} bitstrings in $\XC$. However, we find that in fact trainability difficulties arise even if model probabilities are only exponentially concentrated over the training dataset (but perhaps not on points outside the dataset). That is, all that is required for untrainability is that the probability of measuring a sample that is also in the dataset is practically zero. 
This is likely to be the case even for highly structured quantum circuits if the generative model is built without a strong inductive bias. 
We formalise this intuition in Supplementary Note~\ref{app:pairwise-no-overlap}.

We now argue that the exponential concentration of probabilities $\qth(\xv)$ over the dataset causes $\Tilde{\LC}(\thv)$ to also exponentially concentrate. To understand why, let us look at the probability of measuring one specific bitstring (e.g., $\xv_{0}$ - the all-zero bitstring) and assume that $\qth(\xv_0)$ is exponentially concentrated towards some exponentially small value $\mu$. Then, for any given parameter constellation, it is highly likely that $\qth(\xv_0)$ is exponentially close to $\mu$. To estimate $\qth(\xv_0)$ on a quantum computer we sample $N$ bitstrings from the quantum model and record the observations. The chance that none of the sampled bitstrings are the specific bitstring that we are interested in is $(1 - \qth(\xv_0))^N \approx 1 - N \mu $. However, the number of circuits $N$ that can be efficiently run is necessarily limited - here we will assume $N \in \text{poly}(n)$. Thus we have that the probability of not measuring the bitstring we are interested in is exponentially close to 1. That is, the statistical estimate of $\qtth(\xv_0)$ is almost always zero.
We can then generalize this intuition for a single bitstring to the estimation of each of the (polynomially many) target bitstrings and therefore the whole loss function. The following theorem formalizes this argument.

\begin{theorem}[Concentration of pairwise explicit loss for concentrated models]\label{thm:explicit-loss}
Consider the loss function of the form in Eq.~\eqref{eq:pairwiseexplicitloss}. 
Assume that for all bitstrings in the training dataset, $\xv \in \Tilde{P}$, the quantum generative model $\qth(\xv)$ exponentially concentrates towards some exponentially small value (as defined in Definition~\ref{def:exp-concentration}). Suppose that $N \in \OC(\poly(n))$ samples are collected from the quantum model corresponding to the set of sampled bitstrings $\Tilde{Q}_{\thv}$, and that the training dataset $\tilde{P}$ contains $M \in \OC(\poly(n))$ samples.
We define the fixed point of the loss as
\begin{align}\label{eq:fixedpoint}
    \LC_0(\Tilde{P}, \Tilde{Q}_{\thv}) = \sum_{\xv \in \PC} f(\pt(\xv), 0) + \sum_{\xv \in \QC_{\thv}} f(0, \qtth(\xv)) \;,
\end{align}
with $\PC$ (and $\QC_{\thv}$) being a set of \textit{unique} bitstrings in $\Tilde{P}$ (and $\Tilde{Q}_{\thv}$).
Then, the probability that the estimated value $\Tilde{\LC}(\thv)$ is equal to $\LC_0(\Tilde{P}, \Tilde{Q}_{\thv})$ is exponentially close to 1, i.e.,
\begin{align}
    {\rm Pr}_{\Tilde{Q}_{\thv},\thv}[\Tilde{\LC}(\thv) = \LC_0(\Tilde{P}, \Tilde{Q}_{\thv})] \geq 1 - \delta \;,
\end{align}
with $\delta \in \OC\left(\frac{\poly(n)}{c^n}\right)$ for some $c > 1$. 
\end{theorem} 

As a direct consequence of Theorem~\ref{thm:explicit-loss}, the following corollary gives the concentration points of some specific explicit loss functions mentioned in this work.
\begin{corollary}[Concentration points of common explicit loss functions]\label{corol:fixedpoints}
Under the same conditions as in Theorem~\ref{thm:explicit-loss}, the following loss functions concentrate at
\begin{itemize}
    \item KL-divergence: 
    \begin{align}\label{eq:KLfixedvalue}
        \LC_0^{{\rm KLD}}(\Tilde{P},\Tilde{Q}_{\thv}) = \sum_{\xv \in \PC} \pt(\xv) \log \left( \frac{\pt(\xv)}{\epsilon}\right) \;. 
    \end{align}
    Here $\epsilon \ll 1$ is a clipping value, which is common practice to avoid the singularity of the logarithm at $\qth(\xv) = 0$.
    \item Classical fidelity:
    \begin{align}
        \LC_0^{\rm CF}(\Tilde{P},\Tilde{Q}_{\thv}) = 1 \;.
    \end{align}
    \item Reverse KL-divergence:
    \begin{align}
        \LC_0^{\rm rev-KLD}(\Tilde{P},\Tilde{Q}_{\thv}) = \sum_{\xv \in \QC_{\thv}}\qtth(\xv) \log\left( \frac{\qtth(\xv)}{\epsilon}\right) \;.
    \end{align}
    \item Total variation distance:
    \begin{align}
        \LC_0^{\rm TVD}(\Tilde{P},\Tilde{Q}_{\thv}) =  2 \;.
    \end{align}
\end{itemize}
\end{corollary}

Looking at the expressions for the fixed points given above, in the case of the KL divergence, classical fidelity and total variational distance, the fixed point is independent of $\thv$. Thus it is clear that the costs cannot be used to train the quantum circuit model. In the case of the reverse KL divergence, the fixed point depends on $\thv$ but is independent of the training data and thus the reverse KL also cannot be used to train the model to learn the target distribution. 

More generally, for all explicit losses of the form Eq.~\eqref{eq:pairwiseexplicitloss}, the concentration point $\LC_0(\Tilde{P}, \Tilde{Q}_{\thv})$, Eq.~\eqref{eq:fixedpoint}, can be separated into two terms: (i) the term that involves only $\Tilde{P}$ and (ii) the other that involves only $\Tilde{Q}_{\thv}$. In other words, the $\thv$ dependence of the estimator of the loss is independent of the target distribution and thus the estimate of the loss is worthless for training the generative model. 
This no-go result is rigorously established in Corollary~\ref{coro:untrain}. 
Our approach is to show that the loss function at two arbitrary parameter values $\thv_1$ and $\thv_2$, contains no information about the training distribution.

\begin{figure}[t]
\includegraphics[width=0.99\linewidth]{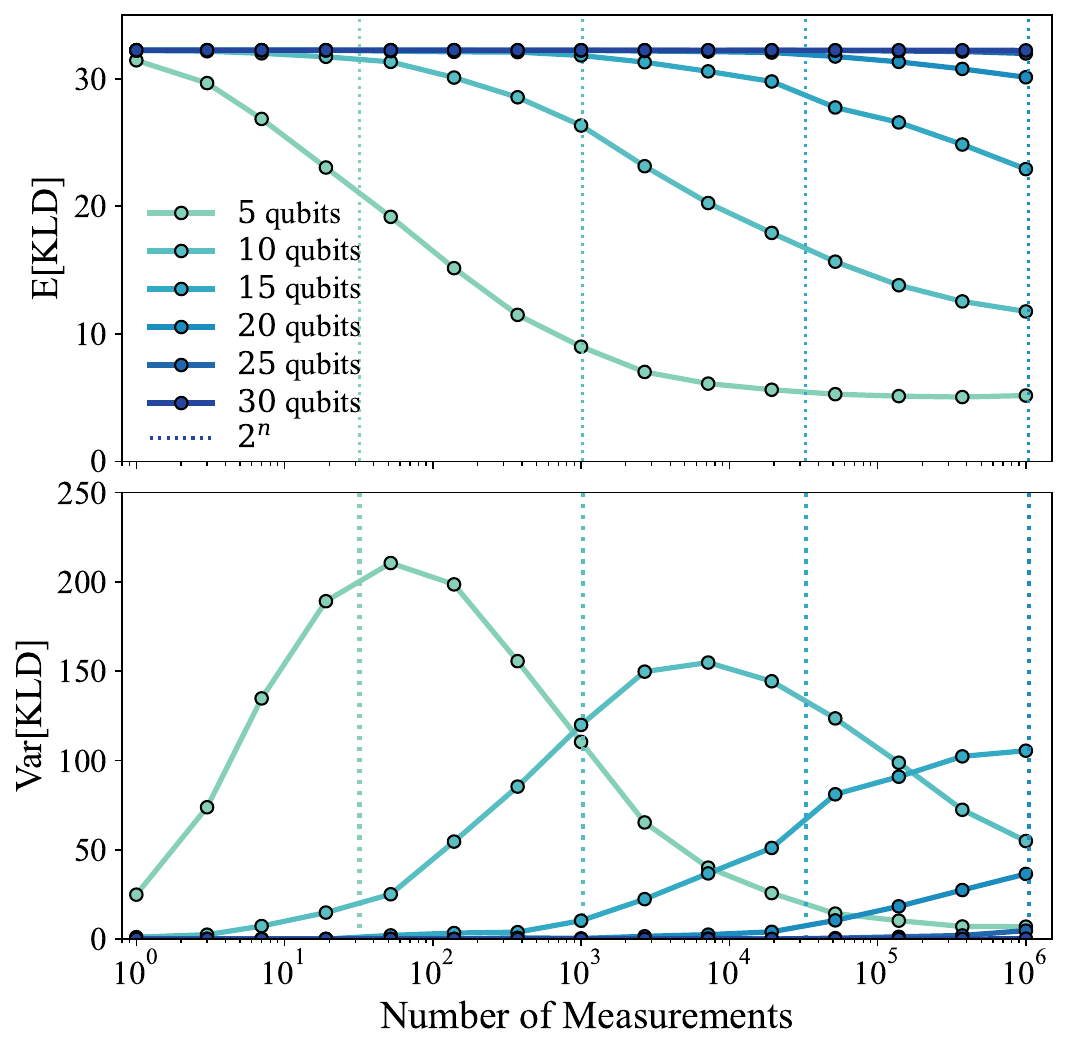}
\caption{\textbf{Variance of the KL divergence with finite shots.} Concentration of the KL divergence loss as a function of the number of measurements and qubits for random product state circuits. Here we take the target distribution to be $\pt(\vec{0}) = 1$ and take the cutoff of the KLD to be $\epsilon = 10^{-14}$. Vertical lines indicate where the number of measurements equal $2^n$. Thus, we see that the the KLD estimate is biased upwards with any finite number of measurements, and the number of measurements required to achieve a reasonable level of uncertainty increases exponentially with the number of qubits $n$.} \label{fig:general-no-go}
\end{figure}

\begin{corollary}[Untrainability of pairwise explicit loss functions]\label{coro:untrain}
Under the same conditions as in Theorem~\ref{thm:explicit-loss}, the probability that the difference between the two statistical estimates of the loss function at $\thv_1$ and $\thv_2$ does not contain any information about the training distribution is exponentially close to 1. Particularly, we have
\begin{align}
    {\rm Pr}_{\Tilde{Q}_{\thv},\thv}[\Tilde{\LC}(\thv_1) - \Tilde{\LC}(\thv_2) = \Delta\LC_0(\Tilde{Q}_{\thv_1},\Tilde{Q}_{\thv_2})]\geq 1 - 2\delta \; ,
\end{align}
with $\delta \in\OC\left(\frac{\poly(n)}{c^n}\right)$ for some $c > 1$, $\Tilde{Q}_{\thv_1}$ (and $\Tilde{Q}_{\thv_2}$) is a set of sampling bitstrings obtained from the quantum generative model at the parameter value $\thv_1$ (and $\thv_2$), as well as
\begin{align}
    \Delta \LC_0(\Tilde{Q}_{\thv_1},\Tilde{Q}_{\thv_2}) = \sum_{\xv \in \QC_{\thv_1}} f(0,\Tilde{q}_{\thv_1}(\xv)) - \sum_{\xv \in \QC_{\thv_2}} f(0,\Tilde{q}_{\thv_2}(\xv)) \;,
\end{align}
with $\QC_{\thv_1}$ (and $\QC_{\thv_2}$) being a set of unique bit-strings in $\Tilde{Q}_{\thv_1}$ (and $\Tilde{Q}_{\thv_2}$). Crucially, $\Delta \LC_0(\Tilde{Q}_{\thv_1},\Tilde{Q}_{\thv_2})$ does not depend on any $\pt(\xv) \in \Tilde{P}$.
\end{corollary}

To support our analytic claims we further conducted a numerical study of the exponential concentration of pairwise explicit costs. For concreteness, we here decided to focus on the KL divergence. 
In Fig.~\ref{fig:general-no-go}, we plot the mean and variance (over $\thv$) of the KL divergence for the target distribution $\pt(\vec{0}) = 1$ as a function of the number of measurement shots and qubits. For simplicity we take our model to be a (Haar) random product state. 

We see in Fig.~\ref{fig:general-no-go}a) that with a polynomial number of measurements, as per Eq.~\eqref{eq:KLfixedvalue}, the empirical estimate of the loss concentrates at $\log(1/\epsilon) \approx 32.2$ for $\epsilon=10^{-14}$. Correspondingly, with a polynomial number of measurements the variance in Fig.~\ref{fig:general-no-go}b) is exponentially close to zero. Using an exponential number of measurements, the estimate of the KL tends towards its true value and the variance is again small. The transition between these two regimes is marked by a very high variance corresponding to the case where the measurement count is high enough for there to be some overlap between the sampled bits strings and the $\vec{0}$ bitstring, but not enough overlap to obtain a reliable estimate of $\qth(\vec{0})$.
This results in the loss estimate to sporadically fluctuate between $\log(1/\epsilon)$ and $\log(1/\qth(\xv))$ with $\qth(\xv)>0$. 
While in Fig.~\ref{fig:general-no-go} the target dataset consists of a single bitstring, larger datasets only shift the curves to the left by a polynomial amount.
\medskip
\paragraph*{Broader Implications.}

While our results above are formulated for training QCBMs with pairwise explicit costs, we argue that the underlying problem is more general and immune to simple solutions.
One approach, for example, might be to take non-data-dependent functions of pairwise explicit losses, as in the case of the Rényi-divergence in Eq.~\eqref{eq:renyi_divergence}. However, such loss functions exponentially concentrate in the same manner as the explicit losses themselves when employing the conventional measurement strategy. A more promising but challenging approach would be to attempt to measure such losses via quantum strategies. We discuss this further in Section~\ref{sec:quantumfid}.

More generally, while we provide strict no-go results only for pairwise explicit losses, we believe that any explicit losses in the general form of Eq.~\eqref{eq:explicitcost} will suffer from concentration or exponential imprecision due to the inherent inability of implicit models to accurately estimate the model probabilities in polynomial time\footnote{A possible exception is if a particular implicit model instead allows for efficient estimation of gradients of an explicit loss function, as it is the case for RBMs training on the KL divergence loss function}. We are however not aware of any practical explicit loss function that cannot be brought into the pairwise explicit form.

We further stress that our results hold for unstructured ans\"{a}tze or ans\"{a}tze that lack an appropriate inductive initial bias. Thus, while explicit losses such as the KLD will not work at scale with implicit models straight out-of-the-box, our no-go theorems could be side-stepped using clever initialization strategies in conjunction with specialized ans\"{a}tze. For example, while we argue in Supplementary Note~\ref{app:pairwise-no-overlap} that initializing the quantum circuit model on a subset of training states will not alleviate the fundamental issue when using a generic ansatz, this may work if one leverages a quantum circuit that constrains the model to the symmetry sector of the data. Among other hard constraints, this is conceivable if the data consists only of samples with a certain hamming weight or cardinality, as it can be the case in certain financial applications~\cite{chang2000heuristics,alcazar2021geo}.
However, many real world datasets may not contain strong symmetries that one can leverage so straight-forwardly.
It is therefore critically important to study the effect of strong parameter initializations and inductive biases using explicit losses-- both theoretically and experimentally.

\subsubsection{Implicit losses: Maximum Mean Discrepancy}\label{sec:mmd}
In the previous section we saw that an explicit loss function, used in conjunction with an implicit generative model and the conventional sampling strategy, exhibits exponential concentration and hence is untrainable. The root cause was, at least in part, a miss-match between using an explicit loss function with an implicit model. Thus it is natural to ask whether an implicit quantum loss would fare better. 

Here we focus on analysing the MMD loss function (see Eqs.~\eqref{eq:mmd-loss-implicit} and~\eqref{eq:mmd-loss-general}), which is a commonly-used implicit loss.
In contrast to the pairwise explicit losses discussed previously, each bitstring drawn from the model is generally compared with all training bitstrings, with the
kernel function $K(\xv,\vec{y})$  controlling the contribution of each comparison. 
With a poor choice in kernel it is clear that the MMD will be susceptible to exponential concentration. For example, the Gaussian kernel with the bandwidth $\sigma \rightarrow 0$ is equivalent to a delta function kernel, $K(\xv,\vec{y}) = \langle \xv, \yv\rangle = \delta_{\xv\yv}$. In this case the MMD reduces to the pairwise explicit loss $ \sum_{\xv \in \XC} (p(\xv) - \qth(\xv))^2$ (see Supplementary Note~\ref{sec:technical_nuances} for details), and consequently is subject to our no-go result in Theorem~\ref{thm:explicit-loss}. This thus prompts the question of how exactly $\sigma$ affects trainability.

\medskip
\paragraph*{Properties of the MMD loss.}

To study the properties of the MMD loss, it is helpful to note that each term in the MMD can be viewed as the expectation value of an observable whose properties depend on the choice of $\sigma$. This change in perspective allows us to leverage existing knowledge from the VQA trainability literature. In particular, prior no-go results on VQAs with observable-type loss functions are now directly applicable here, including those on cost function induced~\cite{cerezo2020cost}, expressiblity-induced~\cite{mcclean2018barren, holmes2021connecting}, and noise-induced~\cite{wang2020noise} barren plateaus.

Specifically, each term in the MMD can be written as 
\begin{align}\label{eq:mmd-expectation}
    \MC(\rho, \rho') = \Tr[  O^{(\sigma)}_{\rm MMD} (\rho \otimes \rho')] \;, \;
\end{align}
where we have defined the MMD observable 
\begin{equation}
    O^{(\sigma)}_{\rm MMD} := \sum_{\xv,\vec{y}}  K_{\sigma}(\xv,\yv) |\xv \rangle \langle \xv | \otimes |\yv \rangle \langle \yv |\;.
\end{equation}
This observable acts on $2n$ qubits, namely $n$ qubits corresponding to the QCBM, $\rho_{\thv} = \ket{\psi(\thv)} \bra{\psi(\thv)}$, and $n$ qubits corresponding to the dataset, $\rho_{\pt} = \sum_{\yv} \pt(\yv) |\yv \rangle \langle \yv |$. 
For the first term in the MMD, both $\xv$ and $\yv$ are sampled from the QCBM and we have $\rho = \rho' = \rho_{\thv}$.
The cross-term instead has $\rho = \rho_{\thv}$ and $\rho' = \rho_{\pt}$, and the final term has $\rho = \rho' = \rho_{\pt}$. 

\begin{figure*}
    \centering
    \includegraphics[width=0.95\linewidth]{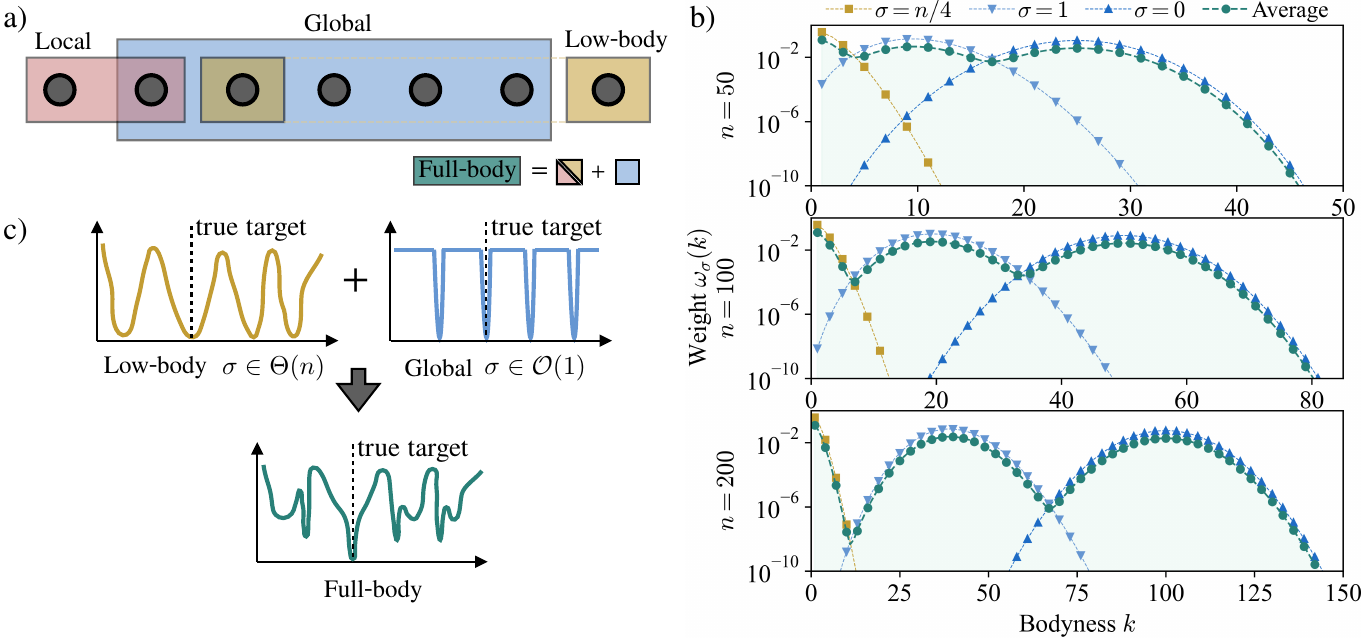}
    \caption{\textbf{Bodyness of the MMD loss.} a) We illustrate the difference between `low-body', `local', `global' and `full-body' operators. An operator $O$ is low-bodied if it acts non-trivially on at most $\OC(\log(n))$ qubits. If a low-bodied operator acts on qubits that are adjacent to each other, then $O$ is said to be local. On the other hand, $O$ is global if it acts non-trivially on $\Theta(n)$ qubits. Lastly, a full-body operator consists of the sum of several operators that are low-body/local and global. b) We depict the weight $w_\sigma(k)$ for the terms in the MMD operator as a function of their bodyness $k$ for $n=50$, $100$ and $200$ qubits and the bandwidths $\sigma = 0$, $ 1$ and $n/4$. The average weight over these three $\sigma$ values is also shown. For small $\sigma$, the MMD operator is a sum of predominantly global operators, i.e., with $\sigma \in \OC(1)$ the mean bodyness is $\Theta(n)$. In contrast, $\sigma\in\Theta(n)$ results in predominantly low-bodied operators. c) Sketch of the expected landscapes for low-body, global and full-body losses respectively. Because low-body and global operators are exclusively sensitive to low-body and global features, respectively, their loss landscapes exhibit spurious minima, which don't coincide with the minimum of the true target distribution. A full-body loss on the other hand should have a single optimal solution solution where all its constituent operator's minima align.}
    \label{fig:weight_mmd_operator}
\end{figure*}

In the Pauli basis, the MMD observable $O^{(\sigma)}_{\rm MMD}$ takes the elegant form 
\begin{align}\label{eq:mmd-observable}
    O^{(\sigma)}_{\rm MMD} = \sum_{l=0}^n w_{\sigma}(l) D_{2l} \; ,
\end{align}
where $D_{2l}$ are normalized $2l$-body diagonal operators (defined explicitly in Supplementary Note~\ref{app:mmd-k-body-observable}), and
\begin{equation}\label{eq:mmd_weight}
    w_{\sigma}(l) = {n \choose l} (1-p_\sigma)^{n-l}p_\sigma^l
\end{equation}
are Bernoulli-distributed weights with effective \textit{probability}
\begin{equation}
    p_\sigma = (1 - e^{-1/2\sigma})/2\;.
\end{equation}
Thus estimating the MMD loss function in Eq.~\eqref{eq:mmd-loss-implicit} using a batch of measurements $\Tilde{Q}$ is equivalent to using the same measurements to estimate a weighted expectation of the observables $D_{2l}$.

The properties of the MMD observable clearly depend on the distribution of the terms of different bodyness through the $w_\sigma(l)$ factor.
Fig.~\ref{fig:weight_mmd_operator} shows how $w_\sigma(l)$ are distributed for different $\sigma$. Owing to the Bernoulli-distributed weights, we can straight-forwardly provide the average bodyness of $O^{(\sigma)}_{\rm MMD}$, which is given by 
\begin{align}\label{eq:mean_bodyness}
    \Ebb_{l \sim \omega_\sigma(l)}[2l] = 2np_\sigma \;,
\end{align}
and the variance in the bodyness, which is 
\begin{align}\label{eq:var_bodyness}
    \Var_{l \sim \omega_\sigma(l)}[2l] = 4np_\sigma(1-p_\sigma) \;.
\end{align}
From these expressions it follows that the MMD loss is predominantly composed of global operators when $\sigma\in \OC(1)$. More concretely the following proposition holds. 
\begin{proposition}[MMD consists largely of global terms for $\sigma\in \OC(1)$]\label{prop:mmd-global}
For $\sigma\in \OC(1)$, the average bodyness of the MMD operator containing Pauli terms with weight $w_\sigma(l)$ is
\begin{equation}
     \Ebb_{l \sim w_\sigma(l)}[2l] \in \Theta(n) \, .
\end{equation}
Similarly, the variance in the bodyness is given by 
\begin{equation}
   \Var_{l \sim w_\sigma(l)}[2l] \in \Theta(n) \, .
\end{equation}
\end{proposition}
This shows that with fixed-size bandwidths $\sigma$, as is commonly done (e.g., Ref.~\cite{Born_machine_Liu}), the MMD suffers from global loss function-induced barren plateaus~\cite{cerezo2020cost} and hence is untrainable. This practice of using constant bandwidths is carried over from classical ML literature~\cite{li2017mmd,wang2018improving,Li2018domain}, but Proposition~\ref{prop:mmd-global} shows that this is fundamentally incompatible with quantum generative models using unstructured circuits.

In contrast, we show that if the bandwidth scales linearly in the number of qubits, $\sigma\in\Theta(n)$, the MMD loss function is approximately low-bodied.
We recall that being low-bodied is more general than being \textit{local}, the latter corresponding to the case where an operator is low-bodied and each term only acts non-trivially on adjacent qubits. 
The following proposition formalizes this relation by quantifying the error made when truncating the MMD observable after a certain bodyness. 
\begin{proposition}[MMD consists largely of low-body terms for $\sigma\in\Theta(n)$]\label{prop:mmd-k-body}
Let $\Tilde{\LC}^{(\sigma, k)}_{\rm MMD}(\thv)$ be a truncated MMD loss with a truncated operator
$\tilde{O}^{(\sigma, k)}_{\rm MMD}$ that contains up to the $2k$-body interactions in $O^{(\sigma)}_{\rm MMD}$,
\begin{align} \label{eq:MMD_truncated_operator}
     \tilde{O}^{(\sigma,k)}_{\rm MMD} :=  \sum_{l=0}^k w_{\sigma}(l) D_{2l} \; ,
\end{align}
where $w_{\sigma}(l)$ are Bernoulli-distributed weights defined in Eq.~\eqref{eq:mmd_weight}. For $\sigma \in \Theta(n)$, the difference between the exact and local approximation of the loss is bounded as
\begin{align}
    |\LC^{(\sigma)}_{\rm MMD}(\thv) - \Tilde{\LC}^{(\sigma, k)}_{\rm MMD}(\thv)| \leq  \epsilon(k) \;,
\end{align} 
with
\begin{align}
    \epsilon(k) \in \OC\left( n (c/k)^k\right) \;,
\end{align}
for some positive constant $c$. 
\end{proposition}
This implies that one can view the MMD loss with a bandwidth $\sigma\in\Theta(n)$ as composed almost exclusively of low-body contributions. We therefore expect, given the results of Refs.~\cite{cerezo2020cost, napp2022quantifying}, that the the MMD is trainable for $\sigma \in \Theta(n)$  for quantum generative models which employ shallow quantum circuits. We note that there appears to be no merit in increasing $\sigma$ beyond $\Theta(n)$, as that simply increases the relative weight of the constant $l=0$ term in Eq.~\eqref{eq:mmd-observable}. That is, the MMD operator tends towards the trivial identity measurement for $\sigma \rightarrow \infty$. 

\medskip

To probe this further, and get a better understanding of the effect of $\sigma$ on the trainability of the MMD loss, we start by considering the case of QCBM with a product ansatz. This allows us to find a closed-form expression of the MMD variance as a function of the circuit parameters (Supplemental Proposition~\ref{sup-prop:mmd-var-tensor}) from which we can study the concentration of the MMD for different $\sigma$ values. Our findings are summarized by the following Theorem (proven in Supplementary Note~\ref{app:mmd_var_theorem}).

\begin{theorem}[Product ansatz trainability of MMD, informal]\label{thm:mmd-sigma}
Consider the MMD loss function $\LC^{(\sigma)}_{\rm MMD}(\thv)$ as defined in Eq.~\eqref{eq:mmd-loss-implicit}, which uses the classical Gaussian kernel as defined in Eq.~\eqref{eq:gaussian-kernel} with the bandwidth $\sigma>0$, and a quantum circuit generative model that is comprised of a tensor-product ansatz $U = \bigotimes_i^n U_i(\theta_i)$ with $\{ U_i(\theta_i) \}_{\theta_i}$ being single-qubit (Haar) random unitaries. Given a training dataset $\Tilde{P}$, the asymptotic scaling of the variance of the MMD loss depends on the value of $\sigma$.

\medskip
    \noindent For $\sigma \in \OC(1)$, we have
    \begin{align}
        \Var_{\thv}[\LC^{(\sigma)}_{\rm MMD}(\thv)] \in \OC(1/b^n) \;,
    \end{align}
    with some $b>1$.
    
    \noindent On the other hand, for $\sigma \in \Theta(n)$, we have
    \begin{align}
        \Var_{\thv}[\LC^{(\sigma)}_{\rm MMD}(\thv)] \in  \Omega(1/n)\;.
    \end{align}
\end{theorem}

We numerically verify Theorem~\ref{thm:mmd-sigma} in Fig.~\ref{fig:mmd_var_training}. In panel a) we show that the analytical predictions for different bandwidths coincide perfectly with the numerical estimates. The exponentially vanishing loss variances observed for $\sigma \in O(1)$ are expected to render the loss untrainable. This is demonstrated in panel b), where we further train a QCBM  with $\sigma=n/4$ (which approximately maximizes the variance) and $\sigma=1$.
We find that a QCBM with $\sigma = n/4$ can be successfully trained even for $n=1000$ qubits. In contrast, the training starts to fail to learn the $|\vec{0}\rangle$ target state after $n \approx 50$ and is fully untrainable at $n=100$ when $\sigma = 1$ is used.

\begin{figure}
    \centering
    \includegraphics[width=0.99\linewidth]{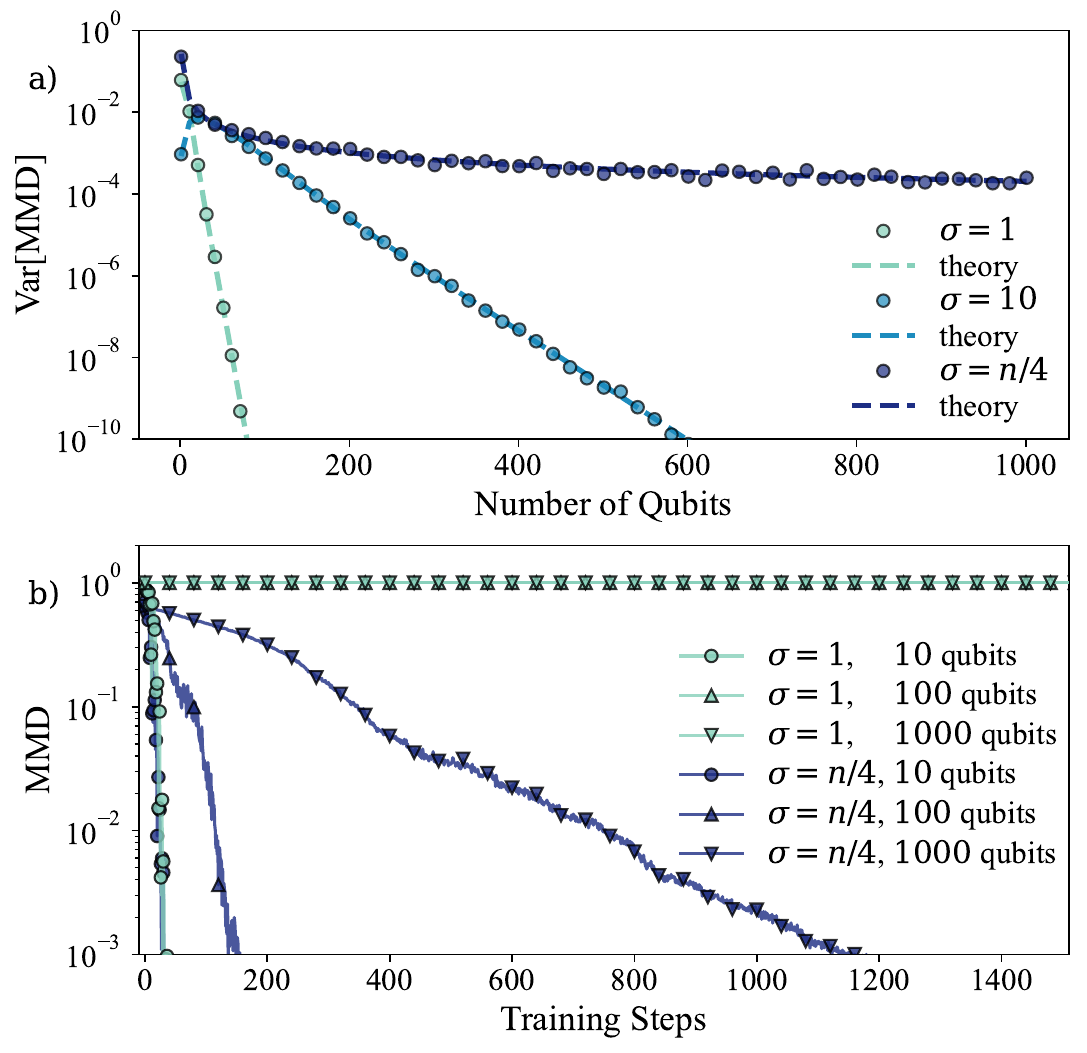}
    \caption{\textbf{$\bm\sigma$-dependence of the MMD loss function.} a) Comparison of the MMD variance between the analytical prediction in Eq.~\eqref{eq:var-mmd-tensor} and empirical variance using 100 measurements from a product state ansatz. b) Training a product state ansatz on the $|\vec{0}\rangle$ target state for $\sigma=1$ and $\sigma=n/4$ using the {CMA-ES~\cite{CMA-ES}} optimizer and 512 measurements. In both cases the QCBM ansatz consists of a single layer of Ry rotations on each qubit.}
    \label{fig:mmd_var_training}
\end{figure}

It is interesting to note that the approximately optimal bandwidth $\sigma\sim \frac{n}{4}$ for the product state ansatz coincides with the so-called \textit{median heuristic}~\cite{Gretton2012mmd} from classical ML literature. For random circuits, the median (hamming) distance between bitstrings is in fact $\frac{n}{2}$, which we satisfy with the factor of $2$ in our kernel convention.

\medskip 

To go towards more practical generative modelling, we recall that Ref.~\cite{napp2022quantifying} proves that cost functions of the form of Eq.~\eqref{eq:VQEcost} using $2k$-body observables with $k \in \OC(\log(n))$ are trainable using 1D-random $\log(n)$ depth circuits.
Since Proposition~\ref{prop:mmd-k-body} implies that the MMD is well approximated by a $\log(n)$-body cost, it should follow that the MMD is also trainable at $\log(n)$ depths. There are a few technical caveats associated with constructing a full proof. For example, the first term of the MMD requires working with 4-designs instead of 2-designs, and the second term depends on the target distribution, leading to additional subtleties. However, there is no strong reason to expect that these technicalities make the MMD untrainable.

To extend the trainability result beyond a simple tensor product ansatz, we consider a generic Pauli rotation ansatz of the form $U_{\rm PQC}(\thv) = U(\alv) U_{\rm tensor}(\btv)$ where
\begin{equation}\label{eq:pauli-rotation-ansatze}
    U(\alv) = \prod_{k=1}^M e^{- i \alpha_k G_k /2} V_k
\end{equation}
such that the generators $\{ G_k \}_{k=1}^M$ are some $n$-qubit Pauli strings $G_k \in  \{ \id, X, Y, Z \}^n$ and $\{ V_k\}_{k=1}^M$ is a set of non-parametrized Clifford gates, and $ U_{\rm tensor} (\btv) = \bigotimes_{i=1^n} U_i(\beta_i)$ is a layer of random single qubit rotations $U_i(\beta_i)$. We additionally assume all parameters are uncorrelated. Then, we have the variance of the MMD loss scales as follows. 
\begin{theorem}[General Pauli rotation ansatz trainability of MMD, informal]\label{thm:mmd-train-general}
Consider a general Pauli rotation ansatz of the form in Eq.~\eqref{eq:pauli-rotation-ansatze} and the MMD loss function as defined in Eq.~\eqref{eq:mmd-loss-implicit} using the Gaussian kernel in Eq.~\eqref{eq:gaussian-kernel} with the bandwidth $\sigma \in \Theta(n)$. As long as the average light cone of the back-propagated MMD observable with $U'(\vec{\alpha})$ remains in the order of $\log(n)$, then the QCBM is trainable in the sense that
\begin{align}
    \Var_{\thv} [\LC^{(\sigma)}_{\rm MMD}(\thv)] \in \Omega(1/\poly(n)) \;.
\end{align}
\end{theorem}

Theorem~\ref{thm:mmd-train-general} indicates that in practice one has to only determine the average light cone of the effective MMD observable back-propagated by a \textit{given} ansatz (see Eq.~\eqref{eq:avg-light-cone} in Supplementary Note~\ref{appx:mmd-general-pauli} for a formal definition) to have a trainability guarantee. In Supplementary Note~\ref{appx:mmd-general-pauli}, we provide further details on how to compute this average light cone for the Pauli rotation ansatz. An example of an ansatz satisfying the few-body light cone condition is a shallow depth circuit with nearest neighbour connectivity (which could be either hardware efficient or problem inspired). Crucially, we emphasise that this light cone argument goes beyond the 2-design assumption and is expected to work even more generally to any ansatz that may not even be in the Pauli rotation form.

It is worth noting that as few mild technical assumptions are required to more formally state Theorem~\ref{thm:mmd-train-general}. These, along with our proof are provided in Supplementary Note~\ref{appx:mmd-general-pauli}. We further remark that although some proof techniques are similar to Ref.~\cite{letcher2023tight}, the main technical challenges here are to analytically show that the covariance between different terms which involve higher moments vanish and compute the exact form of the variance lower bound of the purity term. We further remark that although some proof techniques are similar to Ref.~\cite{letcher2023tight}, the main technical challenges here have arisen from dealing with the two system registers of the MMD which involves computing the exact form of the higher moments.

Theorem~\ref{thm:mmd-train-general} is further supported by our numerical evidence for the trainability of the MMD for deeper circuits and more realistic datasets shown in Fig.~\ref{fig:mmd_variance}. Here we plot the loss variance as a function of circuit depth $L$ and the number of qubits $n$ for $\sigma = n/4$ on four datasets from four different target distributions. We observe that the polynomial scaling of the loss variance does in fact extend beyond product states to shallow circuits, i.e., $L \in \OC(\log(n))$. However, for sufficiently deep circuits, i.e., $L \in \Omega(n)$, the MMD variance appears to decay exponentially. This aligns with expressibility-induced barren plateaus observed in other VQA applications, which occur even for maximally local loss functions, i.e., $k=1$.

\begin{figure*}
    \centering
    \includegraphics[width=0.99\linewidth]{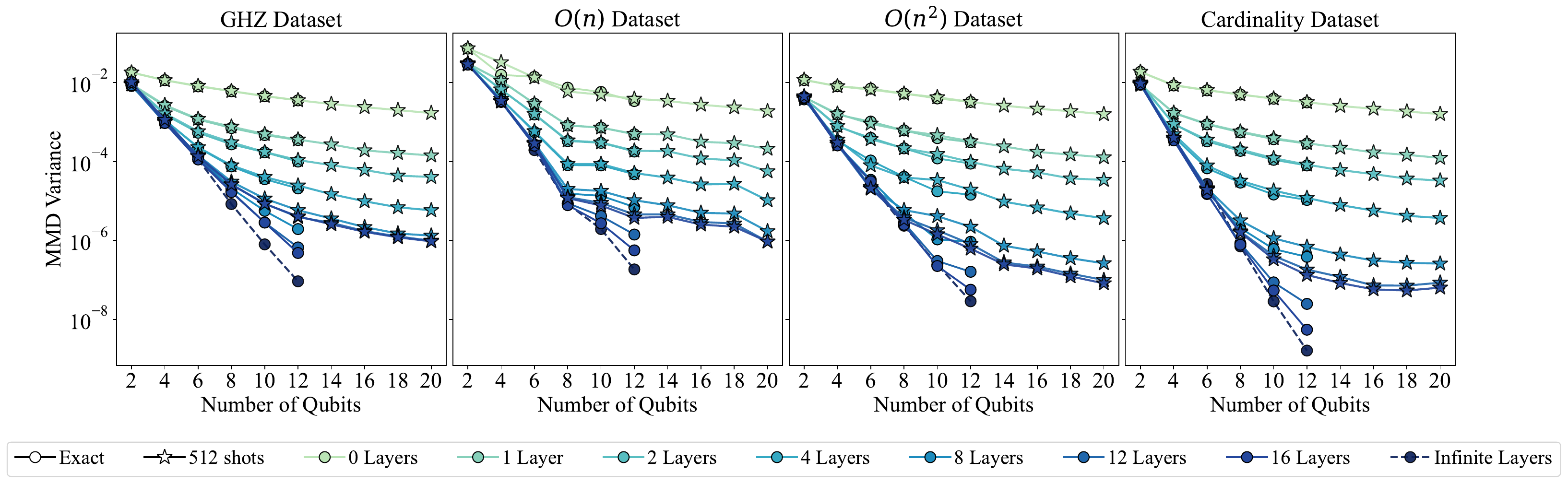}
    \caption{\textbf{Study of loss concentration with the MMD loss function.} Numerical evidence that the MMD loss with Gaussian bandwidth parameter $\sigma=n/4$ does not exhibit global or explicit loss function barren plateaus, but does exhibit loss concentration with deep quantum circuits. We study the loss concentration in randomly initialized line-topology circuits for various datasets, and increasing number of qubits $n$ and circuit depth. The infinite layers results were generated by drawing wavefunctions from the Porter-Thomas distribution~\cite{porter1956fluctuations}. The GHZ dataset consists of the all-0 and all-1 bitstrings ($\OC(1)$ support), the $\OC(n)$ and $\OC(n^2)$ datasets consist of $n$ and $n^2$ random bitstrings, respectively, and the cardinality dataset contains all bitstrings with $\frac{n}{2}$ cardinality ($\OC(2^n)$ support). There does not appear to be a strong data-dependence for the magnitude of the loss variance.}
    \label{fig:mmd_variance}
\end{figure*}

\medskip

\subsubsection{The role of local minima}

Our results so far appear to indicate that picking a single bandwidth $\sigma \in \Theta(n)$ maximizes the trainability of the MMD loss function with a Gaussian kernel. While it is true that this choice maximizes the expected magnitude of initial gradients for a QCBM, non-vanishing gradients are a necessary condition but not sufficient to guarantee reliable training performance. 
And in fact it turns out that while low-body losses exhibit large gradients they come with other limitations. 
Particularly, we show that the bodyness of a generative loss function defines the maximal order of marginals of the target distribution that can be distinguished. That is, the model only learns to match the target distribution on subsets of bits, i.e. on its marginals. This introduces a continuous family of minima which are indistinguishable from the true minimum when using a low-bodied loss function, but which are systematically wrong for the purposes of generative modelling.
The worry is that the non-vanishing loss gradients in low-bodied losses are predominantly due to the presence of such spurious minima and do not point in the direction of the true global minimum. This is sketched in Fig.~\ref{fig:weight_mmd_operator}. 

Formally, let $\qth(\xv_A)$ denote the marginal model distribution on a subset $A\subseteq\{1, 2, ..., n\}$ of qubits, and $\pt(\xv_A)$ the marginal target distribution on that same subset. For more details we refer to Eq.~\eqref{eq:marginal_prob} and Eq.~\eqref{eq:marginal_prob_model} in Supplementary Note~\ref{app:mmd-faithfulness}. The connection between the bodyness of the loss operator and the marginals of the model and target distributions is then formalized in the following Proposition.

\begin{proposition}[The truncated MMD loss is not faithful]\label{prop:mmd_not_faithful}
   Consider a distribution $\qth(\xv)$ that agrees with the training distribution $\pt(\xv)$ on all the marginals up to $k$ bits, but disagrees on higher-order marginals. The distribution $\qth(\xv)$ minimizes the truncated MMD loss. That is, suppose
   \begin{align}
       \qth(\xv_A) = \pt(\xv_A) \;,
   \end{align}
    for all $A\subseteq\{1,2,...,n\}$ with $|A| \leq k$, then
    \begin{align}
       \tilde{\LC}^{(\sigma, k)}_{\rm MMD}(\thv) = 0 \;.
   \end{align}
    Crucially, this is true even if for some $B\subseteq\{1,2,...,n\}$ with $|B| > k$
   \begin{align}
       \qth(\xv_B) \neq \pt(\xv_B) \;.
   \end{align}
\end{proposition}

In other words, if the MMD operator can be approximated well by a 
truncated operator with at most $2k$-body terms, model distributions that match the target distribution exactly up to $k$-body marginals or higher cannot be distinguished from ones that match fully. 
As an example of such distributions, consider the uniform distribution over the bitstrings $[001, 011, 101, 110]$, where the third bit is the bit-wise addition of the previous two bits. Using only second-order marginals, it is not possible to distinguish this correlated distribution from the uniform distribution over all eight possible outcomes. 

Notably, long-range correlations in the data can still be learned by the low-bodied MMD loss, just not ones that are particularly high-order\footnote{Note that this is in contrast to a loss composed purely of local terms which would be restricted to learning local/short-range correlations.}. 
Not all distributions will however exhibit such higher-order correlations and thus some distributions will be learnable using losses composed of low-body terms. Viewing this result through the lens of generalization to the underlying distribution, there are two opposing consequences that would need to be studied in future works. On the one hand, convergence to the systematically wrong low-order minima is likely to impede generalization. On the other hand, one could also imagine that being ignorant to very high order marginals of the training set could reduce overfitting and thus enhance generalization.

Proposition~\ref{prop:mmd_not_faithful} thus establishes that to fulfill the promise of quantum generative models, that is to be able to learn both long-range \textit{and} many-body correlations, one cannot use exclusively low-body losses. However, such a requirement is in immediate tension with the low-bodyness required for the trainability guarantees (see Theorem~\ref{thm:mmd-sigma}). In particular, in Proposition~\ref{prop:mmd-k-body} we show that for $\sigma\in\Theta(n)$ the contribution of $k\in\Theta(n)$ terms are exponentially small in $n$. Thus, although the loss is still strictly faithful given an infinite shot budget, with a reasonable shot budget we will not be able to resolve the contribution from the exponentially small high-body terms. Hence, there can be spurious minima that we cannot resolve from the true minimum and therefore for all practical purposes the loss is effectively not truly faithful. 

One approach to resolving this tension would be to adapt the initial value of $\sigma$ from $\Theta(n)$, where the loss exhibits large gradients but predominantly learns low-order marginals,  towards $\OC(1)$ to also learn high-order correlations as the model improves. This is in line with studies from the classical ML literature showing that bandwidths for optimal MMD performance are oftentimes smaller than the so-called median heuristic~\cite{Gretton2012optimal,sutherland2016generative,garreau2017large}, which coincides with our result of $\sigma\in\Theta(n)$.
Another approach, which is also already employed in classical ML literature, is to use a kernel that averages the effects of several $\sigma$~\cite{li2017mmd,wang2018improving,Li2018domain}. That is, the kernel is taken to be
\begin{equation}\label{eq:average_sigmas}
    K_\textbf{c}(\xv, \yv) = \frac{1}{|\textbf{c}|} \sum_{i\in \textbf{c}} K_{\sigma_i}(\xv, \yv) \equiv  \sum_{l=1}^k  \langle w_{\sigma}(l) \rangle_{\textbf{c}} D_{2l}
\end{equation}
for a set of bandwidths $\textbf{c}=\{\sigma_1, \sigma_2, ... \}$. The resulting weight of each $2k-$body term of the new MMD observable is an average of the weightings corresponding to each $\sigma_i$ in $\mathbf{c}$ as shown in Fig.~\ref{fig:weight_mmd_operator}. Theorem~\ref{thm:mmd-sigma} shows that for a QCBM without inductive bias to not fall prey to exponential concentration, at least one of the $\sigma_i$ needs to be $\Theta(n)$. But the results of Proposition~\ref{prop:mmd_not_faithful} suggest that for data sets exhibiting high-order correlations a small bandwidth $\sigma_i\in\OC(1)$ is required for correct convergence. It stands to reason that the optimal set $\textbf{c}$ contains a spectrum of bandwidths that both enable trainability and faithful convergence to the target distribution (as sketched in Fig.~\ref{fig:weight_mmd_operator}c)). 
How successful this strategy is in practice remains to be determined.
\medskip

\paragraph*{Broader Implications.}

Our work highlights that one can treat classical machine learning losses as quantum observables to study their properties. This implies that our results transfer to other types of quantum generative models beyond the QCBM that will also be affected by the fundamental limitations described by Proposition~\ref{prop:mmd_not_faithful}. In fact, we show in Supplementary Note~\ref{app:faithful_arbitrary} that any generative modelling loss function for classical data that can be brought into the form $\mathcal{L}(\thv) = \text{Tr}[\mathcal{M}\rho_{\thv}]$, with a diagonal measurement operator $\mathcal{M}$, faces the same tension described above. That is, if $\mathcal{M}$ contains at most $k$-body terms in the Pauli basis representation, then the loss cannot distinguish two distributions that agree on all $k$-order marginals but disagree on higher-order marginals. Thus losses composed exclusively of local terms (with the conventional measurement strategy) cannot be used in generative modelling to learn complex higher-order correlations.

With a little thought it becomes clear that an exclusively global loss is also undesirable. Not only do such losses exhibit exponential concentration for unstructured circuits, they will also in general possess spurious minima in virtue of only probing global properties of the distribution (i.e. the average global parity), as shown in Fig~\ref{fig:weight_mmd_operator}.  Instead we advocate using \textit{full}-body losses which contain both low and high-body terms, such as those obtained by averaging in Fig.~\ref{fig:weight_mmd_operator}. Even then, global contributions cannot be vanishingly small or else they will not be possible to resolve with a realistic shot budget. 

For another example, one may aim to train a quantum generative model using a QGAN framework, where a Discriminator $D$ provides a score $D(\xv)$ to every sample. The corresponding operator can then be written as $\mathcal{M} = \sum_{\xv} D(\xv) |\xv\rangle\langle\xv|$. The Discriminator may have to initially implement an effectively low-bodied operator to facilitate initial gradients, but later in training become higher-bodied to learn global features. That is not to say that the Discriminator should only classify marginals of the bitstring such as in Ref.~\cite{leadbeater2021fdivergence}. Rather, the architecture and initialization should be such that the operator $\mathcal{M}$ in the Pauli basis initially contains low-body terms but can include high-body terms during convergence.
Interestingly, the interpolation from trainable to faithful could be naturally full-filled during training when the Generator and Discriminator are optimized in tandem.

Fine tuning the interplay between the loss function gradients, density of local minima and the faithfulness of a generative loss is beyond the scope of this work, but is an important direction for future research. We especially emphasize the necessity to evaluate the implications of our results on models and datasets of practical relevance. In Section~\ref{sec:train-hea-dataset} we take steps in this direction by investigating training a QCBM to model real data from the HEP domain.

\subsubsection{Quantum strategies: quantum fidelity}\label{sec:quantumfid}

\begin{figure*}
    \centering
    \includegraphics[width=0.99\linewidth]{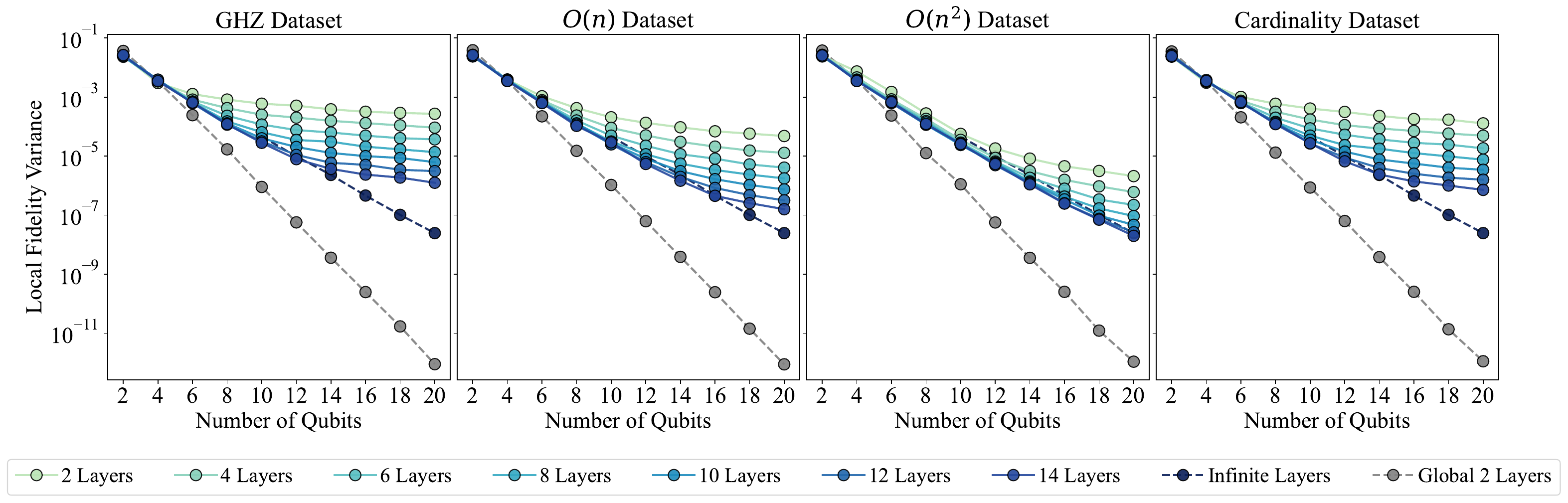}
    \caption{\textbf{Study of loss concentration with the local quantum fidelity loss function.} Numerical evidence that the local quantum fidelity loss function does not exhibit global or explicit loss function barren plateaus. It does however exhibit expressivity-induced barren plateaus with deeper and deeper circuits, as it is the case for all generic VQA-type loss functions in the form of Eq.~\eqref{eq:VQEcost}. In contrast, the global quantum fidelity variance decays exponentially at all circuit depths. The numerical setup is the same as for the MMD in Fig.~\ref{fig:mmd_variance}, and the infinite layers results were generated by drawing wavefunctions from the Porter-Thomas distribution~\cite{porter1956fluctuations}}.
    \label{fig:fidelity_variance}
\end{figure*}

While the classical fidelity in Eq.~\eqref{eq:classical_fidelity} is an explicit cost function, the quantum fidelity, defined in Eq.~\eqref{eq:quantumfid}, allows for a simple known quantum estimation strategy. Key to the quantum fidelity loss is to interpret the training distribution as a target state $\ket{\phi} = \sum_{\xv} \sqrt{\pt (\xv)} \ket{\xv}$. The QCBM model loss can then be rewritten as the expectation of an observable, e.g. in the form of Eq.~\eqref{eq:VQEcost}, with $\rho = |\phi \rangle\langle \phi |  $ and $O = | \vec{0}\rangle\langle \vec{0} |$ being the all-zero projective measurement. 
Crucially, as $O = | \vec{0}\rangle\langle \vec{0} |$ is a global projector, the quantum fidelity is subject to a globality-induced barren plateaus~\cite{cerezo2020cost} and the loss exponentially concentrates towards one~\cite{arrasmith2021equivalence}. That is, we have
\begin{align}
    \Var_{\thv}[ \LC_{QF}(\thv)] \in \OC(1/b^n) \;.
\end{align}
This global-measurement-induced barren plateau can however be avoided by localising $\LC_{QF}(\thv)$. That is, we replace the global projective measurement $|0 \rangle \langle 0|$ with its local version $H_L = \frac{1}{n} \sum_{i = 1}^n \ket{0_i}\bra{0_i} \otimes \id_{\bar{i}} $, where $\bar{i}$ indicates all qubits except qubit $i$. The new localised version of the quantum fidelity loss is given by
\begin{align}\label{eq:local-quantum-fidelity}
    \LC_{QF}^{(L)}(\thv) = 1 - \bra{\phi} U(\thv) H_L U^{\dagger}(\thv) \ket{\phi} \;.
\end{align}
This local loss is faithful to its global variant for product state training in the sense that it vanishes under the same conditions~\cite{khatri2019quantum}, i.e. when the QCBM distribution matches the data distribution exactly. However, it enjoys trainability guarantees via the results of Ref.~\cite{cerezo2020cost}. This implies that, unlike the MMD and other classical losses that utilize the conventional measurement strategy, the local quantum fidelity can effectively distinguish between high-order marginals even at $k=1$ bodyness. 
However, although the local loss function can evade global measurement-induced BPs, it still suffers under BPs from other sources, such as expressibility or noise. Additionally, it is not yet explored how practical a fidelity loss is for the purposes of generalizing from training data.

Fig.~\ref{fig:fidelity_variance} depicts numerical variance results for the fidelity loss on a range of datasets, circuit depths and numbers of qubits. For all datasets, the local quantum fidelity exhibits only polynomially decaying variance over random parameters when the quantum circuits are not too deep. As a reference, we additionally depict the global quantum fidelity which exponentially decays for all circuit depths.

\medskip

The challenge now becomes how to estimate $\LC_{QF}^{(L)}(\thv)$ using measurements from the quantum computer. The seemingly straight-forward approach is to prepare the initial state $\ket{\phi}$, evolve it under $U^\dagger (\thv)$, and then evaluate the observable defined by $H_L$ through measurements in the computational basis. However, loading classical data into a quantum state $\ket{\phi}$ is not expected to be feasible in general. In Supplementary Note~\ref{ap:localcost}, we propose an approach that can be used to estimate $\LC_{QF}^{(L)}(\thv)$ using a series of Hadamard tests without needing to prepare $\ket{\phi}$. We note that, while in theory our approach requires a number of Hadamard tests that scales with the amount training data, we expect stochastic techniques, such as stochastic gradient descent~\cite{sweke2020stochastic}, to be sufficient in practice.

\medskip

\paragraph*{Broader Implications.}

In this section we have presented one example of a quantum strategy to measure a fidelity-based loss for quantum generative modelling. While this approach puts more load on the quantum computer as compared to losses employing the conventional measurement strategy, it enjoys simultaneous trainability and faithfulness to the target distribution.

An interesting extension would be to explore other quantum approaches for efficiently training QML models.
One could for example attempt to compute the KL divergence or other explicit losses directly on the quantum computer.
Although the implementation of non-linear operations on quantum computers has been demonstrated in Ref.~\cite{holmes2023nonlinear, gilyen2019quantum, martyn2021grand}, we are not yet aware of quantum strategies beyond one related demonstration for the Rènyi divergence in Ref.~\cite{kieferova2021quantum}. One alternative approach would be to attempt to indirectly turn the QCBM into an explicit generative model by estimating its probabilities using amplitude amplification or other techniques. As discussed in Section~\ref{sec:exactvariance}, with access to exact probabilities the KL divergence can provably avoid barren plateaus even with unstructured circuits for certain target distributions.

\subsection{Training on a HEP dataset}\label{sec:train-hea-dataset}

\begin{figure*}
    \centering
    \includegraphics[scale=0.5]{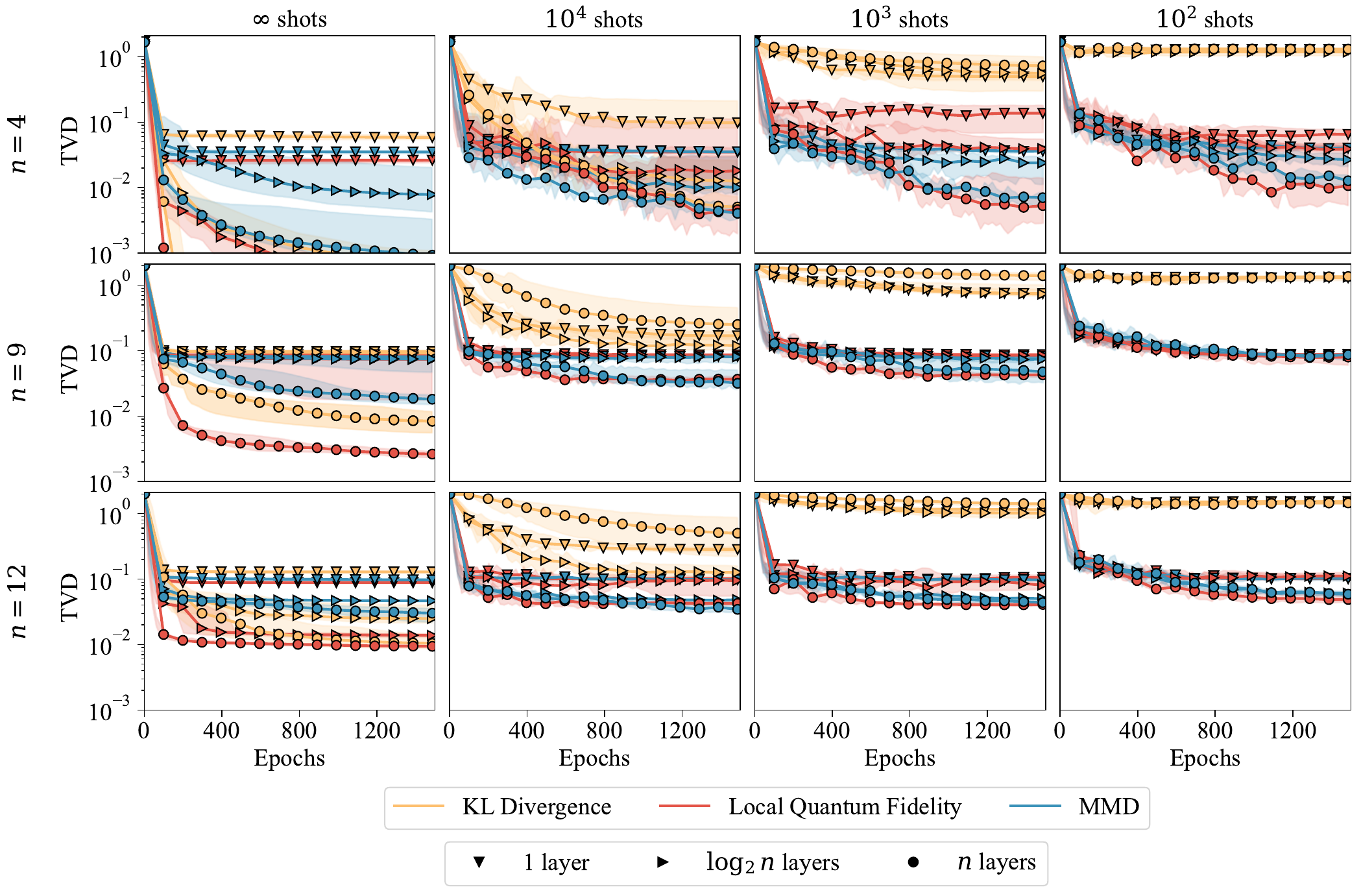}
    \caption{\textbf{Finite-shot comparison of loss functions.} TVD, computed with infinite statistics, on the training curve of the QCBMs with varying number of qubits $n=4$, $n=9$ and $n=12$ (rows) and layers (symbols), where the gradients are computed with different number of shots (columns) for different loss function (colours).}
    \label{fig:TV}
\end{figure*}

\begin{figure*}
    \centering
    \includegraphics[scale=0.55]{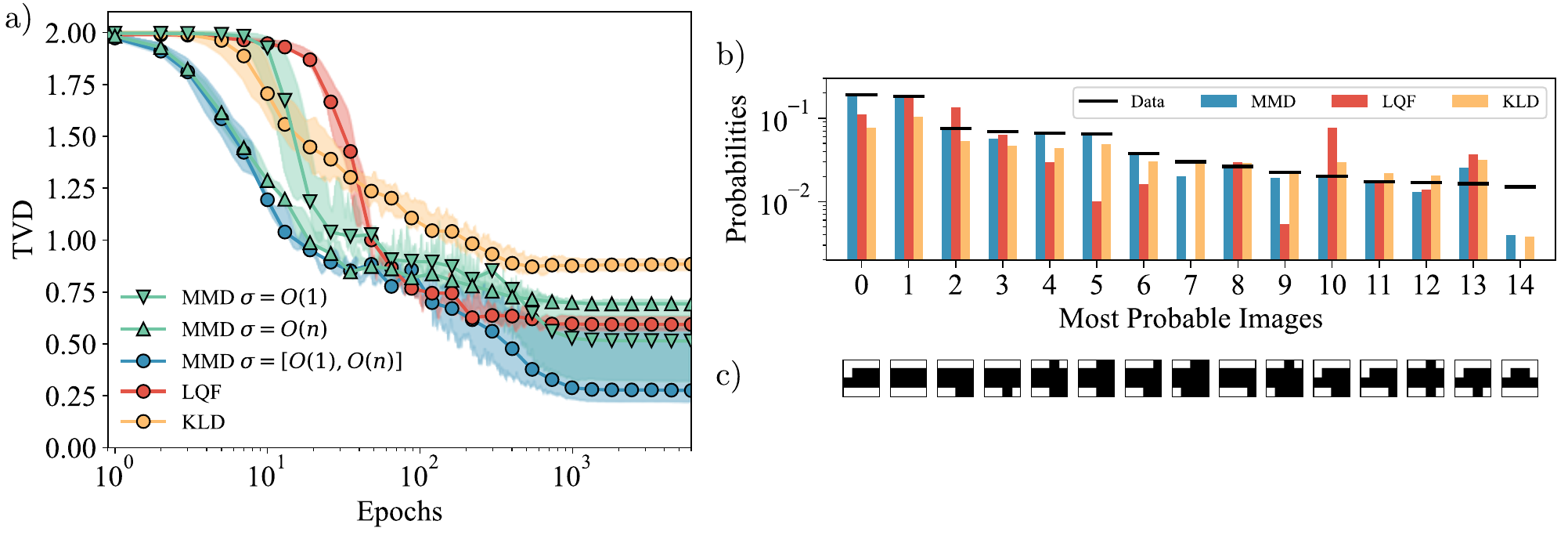}
    \caption{\textbf{16 qubit QCBM training.} a) Exact TVD computed during training of $n=16$ QCBMs with $\log_2{n}=4$ layers and 100 shots per loss evaluation. We report median TVD values and 25\% to 75\% percentiles for the MMD, LQF and KLD loss functions. For the MMD, bandwidths $\sigma = 0.01, n/4$, and $n$ are used to both showcase improved trainability of large $\sigma$ and improved convergence of small $\sigma$. b) Histograms of the trained QCBMs on the 15 most probable images, which are shown in panel c). The black lines denote the training dataset probabilities.}
    \label{fig:16}
\end{figure*}
In this section, we perform realistic training of QCBMs on a more practical dataset which is derived from HEP colliders experiments. We compare the implicit cost functions MMD and local quantum fidelity (LQF) with the explicit KL divergence for an increasing number of circuit depth $L$ and the number of qubits $n$, and across several measurement budgets. To summarize our results, we observe that the presence of shot noise causes the training with KLD to fail, while the MMD and LQF both hold up significantly better. 

\medskip

\paragraph*{Dataset.}
We consider a dataset consisting of energy depositions in an electromagnetic calorimeter (ECAL)~\cite{ECAT_data}. The data was generated using a Monte Carlo approach (theGeant4 toolkit~\cite{Geant4}), which accurately describes the ECAL detector behaviour under a typically proton- proton collision at a LHC experiment. The dataset consists of the energy deposition on a $25 \times 25 \times 25$ grid, that we downsized to a two dimensional grid of various sizes. The images are converted to a black and white scale by considering the pixel `hit' if the energy deposition exceeds a certain threshold, which is chosen as one tenth of the mean energy deposit.  We map each pixel to a qubit and take the state $|1\rangle$ to represent a hit. This dataset naturally has a polynomial support and thus is precisely the type of dataset that we might hope to learn using quantum machine learning.

\medskip

\paragraph*{Training.} 
We use a parametrized quantum circuit of the form 
\begin{equation}
\label{eq:ansatz}
    U(\thv) = \left[\prod_{l=1}^{L} \text{R}(\thv_l) W_l(\vec{\alpha}_l)\right] \text{R}(\thv_0) \; ,
\end{equation}
where R$(\thv_l)$ is a layer of arbitrary single qubit unitaries that can be parameterised using $3n$ Euler angles, $W(\vec{\alpha}_l) := \prod_{i=1}^{n-1}\text{CX}_{i,i+1}\text{RY}_i(\alpha_l^i)\text{CX}_{i,i+1}$ acts as parametrized entangling gate with CX$_{i,j}$ a CNOT gate between qubits $i$ and $j$ and RY$_i(\alpha_l^i)$ a single qubit rotation of qubit $i$ around the $y$-axis, and the parameters $\thv =  \{\thv_l,\vec{\alpha}_l\}$. We use the total variation distance (TVD), see Eq.~\eqref{eq:loss_TV}, as a common metric to assess the performance of each loss function. To verify performance accurately, we compute the TVD using exact simulation. The gradients for each loss are computed using the parameter shift rule~\cite{quantum_grad} which provides estimates of the analytical gradient, and the parameters are updated with the ADAM~\cite{adam} optimizer with a decaying learning rate lr$(t)$ = $\max{(0.01e^{-\beta t},10^{-5})}$, where $t$ is the  optimization step and $\beta=0.005$ the decaying rate. The computation of the KLD is stabilised using a regulariser of $\epsilon = 10^{-6}$, which has been tuned through trial and error. To follow best-practices with the MMD, we average the gradient estimation over several different bandwidths 
\begin{equation}
    \vec{\sigma} = \begin{pmatrix} 0.01, &0.1,&0.25,&0.5,&1,&10\end{pmatrix} n \;,
\end{equation}
which incurs no additional quantum resources. This makes the loss full-bodied and thus keeps the model trainable while aiding convergence. We note that these are likely not the optimal bandwidths to average over but it demonstrates the best-practice approach.

\medskip

\paragraph*{Results.}

Fig.~\ref{fig:TV} shows the TVD, computed with infinite statistics, on the training curve of the QCBMs with varying number of qubits $n\in \{4,9,12\}$ (rows) and layers (symbols), where the gradients are computed with different number of shots (columns) for different loss function (colours). The lines denote the median over ten random parameter initialisation while the shaded area denotes the $25\%$ to $75\%$ percentile. We observe that the performance of the KLD quickly deteriorates as the number of shots is reduced while the MMD and local quantum fidelity (LQF) remain more stable. We further observe that increasing the expressivity of the QCBM from $\log_2{n}$ to $n$ layers does not lead to a significant increase in performance for a low number of shots.

To demonstrate the scalability to larger systems, we also train a $n=16$ QCBM with $\log_2{n}=4$ layers and 100 shots per function evaluation. The quantum circuit has a linear entangling topology and the initial parameters are chosen uniformly at random. We highlight that the downsized data at $n=16$ is structurally different to the one in Fig.~\ref{fig:TV}, which can yield quantitatively different results. In panel~a) of Fig.~\ref{fig:16}, we depict the median and $25\%$ to $75\%$ percentiles for the KLD and LQF over 50 and 10 random repetitions respectively, whereas for the MMD we use 6 random repetitions per line. We also compare the MMD performance for different bandwidths $\sigma = 0.01$, $\sigma = n/4$, and a kernel averaging $\sigma = [0.01, n/4, n]$. We indeed see that large $\sigma$ show improved initial trainability and small $\sigma$ improved convergence. In panel~b) we show the probability histograms of the 15 most occurring images in the dataset, as well as the final respective model probabilities. The corresponding $4\times 4$ pixel images are displayed at the bottom panel~c). 

In this 16-qubit example, it appears that the LQF is no longer performing on-par with the MMD, as was the case in Fig.~\ref{fig:TV} for smaller system sizes. A possible explanation is that one chooses all relative phases in the data state $|\phi\rangle = \sum_{\xv} \sqrt{\pt(\xv)}|\xv\rangle$, which strongly reduces the number of wavefunctions that minimize the LQF loss even though they may produce the desired measurement distribution. This may not only produce less solutions, it also enforces that the ansatz needs to be able to express exactly that state. While this could be leveraged using specialized real-valued ans\"{a}tze, this is not attempted here. We conclude that the practical properties of the LQF loss as compared to implicit losses using the conventional measurement strategy are still to be studied in more detail.

To emphasise the importance of the size of the support, in Supplementary Note~\ref{app:training} we also consider an exponential version of the dataset, by using a negative logarithm transformation.
We find in this case that the KLD does not suffer from exponential concentration and can be trained. This explains the successes previously observed for training QCBMs using the KLD for small scale problems. However, as the amount of classical training data cannot scale exponentially these successes are not relevant to larger, non-classically simulable, problems.

\section{Discussion}

In this work, we have introduced the notion of explicit and implicit losses, which broadly reflect the capabilities of explicit and implicit generative models~\cite{mohamed2016learning}. We argue that these concepts provide a useful framework to understand the trainability of quantum generative models. In particular, we argue that the mismatch between the indirect access to probabilities provided by implicit models with the explicit probabilities required by explicit losses renders implicit models untrainable via explicit losses. More concretely, focusing our attention on quantum circuit Born machines as a commonly used implicit model, we prove that pairwise explicit losses exponentially concentrate (Theorem~\ref{thm:explicit-loss}). This result prohibits efficient training using a large class of commonly-used losses including the KL divergence, JS divergence and the total variation distance. Such losses may however be usable with explicit ``quantum'' generative models such as tensor network Born machines~\cite{Han2018TNBM,Cheng2019TTNBM,vieijra2022PEPSBM}.

Crucially, our results assume access to a polynomial (in the number of qubits) number of training data samples and measurements from the quantum circuit.
With only moderate numbers of qubits, this assumption is unnecessary and explicit losses such as the KL divergence may appear to be trainable (see e.g. Refs.~\cite{Hamilton2018,leyton2019robust,coyle2020generativeFinance, Zhu2018,rudolph2022generation}).
More generally, if we restrict the number of qubits used to classically simulable sizes, this assumption can be lifted and one could use quantum generative models purely for their efficient sampling capabilities. However, to harness the full potential of quantum generative modelling one surely wants to push to non-classically tractable problem sizes, at which point this assumption is essential. 
For example, even with only 50 qubits, access to $\sim 2^{50} \approx 10^{15}$ training samples or quantum measurements is unrealistic.

While formulated initially for random quantum circuits, the intuition underlying Theorem~\ref{thm:explicit-loss} suggests our no-go result extends to scenarios where the implicit generative model's measurement distribution only has polynomial support (e.g., near-identity initialization of the circuit~\cite{grant2019initialization}), as well as beyond the pairwise explicit form of the explicit loss. One exception may be if the quantum generative model has a strong inductive bias. Hence our work further motivates the search for new methods for constructing parameterised circuits with strong inductive biases (e.g. via warm starts~\cite{sauvage2021flip, liu2023mitigating, cheng2022clifford, mitarai2022quadratic, rudolph2022synergy} or incorporating symmetry constraints~\cite{larocca2022group,schatzki2022theoretical, nguyen2022atheory, ragone2022representation, meyer2022exploiting}).

In contrast to explicit losses, implicit losses are naturally suited to training implicit models. Within this line of thought, we have identified the MMD loss with a Gaussian kernel as a promising implicit loss for training QCBMs. We show that this loss can be interpreted as the expectation value of a quantum observable, where crucially the properties of the observable depend on the bandwidth parameter $\sigma$. 
In the common case where $\sigma$ is independent of the system size, $\sigma \in \OC(1)$, the observable becomes predominantly global and thus exponentially concentrates.
Conversely, when $\sigma$ scales linearly with the system size, $\sigma \in \Theta(n)$, the low-body interaction terms in the observable are largely dominant over the global terms and hence exhibits large gradients. We further use these insights to derive a rigorous polynomially lower bound on the MMD loss variance for a wide range of structured and unstructured circuits with small effective light cones.

Our main results for explicit and implicit losses assume the conventional strategy for estimating a generative loss function from an implicit model, where the model provides a set of samples in the computational basis, which are then used to estimate the loss in conjunction with the training data samples. While this is the standard classical strategy, quantum generative models can employ alternative quantum strategies by leveraging quantum computing power.
As an example, we propose the local quantum fidelity as a trainable loss function for generative modelling. Developing alternative quantum strategies for training quantum generative models is an interesting avenue for future research. A natural candidate might be, as suggested in Ref.~\cite{coyle2020born}, to implement the MMD loss with a \textit{quantum} kernel, where the kernel values themselves are estimated using quantum computers. While one could hope for a potential quantum advantage with this approach (especially when training on quantum data), there is the additional challenge that quantum kernels without inductive bias tend to exponentially concentrate~\cite{thanasilp2022exponential}.

To put our conclusions to the test, we studied how these loss functions perform in more practical scenarios with data derived from High Energy Physics experiments at the LHC. This dataset naturally satisfies our assumptions of a polynomial number of data samples at all system sizes. Our training results are found to be consistent with our theoretical predictions in which both the MMD and quantum fidelity losses significantly outperform the KLD loss when a strict measurement budget is employed.

Finally, while our work addresses the question of whether a given loss exhibits non-exponentially vanishing gradients, we stress that this is just one ingredient among many to ensure the success of quantum generative modelling. Of particular importance is the observation that models with local losses will generally struggle to learn global correlations due to the function's inability to distinguish high-order features in the data. Hence we advocate using \textit{full}-body losses, which contain both low and high-body terms, for quantum generative modelling. More broadly, the ability of a model to successfully generalize will also presumably depend on the choice in loss, but this is beyond the scope of this work. Nonetheless, ensuring non-vanishing loss gradients and ensuring faithfulness of the loss function are critical steps since failing here precludes the successful training and generalization of quantum generative models. Hence, our work constitutes an important first step to understanding of the barriers that need to be overcome to achieve a quantum advantage in generative modelling. 

\section{Data Availability}
Data generated and analyzed during the current study are available from the corresponding author upon reasonable request.

\section{Code Availability}
Code used for the current study is available from the corresponding author upon reasonable request.

\section{Acknowledgments}
    The authors would like to acknowledge Christa Zoufal, Samson Wang and Marco Cerezo for helpful discussions.
    ST and ZH acknowledge support from the Sandoz Family Foundation-Monique de Meuron
program for Academic Promotion. OK, SV and MG are supported by CERN through the CERN Quantum Technology Initiative.

\section{Author Contributions}
M.R., S.L. and S.T. contributed equally to this work. The project was conceived by M.R., S.T., and Z.H. Theoretical results were proved by M.R., S.L., S.T., O.S., and Z.H. Numerical implementations were performed by M.R. and O.K. The practical application of QCBMs was conceived and guided by S.V. and M.G. All authors contributed to the scientific discussions. The manuscript was written by all authors.

\section{Competing Interests}
The authors declare no competing interests.

\bibliography{quantum.bib,cern.bib}

\setcounter{theorem}{0}
\setcounter{proposition}{1}
\setcounter{corollary}{1}

\clearpage
\newpage
\onecolumngrid

\setcounter{section}{0}
\vspace{0.5in}
\begin{center}
	{\Large \bf Supplementary Information} 
\end{center}

\section{Supplementary Note - Technical nuances in explicit and implicit losses}\label{sec:technical_nuances}

In this work, we have argued that there are two main classes of loss functions for generative modelling tasks. First are \textit{explicit} losses which have the form 
\begin{equation}\label{eq:explicitloss-app}
   \LC_{\rm expl}(\thv) := \sum_{\vec{x}_1 ... \vec{x}_{r}} f\Big(p(\xv_1), ..., p(\xv_r), \qth(\xv_{1}), ..., \qth(\xv_{r}) \Big) 
   \, ,
\end{equation}
where $f(\cdot)$ is a function that depends on the target probabilities $p(\xv_i)$ and model probabilities $ \qth(\xv_{i})$ for data variables $\xv_i \in \XC$ with $i = 1, \, ... \, , r$ (but not the data samples themselves). 
The other are \textit{implicit} loss functions which are expressed as
\begin{equation}\label{eq:implicitloss-app}
    \LC_{\rm impl}(\thv) := \mathbb{E}_{\xv_1, ..., \xv_r \sim \{p, q_{\thv}\}} \, g(\xv_1, ..., \xv_r) \, ,
\end{equation}
where $g(\xv_1, ..., \xv_r)$ is some function that depends on the data (but not probabilities), and an expectation is over data variables $\xv_1, ..., \xv_r$ sampled either from the data distribution $p$ or the model distribution $\qth$.

While almost all practical loss functions for generative models fall either into the definitions of explicit or implicit losses (see Sec.~\ref{sec:losses}), there exists a technical caveat which allows loss functions to be classified both as explicit and implicit. 
To provide a rather general example, we consider a general loss of the form
\begin{align}
    \LC(\thv) &= \sum_{\vec{x},\vec{y} \in \XC} \qth(\vec{x})p(\vec{y}) g(\cdot) \label{eq:genlossversion1} \\
    &= \Ebb_{\vec{x} \sim \qth,\vec{y}\sim p}[g(\cdot)] \label{eq:genlossversion2}
\end{align}
where the function $g$ is for now left arbitrary. By comparing Eq.~\eqref{eq:genlossversion1} and Eq.~\eqref{eq:explicitloss-app}, we can see that for $\LC(\thv)$ to be an explicit loss $g$ cannot depend on $\xv, \yv$. Conversely, by comparing Eq.~\eqref{eq:genlossversion2} and Eq.~\eqref{eq:implicitloss-app}, for $\LC(\thv)$ to be implicit, $g$ cannot depend on $p(\xv)$ or $\qth(\yv)$ . Thus, it appears impossible for such a function to be both explicit and implicit. However, taking $g$ as the Kronecker delta $\delta_{\xv,\yv}$ allows one to straddle the definitions. That is, for $g = \delta_{\xv,\yv}$, $\LC(\thv)$ is implicit as it can still be written as an average over samples in Eq.~\eqref{eq:genlossversion2} but also we can write $\LC(\thv) = \sum_{\vec{x} \in \XC} \qth(\vec{x})p(\vec{x})$ which is of the form of an explicit loss given in Eq.~\eqref{eq:explicitloss-app}.

This case is not just purely hypothetical and arises, for example, for the MMD loss, Eq.~\eqref{eq:mmd-loss-general}, with the kernel $K(\xv, \yv) = \delta_{\xv,\yv}$. Here we have 
\begin{align}
    \LC(\thv)
    =& \sum_{\vec{x},\vec{y} \in \XC} \qth(\vec{x})\qth(\vec{y}) \delta_{\xv,\yv}
     - 2 \sum_{\vec{x},\vec{y}  \in \XC} \qth(\vec{x})p(\vec{y}) \delta_{\xv,\yv}
     + \sum_{\vec{x},\vec{y} \in \XC} p(\vec{x})p(\vec{y}) \delta_{\xv,\yv} \nonumber \\
    = & \sum_{\vec{x}\in \XC} \qth(\vec{x})^2 - 2\qth(\vec{x})p(\vec{x}) + p(\vec{x})^2 \nonumber \\
    = & \sum_{\vec{x}\in \XC} \left( \qth(\vec{x}) - p(\vec{x}) \right)^2 \;, \label{eq:kronecker_transformation}
\end{align}
which takes the form of a pairwise explicit loss (see Eq.~\eqref{eq:pairwiseexplicitloss}).
Consequently, the MMD is always an implicit loss function, but with this particular choice of kernel, it can additionally gain an explicit character. Interestingly, this loss is now very related to the TVD loss in Eq.~\eqref{eq:loss_TV} for which there exists no implicit form. 

This example highlights that the distinction between explicit and implicit losses become non-mutually exclusive for implicit losses formulated using delta functions.
Instead of adding this exception into either of the definitions of explicit or implicit losses, we acknowledge and embrace its existence. Through the steps outlined in Eq.~\eqref{eq:kronecker_transformation}, it may be possible to transform an explicit function exactly or approximately into a form which falls under the definition of an implicit function, which could then seemingly be used with implicit models. This may include functions whose Taylor series expansion in the probabilities converges quickly and can thus be approximated with a finite number of terms. Any term with a positive integer power in the probabilities can then be transformed into an expectation over samples by leveraging the procedure outlined above. Unfortunately, we are not yet aware of any practical examples. In fact, we show that a function which can be classified as a pairwise explicit function directly suffers from untrainability with the conventional measurement strategy. There may however be cases where a quantum strategy can be found that leverages this $\delta$-trick.

\section{Supplementary Note - Analysis on pairwise explicit loss functions}

Here we detail the proofs of our no-go results regarding the pairwise explicit loss function. In particular, Theorem~\ref{thm:explicit-loss}, which concerns the concentration of a statistical estimate of a pairwise explicit loss, is proved in Supplementary Note~\ref{app:pairwise-loss-concentration}. In Supplementary Note~\ref{app:pairwise-no-overlap}, we show that a generative model with polynomial support but which has no inductive bias that aligns with the target distribution is untrainable. Finally, we use the Theorem~\ref{thm:explicit-loss} to derive Corollary~\ref{coro:untrain} in Supplementary Note~\ref{app:pairwise-untrainability}, showing the untrainability of the pairwise explicit loss.

For convenience, we recall relevant terms and notations. 
We are interested in training the model probabilities $\qth(\xv)$ to match some unknown target distribution $p(\xv)$ by minimizing the \textit{pairwise explicit loss} of the form
\begin{align}
    \LC(\thv) = \sum_{\xv} f(p(\xv),\qth(\xv)) \;,
\end{align}
where $f(\cdot)$ is some arbitrary function that measure the similarity between $p(\xv)$ and $\qth(\xv)$. As a reminder, this pairwise form encompasses the broad range of practical loss functions including the famous KL divergence, JS divergence, total variational distance as well as classical fidelity. See Sec.~\ref{sec:losses} for details.

Given the training dataset $\Tilde{P}$ and model samples $\Tilde{Q}_{\thv}$ corresponding to the empirical probabilities $\pt(\xv)$ and $\qtth(\xv)$ respectively, the statistical estimate of the loss function can be written as
\begin{align}
    \Tilde{\LC}(\thv) = \sum_{\xv } f(\pt(\xv), \qtth(\xv)) \;.
\end{align}

Crucially, we emphasise that, for the large system's size, the number of training data and model samples can scale at most polynomially with the number of qubits i.e., $M = |\Tilde{P}|, N = |\Tilde{Q}_{\thv}| \in \OC(\poly(n))$.

\subsection{Concentration of the pairwise explicit loss function: Proof of Theorem~\ref{thm:explicit-loss}}\label{app:pairwise-loss-concentration}

In this section, we rigorously prove Theorem~\ref{thm:explicit-loss}, which shows the concentration of the statistical estimate of the pairwise explicit loss. 
For convenience, the theorem is recalled below.

\begin{theorem}[Concentration of pairwise explicit loss for concentrated models]\label{thm:explicit-loss-appendix}
Consider the loss function of the form in Eq.~\eqref{eq:pairwiseexplicitloss}. 
Assume that for all bitstrings in the training dataset, $\xv \in \Tilde{P}$, the quantum generative model $\qth(\xv)$ exponentially concentrates towards some exponentially small value (as defined in Definition~\ref{def:exp-concentration}). Suppose that $N \in \OC(\poly(n))$ samples are collected from the quantum model corresponding to the set of sampled bitstrings $\Tilde{Q}_{\thv}$, and that the training dataset $\tilde{P}$ contains $M \in \OC(\poly(n))$ samples.
We define the fixed point of the loss as
\begin{align}
    \LC_0(\Tilde{P}, \Tilde{Q}_{\thv}) = \sum_{\xv \in \PC} f(\pt(\xv), 0) + \sum_{\xv \in \QC_{\thv}} f(0, \qtth(\xv)) \;,
\end{align}
with $\PC$ (and $\QC_{\thv}$) being a set of \textit{unique} bitstrings in $\Tilde{P}$ (and $\Tilde{Q}_{\thv}$). Then, the probability that the estimated value $\Tilde{\LC}(\thv)$ is equal to $\LC_0(\Tilde{P}, \Tilde{Q}_{\thv})$ is exponentially close to 1, i.e.,
\begin{align}
    {\rm Pr}_{\Tilde{Q}_{\thv},\thv}[\Tilde{\LC}(\thv) = \LC_0(\Tilde{P}, \Tilde{Q}_{\thv})] \geq 1 - \delta \;,
\end{align}
with $\delta \in \OC\left(\frac{\poly(n)}{c^n}\right)$ for some $c > 1$. 
\end{theorem} 
\begin{proof}

The statistical estimate $\Tilde{\LC}(\thv)$ is equal to $\LC_0(\Tilde{P}, \Tilde{Q}_{\thv}) = \sum_{\xv \in \PC} f(\pt(\xv), 0) + \sum_{\xv \in \QC_{\thv}} f(0, \qtth(\xv))$ when there is no overlap between $\Tilde{P}$ and $\Tilde{Q}_{\thv}$ i.e., $ \Tilde{P} \cap \Tilde{Q}_{\thv} = \{ \}$. The proof of the theorem is equivalent to proving that, after taking $N$ measurement shots, the probability of not obtaining any bitstrings in the training dataset is exponentially close to 1. The probability that there is no overlap between $\Tilde{P}$ and $\Tilde{Q}_{\thv}$ is the probability that $\qtth(\xv)=0$ for all $\xv\in\PC$ which is also equivalent to ${\rm Pr}\left[\sum_{x \in \PC} \qtth(\xv) = 0\right]$ as $\qtth(\xv)\geq 0$ for all $\xv\in\XC$. So, we have
\begin{align}
     {\rm Pr}_{\Tilde{Q}_{\thv},\thv}[\Tilde{\LC}(\thv) = \LC_0(\Tilde{P}, \Tilde{Q}_{\thv})] = & {\rm Pr}_{\Tilde{Q}_{\thv},\thv}\left[\sum_{x \in \PC} \qtth(\xv) = 0\right] \\
     = & \int_{0}^1 {\rm Pr}_{\Tilde{Q}_{\thv}}\left[ \sum_{\xv \in \PC} \qtth(\xv) = 0 \Bigg| \sum_{\xv \in \PC} \qth(\xv) = s \right] {\rm Pr}_{\thv} \left[\sum_{\xv \in \PC}\qth(\xv) = s\right] ds \\
     = & \int_{0}^1 (1-s)^N {\rm Pr}_{\thv}\left[\sum_{\xv \in \PC}\qth(\xv) = s\right] ds \\
     \geq & \int_{\mu_s - \sqrt{\sigma_s}}^{\mu_s+\sqrt{\sigma_s}} (1-s)^N {\rm Pr}_{\thv}\left[\sum_{\xv \in \PC}\qth(\xv) = s\right] ds \\
     \geq & (1-(\mu_s + \sqrt{\sigma_s}))^N \int_{\mu_s - \sqrt{\sigma_s}}^{\mu_s+\sqrt{\sigma_s}} {\rm Pr}_{\thv}\left[\sum_{\xv \in \PC}\qth(\xv) = s\right] ds \\
     \geq & (1-(\mu_s + \sqrt{\sigma_s}))^N (1 - \sigma_s) \;. \label{eq:proof-thm1-last}
\end{align}
In the second equality Bayes' theorem is used to introduce the conditional probability of non-overlap between model samples and the training dataset for given $s = \sum_{\xv \in \PC} \qth(\xv)$ and the marginal probability is obtained by summing over all possible values of $\sum_{\xv \in \PC} \qth(\xv)$. The third equality uses the independence of each model sample and the fact that the probability that one drawn model bitstring is not in the training dataset is given by $(1-s)$. The first inequality is due to restricting the integration range (as the integrand is always greater than zero) with $\mu_s$ and $\sigma^2_s$ being the mean and variance of $\sum_{\xv \in \PC} \qth(\xv)$ over $\thv$. The second inequality is due to taking the maximum value of $s$ and thus the minimum value of $(1-s)^N$ within the integration range. To see how to reach the last line, we invoke Chebyshev's inequality on $\sum_{\xv \in \PC} \qth(\xv)$
\begin{align}
    {\rm Pr}_{\thv}\left[\bigg|\sum_{\xv \in \PC} \qth(\xv) - \mu_s \bigg| \geq k \sigma_s \right] \leq \frac{1}{k^2} \;. 
\end{align}
By specifying $k = 1/\sqrt{\sigma_s}$ and inverting the inequality, we have
\begin{align}
    {\rm Pr}_{\thv}\left[\bigg|\sum_{\xv \in \PC} \qth(\xv) - \mu_s \bigg| \leq \sqrt{ \sigma_s} \right] \geq 1 - \sigma_s\;.
\end{align}
With $\int_a^b \text{Pr}[x] dx = \text{Pr}[a\leq x \leq b]$, we then get Eq.~\eqref{eq:proof-thm1-last}.

To further bound the probability, we show that $\mu_s$ and $\sigma_s^2$ are exponentially small in the number of qubits. First, consider the mean
\begin{align}
    \mu_s = & \Ebb_{\thv}\left[ \sum_{\xv \in \PC} \qth(\xv)\right]  \\
    = & \sum_{\xv \in \PC} \mu(\xv) \\
    \leq & N_p {\rm max}_{\xv \in \PC} [\mu(\xv)] \;, \label{eq:proof-thm1-mean}
\end{align}
with $N_p = |\PC| \leq M\in \OC(\poly(n))$ and $\mu(\xv)$ being the average of $\qth(\xv)$ over $\thv$. As ${\rm max}_{\xv \in \PC} [\mu(\xv)] \in \OC(1/b^n)$ for some $b>1$ (due to the assumption that the fixed points are exponentially small), this leads to $\mu_s \in \OC(\poly(n)/b^n)$. 

\medskip

Then, the variance can be upper bounded as
\begin{align}
    \sigma_s^2 = & \Var_{\thv} \left[ \sum_{\xv \in \PC} \qth(\xv)\right] \\
     = &  \sum_{\xv \in \PC} \Var_{\thv}\left[\qth(\xv) \right] +  \sum_{\substack{\xv,\vec{x'} \in \PC \\ \xv\neq \vec{x'}}}\Cov_{\thv}\left[\qth(\xv), \qth(\vec{x'}) \right] \\
    \leq &  \sum_{\xv \in \PC} \Var_{\thv}\left[\qth(\xv) \right] + \sum_{\substack{\xv,\vec{x'} \in \PC \\ \xv\neq \vec{x'}}} \sqrt{ \Var_{\thv}\left[\qth(\xv) \right]  \Var_{\thv}\left[\qth(\vec{x'}) \right]} \\
    \leq & \sum_{\xv \in \PC} \Var_{\thv}\left[\qth(\xv) \right] + \sum_{\substack{\xv,\vec{x'} \in \PC \\ \xv\neq \vec{x'}}}  \frac{\Var_{\thv}\left[\qth(\xv) \right]  + \Var_{\thv}\left[\qth(\vec{x'}) \right]}{2} \\
    = & N_p \sum_{\xv \in \PC}\Var_{\thv}[\qth(\xv)] \\
    \leq & N_p^2 {\rm max}_{\xv \in \PC} [\sigma^2(\xv)] \;, \label{eq:proof-thm1-var}
\end{align}
where in the first inequality we have used Cauchy-Schwarz, and the second inequality is the inequality of arithmetic and geometric means $\sqrt{xy} \leq (x+y)/2$ for $x,y>0$. $\sigma^2(\xv)$ is the variance of $\qth(\xv)$ over $\thv$. 
Finally, assuming that $\qth(\vec{x'})$ exponentially concentrates over all bitstrings in the dataset, we have $\sigma_s \in \OC(\poly(n)/b'^n)$ for some $b'>1$.

Now, we are ready to continue from Eq.~\eqref{eq:proof-thm1-last}.
\begin{align}
    {\rm Pr}_{\Tilde{Q}_{\thv},\thv}[\Tilde{\LC}(\thv) = \LC_0(\Tilde{P}, \Tilde{Q}_{\thv})]  \geq & (1-(\mu_s + \sqrt{\sigma_s}))^N (1 - \sigma_s) \\
    \geq & (1-N(\mu_s + \sqrt{\sigma_s})) (1 - \sigma_s) \\
    \geq & 1 - (N(\mu_s + \sqrt{\sigma_s}) + \sigma_s) \\
    \geq & 1 - \left( N \left(N_p{\rm max}_{\xv \in \PC} [\mu(\xv)] + \sqrt{N_p {\rm max}_{\xv \in \PC} [\sigma(\xv)]}\right) + N_p {\rm max}_{\xv \in \PC} [\sigma(\xv)]\right) \\
    = & 1 - \delta \;,
\end{align}
where in the second inequality we use Bernoulli's inequality, the third inequality is due to dropping the positive term when expanding the product, and in the last inequality we use Eq.~\eqref{eq:proof-thm1-mean} and Eq.~\eqref{eq:proof-thm1-var}. In the last line, we denote $\delta = N \left(N_p{\rm max}_{\xv \in \PC} [\mu(\xv)] + \sqrt{N_p {\rm max}_{\xv \in \PC} [\sigma(\xv)]}\right) + N_p {\rm max}_{\xv \in \PC} [\sigma(\xv)]$ and given the polynomial scaling of $N$ and $N_p$ as well as the exponential scaling of ${\rm max}_{\xv \in \PC} [\mu(\xv)]$ and ${\rm max}_{\xv \in \PC} [\sigma(\xv)]$, we have
\begin{align}
    \delta \in \OC\left(\frac{\poly(n)}{c^n}\right) \;,
\end{align}
for some $c > 1$. This completes the proof of the theorem.
\end{proof}

\subsection{No overlap for a model without inductive bias and with polynomial support}
\label{app:pairwise-no-overlap}

Theorem~\ref{thm:explicit-loss-appendix} applies for any parameterized circuit such that the model probabilities are concentrated over the bitstrings in the training set. This can occur if the model probabilities are exponentially concentrated over \textit{all} bitstrings in $\XC$. This is typically true for any unstructured circuit, such as those given in Proposition~\ref{prop:circuit-example}. However, the requirement that the model probabilities are concentrated over the bitstrings in the training set (but not neccesarily all bitstrings) is much weaker. Crucially, this happens even when the model has polynomial support on the space of bitstrings (see Fig.~\ref{fig:nooverlap}). As an example, we show that if the model does not impose a strong inductive bias that aligns with the target distribution, it is generally unlikely that samples from the model have any overlap with the training dataset, which in turn results in the model becoming untrainable. Specifically, we prove the following Supplemental Proposition.
\begin{supplemental_proposition}[Concentration of pairwise explicit losses for models lacking inductive bias]
    Consider the scenario where the generative model has polynomial support which is chosen at random.
    The probability that these bitstrings are not in the training set $\Tilde{P}$ is exponentially close to 1. That is, with probability
    \begin{align}
        {\rm Pr}_{\Tilde{Q}_{\thv},\thv}[\Tilde{\LC}(\thv) = \LC_0(\Tilde{P}, \Tilde{Q}_{\thv})] \geq 1 - \delta' \;,\; \delta' \in \OC\left(\frac{\poly(n)}{c'^n}\right) \;,
    \end{align}
    for some $c' > 1$, and $ \LC_0(\Tilde{P}, \Tilde{Q}_{\thv})$ the statistical estimate of the loss as defined Eq.~\eqref{eq:fixedpoint}.
\end{supplemental_proposition}

\begin{proof}
Denote $N_\PC$ and $N_\QC$ as the number of non-zero probabilities in the training and model distributions, respectively, i.e., their support. Given that the model support is chosen randomly, there are $\binom{2^n - N_\PC}{N_\QC}$ ways of picking non-zero model probabilities such that the supports do not overlap. Then, the probability that there is no overlap between the two distributions is
\begin{align}
    {\rm Pr}_{\Tilde{Q}_{\thv},\thv}[\Tilde{\LC}(\thv) = \LC_0(\Tilde{P}, \Tilde{Q}_{\thv})] & = \frac{\binom{2^n - N_\PC}{N_\QC}}{\binom{2^n }{N_\QC}} \\
    & = \frac{(2^n - N_\QC)\times ... \times (2^n - N_\QC - N_\PC + 1)}{2^n\times ... \times (2^n - N_\PC+1)}\\
    & = \prod_{k=0}^{N_\PC-1} \frac{2^n - N_\QC - k}{2^n - k}\\
    & = \prod_{k=0}^{N_\PC-1}\left(1-\frac{N_\QC}{2^n-k}\right) \\
    & \geq \left(1-\frac{N_\QC}{2^n-N_\PC+1}\right)^{N_\PC} \\
    & \geq 1-\frac{N_\QC N_\PC}{2^n-N_\PC+1} \;,
\end{align}
where the first lower bound is due to taking the smallest term in the product and the last inequality is due to the Bernoulli's inequality $(1-x)^n \geq 1 - nx$. For the polynomial supports $N_\QC , N_\PC \in \OC(\poly(n))$, the lower bound becomes exponentially close to $1$, which completes the proof.

\end{proof}

This Supplemental Proposition highlights that the fundamental problem underlying exponential concentration is the miss-alignment of the model probabilities and the training data. Any randomly chosen quantum circuit ansatz with random parametrization is bound to fail with an explicit loss because of the exponentially large space of bitstrings. 

A concrete example that falls short is the case of near-identity initialization of the quantum circuit $U(\thv)$ of a QCBM, which corresponds to a near-zero initialization of the parameter vector $\thv$. While in the context of VQA-type problems it has been shown to mitigate vanishing gradients at the initial training step~\cite{grant2019initialization}, this strategy induces only significant probabilities on a polynomial number of bitstrings and thus leads to exponential concentration for general datasets. The reason is that the all-zero bitstring and bitstrings that are few bit-flips away from it have no \textit{a priori} reason to be relevant to the modelling task. 

A minimal expansion of the near-identity initialization, which appears to introduce inductive bias but does not necessarily so, is to initialize the quantum circuit model in one of the training states. That is, one sets $|\psi(\thv) \rangle = |\xv_0\rangle$ in an attempt to avoid loss concentration around the initialization. While this does in fact exhibit initial gradients towards $\qth(\xv_0) = \pt(\xv_0)$, the chance that the model then contains a non-vanishing probability on a significant number of other samples $\xv \neq \xv_0$ is still low (or likely exponentially low). A possible exception would be a dataset exhibiting cluster behaviour in the space of bitstrings. Then one could indeed initialize in and around the centroid bitstring and likely achieve improved performance across various metrics.

\subsection{Untrainability of the pairwise explicit loss function: Proof of Corollary~\ref{coro:untrain}}\label{app:pairwise-untrainability}

In this sub-section, we provide the proof of Corollary~\ref{coro:untrain}, which shows that the concentration of statistical estimate of the pairwise explicit loss makes such losses untrainable.

\begin{corollary}[Untrainability of the pairwise explicit loss function]
Under the same conditions as in Theorem~\ref{thm:explicit-loss}, the probability that the difference between the two statistical estimates of the loss function at $\thv_1$ and $\thv_2$ does not contain any information about the training distribution is exponentially close to 1. Particularly, we have
\begin{align}
    {\rm Pr}_{\Tilde{Q}_{\thv},\thv}[\Tilde{\LC}(\thv_1) - \Tilde{\LC}(\thv_2) = \Delta\LC_0(\Tilde{Q}_{\thv_1},\Tilde{Q}_{\thv_2})]\geq 1 - 2\delta \; ,
\end{align}
with $\delta \in\OC\left(\frac{\poly(n)}{c^n}\right)$ for some $c > 1$, $\Tilde{Q}_{\thv_1}$ (and $\Tilde{Q}_{\thv_2}$) is a set of sampling bitstrings obtained from the quantum generative model at the parameter value $\thv_1$ (and $\thv_2$), as well as
\begin{align}
    \Delta \LC_0(\Tilde{Q}_{\thv_1},\Tilde{Q}_{\thv_2}) = \sum_{\xv \in \QC_{\thv_1}} f(0,\Tilde{q}_{\thv_1}(\xv)) - \sum_{\xv \in \QC_{\thv_2}} f(0,\Tilde{q}_{\thv_2}(\xv)) \;,
\end{align}
with $\QC_{\thv_1}$ (and $\QC_{\thv_2}$) being a set of unique bitstrings in $\Tilde{Q}_{\thv_1}$ (and $\Tilde{Q}_{\thv_2}$). Crucially, $\Delta \LC_0(\Tilde{Q}_{\thv_1},\Tilde{Q}_{\thv_2})$ does not depend on any $\pt(\xv) \in \Tilde{P}$.
\end{corollary}
\begin{proof}
We first note that
\begin{align}
    \LC_0(\Tilde{P}, \Tilde{Q}_{\thv_1}) - \LC_0(\Tilde{P}, \Tilde{Q}_{\thv_2}) = & \left(\sum_{\xv \in \Tilde{P}} f(\pt(\xv), 0) + \sum_{\xv \in \Tilde{Q}_{\thv_1}} f(0,\Tilde{q}_{\thv_1}(\xv))\right) - \left(\sum_{\xv \in \Tilde{P}} f(\pt(\xv), 0) + \sum_{\xv \in \Tilde{Q}_{\thv_2}} f(0, \Tilde{q}_{\thv_2}(\xv)) \right) \\
    = & \sum_{\xv \in \Tilde{Q}_{\thv_1}} f(0,\Tilde{q}_{\thv_1}(\xv)) - \sum_{\xv \in \Tilde{Q}_{\thv_2}} f(0,\Tilde{q}_{\thv_2}(\xv)) \\
    = & \Delta \LC_0(\Tilde{Q}_{\thv_1},\Tilde{Q}_{\thv_2}) \;.
\end{align}

As estimating the loss function at $\thv_1$ and $\thv_2$ are two independent events for $|\Tilde{Q}_{\thv_1}|, |\Tilde{Q}_{\thv_1}| \in \OC(\poly(n))$, we have
\begin{align}
    {\rm Pr}_{\Tilde{Q}_{\thv},\thv}[\Tilde{\LC}(\thv_1) - \Tilde{\LC}(\thv_2) = \Delta\LC_0(\Tilde{Q}_{\thv_1},\Tilde{Q}_{\thv_2})] = &  {\rm Pr}_{\Tilde{Q}_{\thv},\thv}[\Tilde{\LC}(\thv_1) = \LC_0(\Tilde{P}, \Tilde{Q}_{\thv_1}) \cap \Tilde{\LC}(\thv_2) = \LC_0(\Tilde{P}, \Tilde{Q}_{\thv_2}) ] \\
    = & {\rm Pr}_{\Tilde{Q}_{\thv},\thv}[\Tilde{\LC}(\thv_1) = \LC_0(\Tilde{P}, \Tilde{Q}_{\thv_1})]\cdot {\rm Pr}_{\Tilde{Q}_{\thv},\thv}[\Tilde{\LC}(\thv_2) = \LC_0(\Tilde{P}, \Tilde{Q}_{\thv_2})] \\
    \geq & (1-\delta)(1 - \delta) \\
    \geq & 1 - 2 \delta \;,
\end{align}
where the second equality is due to the independence of two events, the first inequality is from applying Theorem~\ref{thm:explicit-loss} and the last line is from dropping the $\delta^2$ term.
\end{proof}

\section{Supplementary Note - Exact variance of KLD -- large variance is not everything}\label{ap:KLD_exact}
In this section, we show that the variance of the exact KL-divergence can be polynomially large even when model probabilities exhibit exponential concentration. Instead, the scaling of the variance depends on the number of supports of the training/target distribution. This is summarised by Proposition 1 in the main text which we repeat here for convenience. 
\setcounter{proposition}{0}
\begin{proposition}\label{prop:exact-kl-var}
Consider the KLD loss as defined in Eq.~\eqref{eq:KLD-loss}. Assume access to the exact target distribution $p(\xv)$ and the model distribution $\qt(\xv)$. Then, we have
\begin{itemize}
    \item For deep (Haar random) parametrized circuit $U(\thv)$, the variance of the loss scales asymptotically ($2^n \gg 1$) as
    \begin{align}
        \Var_{\thv}[\LC^{\rm KLD}(\thv)] = \frac{\pi^2}{6}\sum_{\xv} p^2(\xv) \;.
    \end{align}
    \item  For a random tensor product circuit $U(\thv) = \bigotimes_{i=1}^n U_i(\theta_i)$ where $U_i(\theta_i)$ is a random single-qubit unitary, the variance of the loss scales asymptotically as
    \begin{align}
        \Var_{\thv}[\LC^{\rm KLD}(\thv)] = n-\frac{\pi^2}{6}\sum_{\xv, \xv'} p(\xv) p(\xv')\norm{\xv-\xv'}_H \;, 
    \end{align}
    where $\| \cdot \|_{\rm H}$ is the Hamming distance.
\end{itemize}
\end{proposition}

\begin{proof} The variance of KLD loss can be expressed as
\begin{align}
    \Var_{\thv}[\LC^{\rm KLD}(\thv)] = & \Var_{\thv}\left[ \sum_{\xv} p(\xv) \log\left( \frac{p(\xv)}{\qth(\xv) }\right) \right] \\
    = & \Var_{\thv}\left[ \sum_{\xv} p(\xv) \log\left( \qth(\xv) \right) \right] \\
    = & \sum_{\xv} p^2(\xv) \Var_{\thv} [\log\left( \qth(\xv) \right)] + \sum_{\substack{\xv, \xv' \\  \xv \neq \xv' }} p(\xv) p(\xv') \Cov_{\thv}\left[\log\left( \qth(\xv) \right), \log\left( \qth(\xv') \right)  \right] \;, \label{eq:var-kld-proof-key}
\end{align}
where the second equality is from $\Var_{\alpha}[a A(\alpha) + b] = a^2 \Var_{\alpha} [A(\alpha)]$ with some constants $a, b$ and in the third equality we use the identity $\Var_{\alpha}[\sum_i A_i (\alpha)] = \sum_{i,j} \Cov_{\alpha}[A_i(\alpha), A_j(\alpha)] = \sum_{i} \Var_{\alpha}[A_i(\alpha)] + \sum_{i \neq j}\Cov_{\alpha}[A_i(\alpha), A_j(\alpha)]$. 

To compute each term in Eq.~\eqref{eq:var-kld-proof-key}, we consider two different types of parametrized circuits: (i) the Haar random parametrized circuit and (ii) the tensor product circuit. 

\bigskip

\noindent\textit{\underline{(i) Haar random parametrized circuit.}} Consider an ensemble of parametrized unitaries $\mathbb{U}_{\thv} = \{U(\thv) \}_{\thv}$ that forms a Haar random ensemble. In this scenario with $N_{\rm dim} = 2^n \gg 1$, the distribution of output probabilities $\qth(\xv)$ follows the so-called Porter-Thomas distribution ${\rm Pr}(q := \qth(\xv)) = N_{\rm dim} e^{-N_{\rm dim} q}$~\cite{boixo2018characterizing}. One can see that distribution is not exactly normalized $\int_{0}^1 {\rm Pr}(q) dq = \int_{0}^1 N_{\rm dim} e^{-N_{\rm dim}q} dq = 1 - e^{- N_{\rm dim}}$ and has a correction term that is exponentially small in the dimension of the Hilbert space (i.e., doubly exponential in the number of qubits $n$).

In this regime, we can perform integration over the Porter-Thomas distribution instead of over the parameter distribution to compute an average associated with the model probability. That is, for some function $f(\cdot)$, we have 
\begin{align}
    \Ebb_{\thv} [f(\qth(\xv))] = \int_{\thv} d\thv {\rm Pr}(\thv) f(\qth(\xv)) = \int_{0}^1 dq {\rm Pr}(q) f(q) \;.
\end{align}

Now, consider the first term in Eq.~\eqref{eq:var-kld-proof-key}. The average of $\log\left( \qth(\xv) \right)$ can be written as (with $u= N_{\rm dim} q$)
\begin{align}
\Ebb_{\thv}[\log(\qth(\xv))] = & \int_{0}^1 dq N_{\rm dim} e^{- N_{\rm dim} q} \log(q) \\
= & \int_{0}^{N_{\rm dim}}e^{-u}  \log(u) du - \log(N_{\rm dim}) \int_{0}^{N_{\rm dim}}e^{-u} du \\
= & \int_{0}^{\infty} e^{-u} \log(u) du - \int_{N_{\rm dim}}^{\infty} e^{-u} \log(u) du  - \log(N_{\rm dim}) (1-e^{-N_{\rm dim}}) \\
= & \int_{0}^{\infty} e^{-u} \log(u) du - \left( - e^{-u} \log(u) \bigg\rvert_{u = N_{\rm dim}}^{\infty}  + \int_{N_{\rm dim}}^{\infty} \frac{e^{-u}}{u} du \right)   - \log(N_{\rm dim}) (1-e^{-N_{\rm dim}}) \\
= & - (\gamma + \log(N_{\rm dim}) + E_1(N_{\rm dim})) \;,
\end{align}
where $\gamma = - \int_{0}^{\infty} e^{-u} \log(u) du$ is the Euler Mascheroni constant and $E_1(N_{\rm dim}) = \int_{N_{\rm dim}}^\infty \frac{e^{-u}}{u} du$ is the exponential integral. In the above derivations, the second equality is by the change of variable $u = N_{\rm dim} q$, the fourth equality is due to partial integration on the middle term. 

Similarly, the second moment of $\log\left( \qth(\xv) \right)$ is of the form (again, with $u = N_{\rm dim} q$)
\begin{align}
    \Ebb_{\thv}[\log^2(\qth(\xv))] = & \int_{0}^1 dq N_{\rm dim} e^{- N_{\rm dim} q} \log^2(q)  \\
    = & \int_{0}^{N_{\rm dim}} e^{-u} \log^2(u) du - 2 \log(N_{\rm dim})\int_{0}^{N_{\rm dim}} e^{- u} \log(u) du + \log^2(N_{\rm dim}) \int_{0}^{N_{\rm dim}} e^{-u} du \\
    = & \int_{0}^{\infty} e^{-u} \log^2(u) du  - \int_{N_{\rm dim}}^{\infty} e^{-u} \log^2(u) du - 2 \log(N_{\rm dim})\int_{0}^{N_{\rm dim}} e^{- u} \log(u) du \nonumber \\
    & + \log^2(N_{\rm dim})(1 - e^{- N_{\rm dim}})\\
    = &  \gamma^2 + \frac{\pi^2}{6} - \int_{N_{\rm dim}}^{\infty} e^{-u} \log^2(u) du - 2 \log(N_{\rm dim})\int_{0}^{N_{\rm dim}} e^{- u} \log(u) du + \log^2(N_{\rm dim})(1 - e^{- N_{\rm dim}})\;,
\end{align}
where $\int_{0}^{\infty} e^{-u} \log^2(u) du  = \gamma^2 + \frac{\pi^2}{6}$ as a result of the direct computation and $\int_{0}^{N_{\rm dim}} e^{- u} \log(u) du$ can handled in an identical way to the average shown above. For $N_{\rm dim} \gg 1$, we have $\int_{N_{\rm dim}}^{\infty} e^{-u} \log^2(u) du \approx e^{- N_{\rm dim}} \log^2(N_{\rm dim})$ which can be verified using the L'H\^{o}pital's rule
\begin{align}
    \lim\limits_{N_{\rm dim} \to \infty}
    \frac{e^{-N_{\rm dim}} \log^2 N_{\rm dim}}{\int_{N_{\rm dim}}^{\infty} e^{-u} \log^2(u) du}
    = \lim\limits_{N_{\rm dim}\to \infty} \frac{- e^{-N_{\rm dim}} \log^2 N_{\rm dim} + e^{-N_{\rm dim}} \frac{2\log N_{\rm dim}}{N_{\rm dim}}}{- e^{-N_{\rm dim}} \log^2(N_{\rm dim})} =  \lim\limits_{N_{\rm dim}\to \infty}1- \frac{2}{N_{\rm dim} \log N_{\rm dim}} = 1 \;.
\end{align}
Additionally, in this regime the exponential integral approaches (e.g., see Ref.~\cite{abramowitz1972abramowitz})
\begin{align}
    \lim\limits_{N_{\rm dim} \to \infty} E_1(N_{\rm dim}) =  \lim\limits_{N_{\rm dim} \to \infty} \int_{N_{\rm dim}}^\infty \frac{e^{-u}}{u} du \approx \frac{e^{-N_{\rm dim}}}{N_{\rm dim}} \;.
\end{align}
Hence, for $N_{\rm dim} \gg 1$, the variance of $\log(\qth(\xv))$ scales as
\begin{align}\label{eq:proof-var-kl-haar-log-term}
    \Var_{\thv} [\log(\qth(\xv))] = \frac{\pi^2}{6} \;.
\end{align}

Next in the menu, we show that the covariance terms in Eq.~\eqref{eq:var-kld-proof-key} disappear. In other words, $\qth(\xv)$ and $\qth(\xv')$ are independent and their joint distribution respects the product formulae. To see this, we follow a similar approach in Ref.~\cite{boixo2018characterizing} which derives the Porter Thomas distribution for a random quantum state. We write a quantum state in the computational basis as $|\psi \rangle = \sum_{\xv} (a_{\xv} + i b_{\xv}) |\xv \rangle$. When $|\psi\rangle$ is uniformly chosen at random from the Hilbert space, there is only the normalization constrain on the probabilities i.e., $\sum_{\xv} q(\xv) = \sum_{\xv} (a^2_{\xv} + b^2_{\xv}) = 1$. The joint distribution of probabilities $q$ and $q'$ can be derived by a ratio of the volume of the Hilbert space corresponding to states with $q$ and $q'$, $\vol_{q,q'}$ to the total volume of all states, $\vol_{\rm tot}$. More precisely, we can use delta functions to select the relevant states and compute the associated volumes and the joint distribution as
\begin{align}
    {\rm Pr}(q, q') = \frac{\vol_{q,q'}}{\vol_{\rm tot}} \;,
\end{align}
with 
\begin{align}
    \vol_{q,q'} = & \int_{-\infty}^{\infty} \prod_{\xv} d a_{\xv} d b_{\xv} \delta\left(1 - \sum_{\xv} (a^2_{\xv} + b^2_{\xv}) \right) \delta(a_{\xv_0}^2 + b_{\xv_0}^2 - q) \delta(a_{\xv_1}^2 + b_{\xv_1}^2-q') \;, \\
    \vol_{\rm tot} = & \int_{-\infty}^{\infty} \prod_{\xv} d a_{\xv} d b_{\xv} \delta\left(1 - \sum_{\xv} (a^2_{\xv} + b^2_{\xv}) \right) \;,
\end{align}
where $\delta\left(1 - \sum_{\xv} (a^2_{\xv} + b^2_{\xv}) \right)$ enforces the normalization condition, and $\delta(a_{\xv_0}^2 + b_{\xv_0}^2 - q)$ as well $\delta(a_{\xv_1}^2 + b_{\xv_1}^2-q')$ correspond to select states that include $q$ and $q'$ probabilities, respectively. By expressing a delta function as $\delta(\omega) = \frac{1}{2\pi}\int_{-\infty}^\infty dt e^{i \omega t} $, the numerator can be expressed as
\begin{align}
    \vol_{q,q'} = & \int_{-\infty}^{\infty} \prod_{\xv} d a_{\xv} d b_{\xv} \left(\frac{1}{2\pi} \int_{-\infty}^\infty dt e^{it (\sum_{\xv} (a^2_{\xv} + b^2_{\xv}) - 1) } \right) \left(\frac{1}{2\pi} \int_{-\infty}^\infty ds e^{i s (a_{\xv_0}^2 + b_{\xv_0}^2 - q) }\right) \left(\frac{1}{2\pi} \int_{-\infty}^\infty dw e^{i w (a_{\xv_1}^2 + b_{\xv_1}^2 - q') }\right) \\
    = & \frac{1}{(2\pi)^3} \left[\int_{-\infty}^\infty dt e^{-it} \left( \int_{-\infty}^\infty d a_{\xv} d b_{\xv} e^{it (a_{\xv}^2 + b_{\xv}^2)} \right)^{N_{\rm dim}-2} \right] \cdot \left[\int_{-\infty}^\infty ds e^{-isq} \int_{-\infty}^\infty da_{\xv_0}db_{\xv_0} e^{i(s+t)(a_{\xv_0}^2 + b_{\xv_0}^2)} \right] \nonumber \\
    & \;\; \cdot \left[\int_{-\infty}^\infty dw e^{-iwq'} \int_{-\infty}^\infty da_{\xv_1}db_{\xv_1} e^{i(w+t)(a_{\xv_1}^2 + b_{\xv_1}^2)} \right] \\
    = & \frac{1}{(2\pi)^3} \int_{-\infty}^\infty dt e^{-it} \left( \frac{i\pi}{t}\right)^{N_{\rm dim} - 2} \int_{-\infty}^{\infty} ds e^{-isq} \left(\frac{i \pi}{t+s}\right) \int_{-\infty}^{\infty} dw e^{-iwq'} \left(\frac{i \pi}{t+w}\right) \\
    = & \frac{(i\pi)^{N_{\rm dim}}}{(2\pi)^3} (-2\pi i)^2\int_{-\infty}^\infty dt \frac{e^{-it(1-q-q')}}{t^{N_{\rm dim} - 2}} \\
    =& \frac{(i\pi)^{N_{\rm dim}}}{(2\pi)^3} (-2\pi i)^3 \frac{(-i(1-q-q'))^{N_{\rm dim}-3}}{(N_{\rm dim}-3)!} \\
    =& \frac{\pi^{N_{\rm dim}}(1-q-q')^{N_{\rm dim}-3}}{(N_{\rm dim}-3)!} \;,
\end{align}
where we recognise Gaussian integration in the third equality i.e., $\int_{-\infty}^\infty dz e^{i\alpha z^2} = \sqrt{\frac{i\pi}{\alpha}}$. In the fourth equality, we apply the Cauchy residue theorem for the simple pole to obtain $\int_{-\infty}^\infty dz \frac{e^{-izp}}{t+z} = -2\pi i e^{itp}$. In the fifth equality, the integration of a higher order pole can be computed by adding infinitesimally small imaginary offset
\begin{align}
    \int_{-\infty}^\infty dz \frac{e^{-iz}}{z^N} = \lim\limits_{\epsilon \to 0 } \int_{-\infty + i \epsilon}^{\infty + i \epsilon} dz \frac{e^{-iz}}{z^N} = 
    \frac{-2 \pi i }{ (N-1)!} \left( \frac{d^N}{dz^N}e^{-iz} \right)\bigg\rvert_{z = 0} = \frac{-2\pi i }{(N-1)!}  \cdot (-i)^{N-1} \;,
\end{align}
We refer the readers to Ref.~\cite{boixo2018characterizing, mullane2020sampling} for more details of the complex integration.

Similarly, the total volume of all states is of the form
\begin{align}
    \vol_{\rm tot} = & \frac{1}{2\pi} \int_{-\infty}^\infty dt e^{-it} \left( \int_{-\infty}^\infty d a_{\xv} d b_{\xv} e^{it (a_{\xv}^2 + b_{\xv}^2)} \right)^{N_{\rm dim}} \\
    = & \frac{\pi^{N_{\rm dim}}}{(N_{\rm dim}-1)!} \;.
\end{align}

For large $N_{\rm dim} \gg 1$, the joint distribution of model probabilities can then be expressed as
\begin{align}
    {\rm Pr}(q, q') = (N_{\rm dim} - 1)(N_{\rm dim} - 2)(1 - q - q')^{N_{\rm dim}-3} \approx N_{\rm dim}^2 e^{-N_{\rm dim}(q + q')} = {\rm Pr}(q) {\rm Pr}(q') \;,
\end{align}
which leads to the vanishing of covariance terms in Eq.~\eqref{eq:var-kld-proof-key}
\begin{align}
    \Cov_{\thv}\left[\log\left( \qth(\xv) \right), \log\left( \qth(\xv') \right)  \right] = 0 \;.
\end{align}

Together with Eq~\eqref{eq:proof-var-kl-haar-log-term}, the variance in Eq.~\eqref{eq:var-kld-proof-key} in the regime of large $N_{\rm dim}$ can be written as
\begin{align}
    \Var_{\thv}[\LC^{\rm KLD}(\thv)] =  \frac{\pi^2}{6}\sum_{\xv} p^2(\xv) \;.
\end{align}
We remind the reader again that this result is valid in the large $N_{\rm dim}$ limit and for the uncorrelated distributions the error we are making from this approximation scales at least polynomially in $N_{\rm dim}$ i.e. exponentially with $n$. 
This completes the proof of the first part.

\bigskip

\noindent\textit{\underline{(ii) Tensor product circuit.}} Consider the parametrized circuit of the form $U(\thv)= \bigotimes_{i=1}^{n} U_i(\theta_i) $ where $U_i(\theta_i)$ is an arbitrary single-qubit unitary acting on qubit $i$. In this scenario, the model probability is a product of one-bit marginal probabilities $\qth(\xv) = \prod_{i=1}^n q_{\theta_i}(x_i)$ where $q_{\theta_i}(x_i) = \Tr[U_i(\theta_i) \rho_0^{(i)}U^\dagger_i(\theta_i) |x_i \rangle\langle x_i |]$ with $\rho_0^{(i)}$ being a reduced initial state on the qubit $i$, leading to $\log(\qth(\xv)) = \sum_{i=1}^n \log(q_{\theta_i}(x_i))$ . In addition, these one-bit marginal probabilities are independent which result in 
\begin{align}
    \Cov_{\thv}[\log(q_{\theta_i}(x_i)), q_{\theta_j}(x_j)] = 0 \;,
\end{align}
for $i\neq j$. With these, the variance can be simplified as follows 
\begin{align} 
    \Var_{\thv}[\LC^{\rm KLD}(\thv)] 
    = & \sum_{\xv} p^2(\xv) \sum_{i=1}^n \Var_{\theta_i}[\log(q_{\theta_i}(x_i))] + \sum_{\substack{\xv, \xv' \\  \xv \neq \xv' }} p(\xv) p(\xv') \sum_{i=1}^n \Cov_{\theta_i} [\log(q_i(\theta_i)(x_i)), \log(q_i(\theta_i)(x'_i))] \\
    = & \sum_{\xv, \xv'} p(\xv) p(\xv') \sum_{i=1}^n \Cov_{\theta_i}[\log\left( q_{\theta_i}(x_i) \right), \log\left( q_{\theta_i}(x'_i) \right)] \;. \label{eq:proof-kld-tensor-var}
\end{align}
Here, we consider a single-qubit unitary parametrized by Euler angles i.e., $\theta_i = (a_i, b_i, c_i) \in [0,2\pi] \times [0, \pi] \times [0, 2\pi]$ of the form~\cite{de2018simple}
\begin{align}
    U(\theta_i) = \begin{pmatrix}
 e^{i\frac{a_i + c_i}{2}} \cos \frac{b_i}{2} & e^{i\frac{a_i - c_i}{2}} \sin \frac{b_i}{2} \\
 -e^{-i\frac{a_i - c_i}{2}} \sin \frac{b_i}{2}  & e^{-i\frac{a_i + c_i}{2}} \cos \frac{b_i}{2} \;.
\end{pmatrix} \;\;.
\end{align}
This results in the one-bit marginal probability for $x_i=0,1$ as
\begin{align}
    q_{\theta_i}(0) = \cos^2\left( \frac{b_i}{2}\right) \;\;,\;\;  q_{\theta_i}(1) =\sin^2\left( \frac{b_i}{2}\right) \;,
\end{align}
and the Haar measure in this case is $d\mu_2 := {\rm Pr(\theta_i)} d\theta_i = \sin(b_i) db_i da_i dc_i$.
As $q_{\theta_i}(x_i)$ only depends on $b_i$, then we need to compute average over $b_i$ such that for any function $f(b_i)$ the average is given by $\Ebb_{b_i}[f(b_i)]= \frac{1}{2}\int_{0}^\pi db_i \sin(bi)f(b_i)$ where the factor $1/2$ is to ensure normalisation (i.e. $\Ebb_{b_i}[1]=1$). We give the following integrals that are used to compute the KLD variance
\begin{align}
    &\Ebb_{\theta_i}[\log(q_{\theta_i}(0))]=\int_{0}^\pi dx \sin(x) \log(\cos(x/2)) = -1 \;, \\
    &\Ebb_{\theta_i}[\log(q_{\theta_i}(1))]=\int_{0}^\pi dx \sin(x) \log(\sin(x/2)) = -1\;,    \\ 
      &\Ebb_{\theta_i}[\log^2(q_{\theta_i}(0))]=2\int_{0}^\pi dx \sin(x) \log^2(\cos(x/2)) = 2 \;, \\
   & \Ebb_{\theta_i}[\log^2(q_{\theta_i}(1))]=2\int_{0}^\pi dx \sin(x) \log^2(\sin(x/2)) = 2 \;, \\
   & \Ebb_{\theta_i}[\log(q_{\theta_i}(0))\log(q_{\theta_i}(1))] = 2\int_{0}^\pi dx \sin(x) \log(\cos(x/2))\log(\sin(x/2)) = 2-\frac{\pi^2}{6}\;,
\end{align}
which in turn leads to
\begin{align}
    & \Var_{\theta_i} [\log(q_{\theta_i}(0))] =  \Var_{\theta_i} [\log(q_{\theta_i}(1))] = 1 \;, \label{eq:proof-kld-tensor-var-bit} \\
    & \Cov_{\theta_i}[\log(q_{\theta_i}(0)), \log(q_{\theta_i}(1))] = 1-\frac{\pi^2}{6} \approx -0.645\label{eq:proof-kld-tensor-cov-bit} \;. 
\end{align}
Inserting this into \eqref{eq:proof-kld-tensor-var} leads to

\begin{align}
    \Var_{\thv}[\LC^{\rm KLD}(\thv)] &= \sum_{\xv, \xv'} p(\xv) p(\xv') \sum_{i=1}^n \left( \delta_{x_i,x_i'} + (1-\delta_{x_i,x_i'})(1-\frac{\pi^2}{6})\right)  \\
    &= \sum_{\xv, \xv'} p(\xv) p(\xv') \sum_{i=1}^n \left(1-\frac{\pi^2}{6}(1-\delta_{x_i,x_i'})\right) \\
    &= n-\frac{\pi^2}{6}\sum_{\xv, \xv'} p(\xv) p(\xv')\sum_{i=1}^n |x_i-x_i'| \\ 
    &= n-\frac{\pi^2}{6}\sum_{\xv, \xv'} p(\xv) p(\xv')\norm{\xv-\xv'}_H \;, 
\end{align}
where the third equality uses $|x_i-x_i'|=1-\delta_{x_i,x_i'}$ for $x_i,x_i'\in \{0,1\}$ and in the fourth equality we introduce the hamming distance between two bit-string i.e. $\sum_{i=1}^n|x_i-x_i'|=\norm{\xv-\xv'}_H$. This completes the proof of the proposition.

\end{proof}

As a final comment on the significance of Proposition~\ref{prop:exact-kl-var}, we remark that it is likely that if $\log(\qth(\xv))$ could be measured directly from quantum computers (and not as a post process of samples obtained from quantum computers), this would lead to a model with substantial (at worst polynomially vanishing) loss variances. Hence, this suggests a new (if very ambitious!) research opportunity to circumvent barren plateaus. Lastly, even with the log function measured directly and/or a direct access to the model probabilities, one could still face untrainability induced by a target distribution. For a deep parametrized circuit this happens when the purity of the target distribution becomes exponentially small i.e., $\sum_{\xv} p^2(\xv) \in \OC(b^{-n})$ for some $b>1$. The last panel in Fig~\ref{fig:var-kld-exact} numerically supports this phenomenon for the cardinality dataset which has exponentially many support for deep circuits. This presents a target-oriented barrier to trainability similar to Ref.~\cite{holmes2020barren}.

\section{Supplementary Note - Analysis on the MMD loss function}

In this section, we provide analysis on the MMD loss functions, including detailed proofs as well as further discussion of our analytical results in the main text. Specifically, Supplementary Note~\ref{app:mmd-k-body-observable} shows how the MMD loss can be viewed as the expectation of an observable and analyzes how its properties depend on the bandwidth $\sigma$. A detailed analysis of the MMD loss landscape for a tensor product QCBM is provided in Supplementary Note~\ref{app:mmd-var}. The trainability results are extended to a general Pauli rotation QCBM in Supplementary Note~\ref{appx:mmd-general-pauli}. Lastly, in Supplementary Note~\ref{app:mmd-faithfulness}, we investigate the suitability of the MMD for learning global properties of a target distribution depending on our choice in bandwidth parameter.

For convenience, we start by recalling that the MMD loss is of the form
\begin{align}
    \LC_{\rm MMD}(\thv)
    =& \sum_{\vec{x},\vec{y} \in \XC} \qth(\vec{x})\qth(\vec{y}) K(\vec{x},\vec{y})
     - 2 \sum_{\vec{x},\vec{y}  \in \XC} \qth(\vec{x})p(\vec{y}) K(\vec{x},\vec{y}) + \sum_{\vec{x},\vec{y} \in \XC} p(\vec{x})p(\vec{y}) K(\vec{x},\vec{y}) \;, 
\end{align}
with the classical Gaussian kernel
\begin{align}
    K_\sigma (\vec{x},\vec{y})  & = e^{-\frac{\| \vec{x} - \vec{y}\|^2_2}{2\sigma}} \\
    & = \prod_{i=1}^n e^{-\frac{(x_i - y_i)^2}{2\sigma}}   \;,
\end{align}
where $\|. \|_2$ is the 2-norm, $\sigma > 0$ is the so-called \textit{bandwidth} parameter, and $x_i, y_i$ are the value of bit $i$ in bitstrings $\xv, \yv$ (of length $n$), respectively.

\subsection{MMD as an observable} \label{app:mmd-k-body-observable}
In this section, we explain how the MMD loss function can be seen as an expectation value of some observable and analyse how the observable behaves for different values of the kernel bandwidth. 

We start by noting that each term in the MMD loss function can be seen as the expectation value of an observable
\begin{align} \label{eq:MMD-term-rho-rho-prime}
    \MC(\rho, \rho') = \Tr[  O^{(\sigma)}_{\rm MMD} (\rho \otimes \rho')] \;, \;
\end{align}
with the MMD observable defined as
\begin{equation}  \label{eq:MMD_operator_general_form}
    O^{(\sigma)}_{\rm MMD} := \sum_{\xv,\vec{y}}  K_{\sigma}(\xv,\yv) |\xv \rangle \langle \xv | \otimes |\yv \rangle \langle \yv | \;,
\end{equation}
which acts on $2n$ qubits. 
To obtain each term in the MMD, $\rho$ and $\rho'$ can be either the quantum state of our QCBM model $\rho_{\thv} = |\psi(\thv)\rangle\langle \psi(\thv)|$ and/or the quantum state corresponding to the training data $\rho_{\tilde{p}}$ such that $\pt(\xv) = \Tr[\rho_{\tilde{p}} |\xv \rangle \langle \xv |]$.
In particular, for computing the first MMD term, we have both  $\rho = \rho' = \rho_{\thv}$ and, for computing the cross-term, we have $\rho ' = \rho_{\tilde{p}}$ instead, and for the final term $\rho = \rho' = \rho_{\tilde{p}}$.

The MMD observable $ O^{(\sigma)}_{\rm MMD} $ can be rewritten in Pauli basis using $|\xv \rangle \langle \xv| = \bigotimes_{i=1}^n | x_i \rangle \langle x_i | =  \bigotimes_{i=1}^n \frac{1}{2}(\id_i+(-1)^{x_i}Z_i)$ for the first $n$ qubits and $|\vec{y} \rangle \langle \vec{y}| = \bigotimes_{i=1}^n | y_{i} \rangle \langle y_{i} | =  \bigotimes_{i=1}^n \frac{1}{2}(\id_{n+i}+(-1)^{y_{i}}Z_{n+i})$ for the last $n$ qubits, leading to 
\begin{align}
    O^{(\sigma)}_{\rm MMD} &=\sum_{\xv,\yv}\bigotimes_{i=1}^n\left[\left(\frac{\id_i+(-1)^{x_i}Z_i}{2}\right)\otimes\left(\frac{\id_{n+i}+(-1)^{y_{i}}Z_{n+i}}{2}\right)\exp\left(-\frac{(x_i-y_{i})^2}{2\sigma}\right)\right]\\
    &=\bigotimes_{i=1}^n\sum_{x_i,y_{i}}\left[\left(\frac{\id_i+(-1)^{x_i}Z_i}{2}\right)\otimes\left(\frac{\id_{n+i}+(-1)^{y_{i}}Z_{n+i}}{2}\right)\exp\left(-\frac{(x_i-y_{i})^2}{2\sigma}\right)\right]\\
    &=\bigotimes_{i=1}^n\left[(1 - p_\sigma)\id_i\otimes\id_{n+i}+ p_\sigma Z_i\otimes Z_{n+i}\right] \\
    &=\sum_{A\subseteq \NC}(1-p_\sigma)^{n-|A|}p_\sigma^{|A|}\bigotimes_{i\in A}(Z_i\otimes Z_{n+i}) \label{eq:mmd-obs-binom-0}\\
    & = \sum_{l = 0}^n \binom{n}{l} (1-p_\sigma)^{n-l}p_\sigma^{l} D_{2l} \;, \label{eq:mmd-obs-binom}
\end{align}
where we denote 
\begin{align} \label{eq:p_sigma_appx}
    p_\sigma = (1 - e^{-1/(2\sigma)})/2 \;,
\end{align}
ranging between $0$ and $1/2$. The second equality is obtained using $\sum_{\xv,\yv} \bigotimes_{i=1}^n h(x_i, y_{i}) = \bigotimes_{i=1}^n \sum_{x_i,y_{i}} h(x_i, y_{i})$ and the third one by explicitly summing over the $x_i$ and $y_{i}$, each of which has two possible values of $0$ and $1$. To obtain the fourth inequality we expand the tensor product out explicitly and introduce the notation $A$ to denote the set of all possible subsets of the indices $1, ..., n$. That is, 
\begin{equation} \label{eq:A}
    A\subseteq \NC = \{1,2,...,n\} \, .
\end{equation}
For example for $n=3$, we have
\begin{align}
    A \in \{ \{ \}, \{1\}, \{2\},\{3\}, \{1,2\}, \{1,3\},\{2,3\},\{1,2,3\} \}\;.
\end{align}
In the last step, we denote 
\begin{align} \label{eq:d_2l_appx}
    D_{2l} = \frac{1}{\binom{n}{l}} \sum_{\substack{A\subseteq\NC \\ |A| = l}}\, \bigotimes_{i\in A}(Z_i\otimes Z_{n+i}) \;,
\end{align}
which is a normalized sum of Pauli strings, each of which contains Pauli-Z operators of length $2l$, i.e., $2l$-body interactions.

Interestingly, $p_\sigma$ can be seen as the probability of assigning a pair of single-qubit Pauli-Z operators to a Pauli string. 
As a consequence, the coefficient $w_\sigma(l)$ follows a binomial distribution,
\begin{align}\label{eq:w_sigma_appx}
    w_\sigma(l) = \binom{n}{l} (1-p_\sigma)^{n-l}p_\sigma^{l} \;,
\end{align}
and can be interpreted as the probability of having $D_{2l}$ i.e., all possible Pauli-strings with $2l$ Pauli-Z operators. 
By adopting this Monte Carlo sampling perspective, the MMD observable $O^{(\sigma)}_{\rm MMD}$ can be constructed as a sampling average where the operator $D_{2l}$ is sampled with the probability $w_\sigma(l)$. This allows us to analyze the dominant terms in $O^{(\sigma)}_{\rm MMD}$ by using standard properties of the binomial distribution. For example, the largest $w_\sigma(l)$ (i.e. the mode of the distribution) occurs for 
\begin{align}
    (n+1)p_\sigma - 1 \leq l_{\rm max} = \argmax(w_\sigma(l)) \leq (n+1)p_\sigma \; .
\end{align}
We note that $w_\sigma(l)$ is monotonically increasing for $l < l_{\rm max}$ and is monotonically decreasing for $l > l_{\rm max}$. The average (i.e. mean) bodyness of the MMD observable is given by 
\begin{align}\label{eq:app_mean_body}
    \Ebb_{l \sim w_\sigma(l)}[ 2l ] = 2n p_\sigma \;,
\end{align}
and its variance is
\begin{align}\label{eq:app_var_body}
    \Var_{l \sim w_\sigma(l)}[2l] = 4 n p_\sigma(1 - p_\sigma)  \;.
\end{align}

\medskip

\noindent We are now ready to investigate how $O^{(\sigma)}_{\rm MMD}$ depends on $\sigma$. 

\medskip

\noindent (i) \underline{For a constant bandwidth $\sigma \in \OC(1)$:} We have the following proposition.
\setcounter{proposition}{2}
\begin{proposition}[MMD is global with a constant bandwidth]
For $\sigma\in \OC(1)$, the average bodyness of the MMD operator containing Pauli terms with weight $w_\sigma(l)$ is
\begin{equation}
     \Ebb_{l \sim w_\sigma(l)}[2l] \in \Theta(n) \, .
\end{equation}
Similarly, the variance in the bodyness is given by 
\begin{equation}
   \Var_{l \sim w_\sigma(l)}[2l] \in \Theta(n) \, .
\end{equation}
\end{proposition}

\begin{proof}
    By considering that $p_\sigma,$ $(1-p_\sigma) \in \OC(1)$ for $\sigma\in\OC(1)$ in conjunction with Eqs.~\eqref{eq:app_mean_body} and~\eqref{eq:app_var_body}, we have the scaling of the average bodyness and its variance as claimed.
\end{proof}

Together, this implies that, on average, we are likely to sample $D_{2l}$ with $l \in \OC(n)$. 
Now, we show that the contribution from low-body terms is negligible in this regime. Consider the sum of probabilities up to $k$ bodies (such that $k < l_{\rm max}$) 
\begin{align}
    \sum_{l = 0}^{k} w_\sigma(l) & = \sum_{l = 0}^k \binom{n}{l} (1 - p_\sigma)^{n - l} p_\sigma^l \\
    & \leq (k + 1) \binom{n}{k} (1 - p_\sigma)^{n - k} p_\sigma^k \\
    & \leq (k + 1) \left(\frac{n e}{k} \right)^k (1 - p_\sigma)^{n - k} p_\sigma^k \;,
\end{align}
where, in the first inequality we take the maximum value of the sum as each term in the sum is monotonically increasing before $l_{\rm max} \in \OC(n)$, the second inequality is due to $\binom{n}{k} \leq \left( \frac{ne}{k}\right)^{k}$. It is straight forward to see that, for $k \in \OC(1)$,
\begin{align}
    \sum_{l = 0}^{k} w_\sigma(l) \in \OC(\poly(n)/b^n) \;,
\end{align}
for some $b > 0$. 
Altogether, for $\sigma \in \OC(1)$, the MMD observable is global and the contribution from low-body interactions is exponentially suppressed in the number of qubits. 

\medskip

\noindent (ii) \underline{For a linearly-scaled bandwidth $\sigma \in \Theta(n)$:} The situation here is opposite to what we have in the case (i). First we note that 
\begin{align}
    p(\sigma) = & \frac{1 - \left( 1 - \frac{1}{2\sigma} + \OC\left(\frac{1}{\sigma^2} \right) \right)}{2} \\
    = & \frac{1}{4\sigma} + \OC\left( \frac{1}{\sigma^2}\right) \\
    \in & \OC\left( \frac{1}{n}\right)
\end{align}
with $\sigma\in\Theta(n)$. Thus in this limit the average and variance of the bodyness of the MMD operator are given by 
\begin{align}
    \Ebb_{l \sim w_\sigma(l)}[ 2l ]  &\in \OC(1)\\
    \Var_{l \sim w_\sigma(l)}[2l]  &\in \OC(1) \;.
\end{align} 
Intuitively, this implies the MMD observable is largely composed of low-body contributions in this bandwidth regime. As a consequence, when computing the MMD loss, the contribution from global terms are negligible. This notion is formalized in Proposition~\ref{prop:mmd-k-body} in the main text, which is proven below.

\begin{proposition}[MMD consists largely of low-body terms for $\sigma \in \Theta(n)$]
Let $\Tilde{\LC}^{(\sigma, k)}_{\rm MMD}(\thv)$ be a truncated MMD loss with a truncated operator
$\tilde{O}^{(\sigma, k)}_{\rm MMD}$ that contains up to the $2k$-body interactions in $O^{(\sigma)}_{\rm MMD}$,
\begin{align} 
     \tilde{O}^{(\sigma,k)}_{\rm MMD} :=  \sum_{l=0}^k w_{\sigma}(l) D_{2l} \; ,
\end{align}
where $w_{\sigma}(l)$ are Bernoulli-distributed weights defined in Eq.~\eqref{eq:mmd_weight}. For $\sigma \in \Theta(n)$, the difference between the exact and local approximation of the loss is bounded as
\begin{align}
    |\LC^{(\sigma)}_{\rm MMD}(\thv) - \Tilde{\LC}^{(\sigma, k)}_{\rm MMD}(\thv)| \leq  \epsilon(k) \;,
\end{align} 
with
\begin{align}
    \epsilon(k) \in \OC\left( n (c/k)^k\right) \;,
\end{align}
for some positive constant $c$. 
\end{proposition}
\begin{proof}
Let $\rho_{\thv} = |\psi(\thv) \rangle\langle \psi(\thv) |$ be the quantum state of our QCBM and $\rho_{\tilde{p}}$ with $\pt(\xv) = \Tr[\rho_{\tilde{p}} |\xv\rangle\langle\xv|]$ be the quantum associated with the training data. We then have 
\begin{align}
     |\LC^{(\sigma)}_{\rm MMD}(\thv) - \Tilde{\LC}^{(\sigma,k)}_{\rm MMD}(\thv) | = &\left|\Tr\left[ \left(O^{(\sigma)}_{\rm MMD} - \tilde{O}^{(\sigma,k)}_{\rm MMD} \right) (\rho_{\thv} \otimes \rho_{\thv})\right] - 2 \Tr\left[ \left(O^{(\sigma)}_{\rm MMD} - \tilde{O}^{(\sigma,k)}_{\rm MMD} \right) (\rho_{\thv} \otimes \rho_{\tilde{p}})\right]\right. \nonumber \\
     & \left. + \Tr\left[ \left(O^{(\sigma)}_{\rm MMD} - \tilde{O}^{(\sigma,k)}_{\rm MMD} \right) (\rho_{\tilde{p}} \otimes \rho_{\tilde{p}})\right] \right| \\ 
     \leq & \left|\Tr\left[ \left(O^{(\sigma)}_{\rm MMD} - \tilde{O}^{(\sigma,k)}_{\rm MMD} \right) (\rho_{\thv} \otimes \rho_{\thv})\right]\right| + 2 \left|\Tr\left[ \left(O^{(\sigma)}_{\rm MMD} - \tilde{O}^{(\sigma,k)}_{\rm MMD} \right) (\rho_{\thv} \otimes \rho_{\tilde{p}})\right]\right| \nonumber \\
     & + \left| \Tr\left[ \left(O^{(\sigma)}_{\rm MMD} - \tilde{O}^{(\sigma,k)}_{\rm MMD} \right) (\rho_{\tilde{p}} \otimes \rho_{\tilde{p}})\right] \right| \\
     \leq & \left\|O^{(\sigma)}_{\rm MMD} - \tilde{O}^{(\sigma,k)}_{\rm MMD}\right\|_{\infty} \| \rho_{\thv} \otimes \rho_{\thv} \|_{1} + 2\left\|O^{(\sigma)}_{\rm MMD} - \tilde{O}^{(\sigma,k)}_{\rm MMD}\right\|_{\infty} \| \rho_{\thv} \otimes \rho_{\tilde{p}} \|_{1} \nonumber \\
     & + \left\|O^{(\sigma)}_{\rm MMD} - \tilde{O}^{(\sigma,k)}_{\rm MMD}\right\|_{\infty} \| \rho_{\tilde{p}} \otimes \rho_{\tilde{p}} \|_{1} \\
     = & 4 \left\|O^{(\sigma)}_{\rm MMD} - \tilde{O}^{(\sigma,k)}_{\rm MMD}\right\|_{\infty} \\
     = & 4 \left\| \sum_{l = k + 1}^n w_\sigma(l) D_{2l} \right\|_{\infty} \\
     \leq & 4 \sum_{l = k + 1}^n \binom{n}{l} (1 - p_\sigma)^{n-l} p_{\sigma}^l \\
     \leq & 4 \sum_{l = k + 1}^n \binom{n}{l}\left(\frac{1 - e^{-1/2\sigma}}{2} \right)^l \\
     \leq &  4 \sum_{l = k + 1}^n \left( \frac{ne}{4l\sigma} \right)^l \;, \label{eq:MMD_truncation_error_upbound_general}
\end{align}
where the first inequality is due to the triangle inequality, the second inequality is due to H\"{o}lder's inequality, the second equality is that the $1$-norm of the quantum state is $1$ (density operators have trace $1$), the third inequality uses triangle inequality and the fact that the infinity norm of Pauli operators is $1$, the fourth inequality is from $1-p_\sigma \leq 1$, and in the last inequality we use $e^{-x} \geq 1 - x$ together with $\binom{n}{l} \leq \left( \frac{ne}{l}\right)^l$. 

To further upper bound the truncation error consider $f(x)=\left(\frac{ne}{4\sigma x}\right)^x$ for $x > 0$. We notice that $f'(x) = f(x)\left[\ln\left(\frac{ne}{4\sigma x} \right) - 1\right]$ which leads to the maximum of $f(x)$ at $x^*=n/(4\sigma)$. This leads to
\begin{equation}
    \sum_{l=k+1}^n\left(\frac{ne}{4\sigma l}\right)^l\leq (n-k)\left(\frac{ne}{4\sigma k^*}\right)^{k^*}\;,
\end{equation}
where $k^*=\max(k,n/4\sigma)$. Finally, if we assume that $\sigma \in \Theta(n)$  and if $k \geq n/(4\sigma)$, then we obtain
\begin{align}
\epsilon(k) \in\OC\left(n\left(\frac{c}{k}\right)^k\right)\;
\end{align}
where $c=\frac{ne}{4\sigma} \in \OC(1)$. This completes the proof.  
\end{proof}

\subsection{MMD variance for a tensor product ansatz}\label{app:mmd-var}
Here we analyse the scaling of the MMD loss variance for a tensor product ansatz. In particular, we derive Theorem~\ref{thm:mmd-sigma} and provide further discussion.

The variance of the MMD can be computed as
\begin{align}
    \Var_{\thv} [\LC_{\rm MMD}(\thv)] = & \Var_{\thv} \left[ \sum_{\vec{x},\vec{y} \in \XC} \qth(\vec{x})\qth(\vec{y}) K(\vec{x},\vec{y})
     - 2 \sum_{\vec{x},\vec{y}  \in \XC} \qth(\vec{x})p(\vec{y}) K(\vec{x},\vec{y}) + \sum_{\vec{x},\vec{y} \in \XC} p(\vec{x})p(\vec{y}) K(\vec{x},\vec{y})\right] \\
     = & \Var_{\thv} [\KC_{q, q}(\thv)] + 4 \Var_{\thv} [\KC_{p,q}(\thv)] - 4 \Cov_{\thv} [\KC_{q, q}(\thv), \KC_{p, q}(\thv)] \;, \label{eq:var-mmd-general-form}
\end{align}
where we have used $\Var[X + Y] = \Var[X] + \Var[Y] + 2 \Cov[X,Y]$ and  $\Var[X + c] = \Var[X]$ for any random variables $X, Y$ and some constant $c$. We also introduce the shorthand notation of the first and second terms in the MMD loss as
\begin{align}
    \KC_{q, q}(\thv) = \sum_{\vec{x},\vec{y} \in \XC} \qth(\vec{x})\qth(\vec{y}) K(\vec{x},\vec{y}) \;,
\end{align}
and
\begin{align}
    \KC_{p, q}(\thv) = \sum_{\vec{x},\vec{y}  \in \XC} \qth(\vec{x})p(\vec{y}) K(\vec{x},\vec{y}) \;.
\end{align}

Throughout this sub-section, we consider the QCBM that is comprised of the tensor product ansatz which is of the form
\begin{align}
    U(\thv) = \bigotimes_{i=1}^n U_i (\thv_i) \;,
\end{align}
with $U_i (\thv_i)$ being a single-qubit random unitary acting on qubit $i$ such that its ensemble of over $\thv_i$ i.e., $\{U_i (\thv_i)\}_{\thv_i}$ forms the single-qubit Haar random ensemble. The model probability of measuring a bitstring $\xv$ can be expressed as
\begin{align} 
    \qth(\xv) & = \Tr\left[ U(\thv)|\vec{0}\rangle\langle\vec{0}| U^\dagger(\thv) |\vec{x}\rangle\langle\vec{x}| \right] \\
    & = \prod_{i=1}^n \Tr\left[ U_i(\thv_i)|0_i\rangle\langle 0_i|  U_i^\dagger(\thv_i) |x_i\rangle\langle x_i|\right] \; . \label{eq:prob-model-product}
\end{align}
where we use $(A \otimes B)(C \otimes D) = AC \otimes BD$ and $\Tr[A \otimes B] = \Tr[A] \Tr[B]$.

\subsubsection{Preliminaries: Haar integration and Pauli operators}
Crucially, as the rotation angles $\thv_i$ are independent and $\{U_i(\thv_i)\}_{\thv_i}$ is a single-qubit Haar random ensemble, averaging over $\thv_i$ is equivalent to averaging over the single-qubit Haar ensemble. Hence, we can invoke Haar integration to perform an average over randomly initialized parameters $\thv_i$ on each individual qubit. As an example, consider the average of the probability $\qth(\xv)$ over single qubit Haar random product states
\begin{align} \label{eq:1-design-integration-p(x)}
    \Ebb_{\thv} [\qth(\xv)] = & \int dU(\thv)\Tr\left[ U(\thv)|\vec{0}\rangle\langle\vec{0}| U^\dagger(\thv) |\vec{x}\rangle\langle\vec{x}| \right] \\
    = & \prod_{i=1}^n \int dU_i(\thv_i)\Tr\left[ U_i(\thv_i)|0_i\rangle\langle 0_i|  U_i^\dagger(\thv_i) |x_i\rangle\langle x_i|\right] \\
    = & \prod_{i=1}^n \frac{\Tr[|0_i\rangle\langle 0_i|]\Tr[|x_i\rangle\langle x_i|]}{2} \\
    = & \frac{1}{2^n} \;,
\end{align}
where we used the Haar integral formula $\int dV V M V^\dagger = \Tr[M]/d_V$ (with $d_V$ as dimension of $V$). The Haar integration for the higher moments can be done in a similar manner. Here, we recall some useful single-qubit Haar integration formulae (see, for example, Eq.~(2.26) in Ref.~\cite{roberts2017chaos})
\begin{align} \label{eq:1-design-int-1-qubit}
\int dV V^{\otimes 1} | 0 \rangle\langle 0| ^{\otimes 1} V^{\dagger \otimes 1} &= \frac{1}{2} \id   \, \\
\label{eq:2-design-int-1-qubit}
    \int dV V^{\otimes 2} | 0 \rangle\langle 0| ^{\otimes 2} V^{\dagger \otimes 2} &= \frac{1}{6} (\id \otimes \id + S_{12})  \, \\
    \int dV V^{\otimes 3} | 0 \rangle\langle 0| ^{\otimes 3} V^{\dagger \otimes 3} &= \frac{1}{24}(\id \otimes \id \otimes\id + S_{12} + S_{13} + S_{23} + S_{23}S_{12} + S_{23}S_{13}) \;, \label{eq:3-design-int-1-qubit} \\
    \int dV V^{\otimes 4} | 0 \rangle\langle 0| ^{\otimes 4} V^{\dagger \otimes 4} &= \frac{1}{120}(\id \otimes \id \otimes\id \otimes \id + S_{12} + S_{13} + S_{14} + S_{23} + S_{24} + S_{34} + S_{34}S_{12} + S_{24}S_{13} + S_{23}S_{14} \nonumber
    \\ &+ S_{23}S_{12} + S_{24}S_{12} + S_{23}S_{13} + S_{34}S_{13} + S_{24}S_{14} + S_{34}S_{23} + S_{34}S_{24} + S_{34}S_{41}  \nonumber \\ 
    &+ S_{34}S_{23}S_{12} + S_{34}S_{24}S_{12} + S_{24}S_{23}S_{13} + S_{24}S_{34}S_{13} +S_{23}S_{34}S_{14} + S_{23}S_{24}S_{14} ) \label{eq:4-design-int-1-qubit}
\end{align}
where $S_{lk}$ is the swap operator between systems $l$ and $k$.

In addition, we will use the following lemma for the variance of an arbitrary operator $O$ in the Pauli basis over random product states.
\setcounter{lemma}{0}
\begin{lemma}\label{lemma:var-o-pauli}
Consider an arbitrary observable $O$ decomposed into the Pauli basis
\begin{align}
    O=\sum_{\sigma\in\mathfrak{p}_n}\lm_{\sigma}\sigma,
\end{align}
where the weights $\lm_\sigma$ are real constants and $\mathfrak{p}_n=\{\id,X,Y,Z\}^{\otimes n}$ is the Pauli ensemble on $n$ qubits.
The variance of $O$ over single qubit Haar random product states is given by
\begin{equation} \label{eq:product_state_variance_pauli_sum}
    \Var_{|\psi\rangle \sim \text{Haar}_1^{\otimes n}}[O]=\sum_{\sigma\in\mathfrak{p}_n\backslash {\mathbb{1}}^{\otimes n}}\frac{\lambda_\sigma^2}{3^{|s(\sigma)|}}\;,
\end{equation}
where $s(\sigma)$ is the subset of qubits on which $\sigma$ acts non trivially and $|s(\sigma)|$ is a cardinality of $s(\sigma)$.
\end{lemma}

\begin{proof}
We consider the arbitrary observable $O$ which can be deomposed into the Pauli basis as
\begin{equation} \label{eq:pauli_decomposition}
    O=\sum_{\sigma\in\mathfrak{p}_n}\lm_{\sigma}\sigma,
\end{equation}
where the weights $\lm_\sigma$ are real constants and $\mathfrak{p}_n=\{\id,X,Y,Z\}^{\otimes n}$ is the Pauli ensemble on $n$ qubits. We denote $s(\sigma)$ as a support of $\sigma$ which is a subset of qubits that $\sigma$ acts non-trivially on and $|s(\sigma)|$ as a cardinality of $s(\sigma)$\footnote{As an example, for $\sigma = X \otimes \id \otimes \id \otimes Z \otimes Y$, we have $s(\sigma) = \{ 1, 4,5 \}$ with $|s(\sigma)| = 3$.}.

The variance of $O$ over single qubit Haar random product states $|\psi\rangle \sim \text{Haar}_1^{\otimes n}$ is of the form
\begin{align}
    \Var_{|\psi\rangle \sim \text{Haar}_1^{\otimes n}}[O] & = \mathbb{E}_{\ket{\psi}\sim\text{Haar}_1^{\otimes n}}\left[ |\bra{\psi}O\ket{\psi}|^2 \right] - \left( \mathbb{E}_{\ket{\psi}\sim\text{Haar}_1^{\otimes n}}\left[ |\bra{\psi}O\ket{\psi}| \right]\right)^2 \\
    & = \langle O^2 \rangle_{|\psi\rangle \sim\text{Haar}_1^{\otimes n}} - \langle O \rangle^2_{|\psi\rangle \sim\text{Haar}_1^{\otimes n}} \;.
\end{align}

First, we consider the average of $O$ over single qubit Haar random product states.
\begin{align}
    \langle O \rangle_{|\psi\rangle \sim\text{Haar}_1^{\otimes n}} = & \sum_{\sigma\in\mathfrak{p}_n}\lm_{\sigma} \langle \sigma \rangle_{|\psi\rangle \sim\text{Haar}_1^{\otimes n}} \\
    = & \sum_{\sigma\in\mathfrak{p}_n}\lm_{\sigma} \prod_{i=1}^n \langle \sigma_i \rangle_{\ket{\psi_i} \sim \text{Haar}_1}\\
    = & \sum_{\sigma\in\mathfrak{p}_n}\lm_{\sigma} \prod_{i=1}^n \delta(\sigma_i = \id) \\
    = & \lm_{\id^{\otimes n}} \;,
\end{align}
where, in the third equality, we use the Haar integration formula in Eq.~\eqref{eq:1-design-int-1-qubit} together with the fact that all single-qubit Pauli matrices are traceless, and we denote $\delta(\sigma_i = \id) = 1$ if $\sigma_i = \id$ (otherwise, $\delta(\sigma_i = \id) = 0$ ).

\medskip

Now, we consider the second-moment of $O$ over a random product state. 
To evaluate this we follow the proof of Lemma B.3 in Appendix B in Ref.~\cite{caro2022outofdistribution} to integrate over the random product states but replace the unitary $\tilde{U}^\dagger W\hat{U}$ in Ref.~\cite{caro2022outofdistribution} with an observable $O$. This directly leads to 
\begin{align} \label{eq:second_moment_haar_product_from_ODG_paper}
    \langle O^2\rangle=\frac{1}{6^n}\sum_{A\subseteq \NC}\Tr[O_A^2]\;,
\end{align}
where $O_A=\Tr_{\bar{A}}[O]$ is the partial trace of $O$ over all qubits except those in $A\subseteq \NC = \{1, 2, ..., n\}$. We recall that $A$ is also defined in Eq.~\eqref{eq:A}.

Consider a given subset of qubits $A$. We first notice that, for a given $\sigma = \bigotimes_{i = 1}^n \sigma_i$, we have the partial trace of the Pauli string over $A$ as 
\begin{align}
    \sigma_A = & \Tr_{\bar A} \left[ \bigotimes_{i = 1}^n \sigma_i \right] \\
    = & \left[ \prod_{i \notin A} \Tr[\sigma_i] \right] \cdot\left[ \bigotimes_{i \in A} \sigma_i \right]  \\
    = & 2^{n - |A|} \delta(s(\sigma) \subseteq A) \bigotimes_{i \in A} \sigma_i  \;,
\end{align}
where we denote $\delta(s(\sigma) \subseteq A) = 1$ if $s(\sigma)  \subseteq A$ and $\delta(s(\sigma) \subseteq A) = 0$ if $s(\sigma) \nsubseteq A$, which is a direct consequence of the Pauli matrices being traceless i.e., $\Tr[\sigma_i] \neq 0$ only if $\sigma_i = \id$. Importantly, $\sigma_A \neq 0$ only if the part that $\sigma$ acts non-trivially is a subset of $A$.

Now, we consider
\begin{align}
    \Tr[O_A^2] & = \Tr_{A} \left[ \Tr_{\bar A}[O] \cdot \Tr_{\bar A}[O]\right] \\
    & = \Tr\left[ \left(\sum_{\sigma\in\mathfrak{p}_n}\lm_{\sigma} 2^{n - |A|}\delta(s(\sigma) \subseteq A) \bigotimes_{i \in A} \sigma_i\right) \left(\sum_{\sigma'\in\mathfrak{p}_n}\lm_{\sigma'}  2^{n - |A|}\delta(s(\sigma ') \subseteq A) \bigotimes_{i \in A} \sigma'_i\right)\right] \\
    & = \sum_{\sigma, \sigma' \in \mathfrak{p}_n} \lm_{\sigma}\lm_{\sigma'} 2^{2(n - |A|)} \delta(s(\sigma) \subseteq A) \delta(s(\sigma') \subseteq A) \left( \prod_{i \in A}\Tr[\sigma_i\sigma'_i] \right)\\
    & = \sum_{\sigma, \sigma' \in \mathfrak{p}_n} \lm_{\sigma}\lm_{\sigma'} 2^{2(n - |A|)} \delta(s(\sigma) \subseteq A) \delta(s(\sigma') \subseteq A)  \left( 2^{|A|} \prod_{i\in A}\delta(\sigma_i = \sigma'_i)\right)\\
    & = \sum_{\sigma \in  \mathfrak{p}_n} \lm_{\sigma}^2 2^{2n - |A|} \delta(s(\sigma) \subseteq A) \;,
\end{align}
where the third equality is due to $\Tr[\left(\bigotimes_{i \in A} \sigma_i\right)\left(\bigotimes_{i \in A} \sigma'_i\right) ] = \prod_{i \in A} \Tr[\sigma_i \sigma'_i]$, and the fourth equality is due to $\Tr[\sigma_i \sigma'_i] = 2 \delta(\sigma_i = \sigma'_i)$ with $\delta(\sigma_i = \sigma'_i)=1$ if $\sigma_i = \sigma'_i$ and $\delta(\sigma_i = \sigma'_i)=0$, otherwise. 
In the last equality, we notice that the condition that $\sigma$ acts non-trivially only on $A$ (i.e., $\delta(s(\sigma) \in A)$ together with the reduced Pauli strings on $A$ are the same for $\sigma$ and $\sigma'$ (i.e., $\prod_{i \in A}\delta(\sigma_i = \sigma'_i)$ ) implies that $\sigma = \sigma'$ for the term to be non-zero, reducing the double sum to the single sum. 

We are ready to continue with
\begin{align}
    \langle O^2\rangle = & \frac{1}{6^n} \sum_{A\subseteq \NC} \sum_{\sigma \in  \mathfrak{p}_n} \lm_{\sigma}^2 2^{2n - |A|} \delta(s(\sigma) \subseteq A) \\
    = & \left( \frac{2}{3}\right)^n \sum_{\sigma \in  \mathfrak{p}_n}  \lm_{\sigma}^2 \sum_{A\subseteq \NC} 2 ^{ - |A|} \delta(s(\sigma) \subseteq A) \\
    = & \left( \frac{2}{3}\right)^n \sum_{\sigma \in  \mathfrak{p}_n}  \lm_{\sigma}^2 \sum_{A' \subseteq \NC \backslash s(\sigma)} 2 ^{ - |s(\sigma)| - |A'|} \\
    = & \left( \frac{2}{3}\right)^n \sum_{\sigma \in  \mathfrak{p}_n}  \lm_{\sigma}^2 2^{- |s(\sigma)|}\sum_{|A'| = 0}^{n - |s(\sigma)|} \binom{n - |s(\sigma)|}{|A'|}2 ^{ - |A'|} \\
    = & \left( \frac{2}{3}\right)^n \sum_{\sigma \in  \mathfrak{p}_n}  \lm_{\sigma}^2 2^{- |s(\sigma)|} \left( \frac{3}{2}\right)^{n - |s(\sigma)|} \\
    = & \sum_{\sigma \in \mathfrak{p}_n} \frac{\lm_{\sigma}^2}{3^{|s(\sigma)|}} \;,
\end{align}
where the third equality is due to the fact that the terms do not vanish only when $s(\sigma) \in A$ and therefore we only have to sum over $A$ that contain $s(\sigma)$. The latter is equivalent to summing $A'$ where $A = A' \cup s(\sigma)$ over $\NC \backslash s(\sigma)$. In fourth equality, we replace the sum over $A'$ by a sum over $|A'|$ and counted the number of ensembles of size $|A'|$ in $\NC\backslash s(\sigma)$ (which is of size $n-|s(\sigma)|$). In the fifth equality, we recognised a binomial sum.

Lastly, we have the variance of the form
\begin{align}
    \Var_{|\psi\rangle \sim \text{Haar}_1^{\otimes n}}[O] & = \sum_{\sigma \in \mathfrak{p}_n} \frac{\lm_{\sigma}^2}{3^{|s(\sigma)|}} - \left(  \lm_{\id^{\otimes n}}\right)^2 \\
    & = \sum_{\sigma \in \mathfrak{p}_n \backslash \id^{\otimes n}} \frac{\lm_{\sigma}^2}{3^{|s(\sigma)|}} \;,
\end{align}
where the sum in the last line excludes the identity term. This completes the proof of the lemma.
\end{proof}

\subsubsection{Generic form of the MMD variance for a tensor product ansatz}

We now give a generic expression of the variance of the MMD loss for an arbitrary bandwidth, which is stated in the following proposition

\begin{supplemental_proposition}\label{sup-prop:mmd-var-tensor}
    Consider the MMD loss function $\LC^{(\sigma)}_{\rm MMD}(\thv)$ as defined in Eq.~\eqref{eq:mmd-loss-general}, which uses the classical Gaussian kernel as defined in Eq.~\eqref{eq:gaussian-kernel} with the bandwidth $\sigma$, and a quantum generative model that is comprised of a tensor-product ansatz $U = \bigotimes_i^n U_i(\theta_i)$ with $\{ U_i(\theta_i) \}_{\theta_i}$ a single-qubit Haar random ensemble for all $i$. Given a training dataset $\Tilde{P}$, we have that the variance of the MMD loss over parameters $\thv$ is
    \begin{align}\label{eq:var-mmd-tensor}
        \Var_{\thv}[\LC^{(\sigma)}_{\rm MMD}(\thv)] = B_\sigma + 4 C_\sigma(\Tilde{P}) \;,
    \end{align}
    with
    \begin{align} \label{eq:MMD_First_Term_Variance}
        B_\sigma =  \left[ \frac{7+6 e^{-1/2\sigma}+2e^{-1/\sigma}}{15}\right]^n - \left[\frac{4+4e^{-1/2\sigma}+e^{-1/\sigma}}{9}\right]^n \;,
    \end{align}
    and
    \begin{align}\label{eq:MMD_cross_terms_variance_with_Pauli_Decomposition}
         C_\sigma(\tilde{P}) =  \sum_{\substack{A \subseteq \NC \\ A\neq \{\}}} (1-p_\sigma)^{2(n-|A|)}\left( \frac{p^2_\sigma}{3}\right)^{|A|} z^2_A(\tilde{P})\;,
    \end{align}
    where $p_\sigma = (1 - e^{-\frac{1}{2\sigma}})/2$, $\NC = \{1,2,...,n\}$, $z_A(\tilde{P}) = \Tr\left[ (\bigotimes_{i \in A} Z_i )\rho_{\tilde{p}} \right] = \sum_{\yv} \pt(\yv) (-1)^{\sum_{i \in A}y_i }$ with $\rho_{\tilde{p}}$ being the quantum state corresponding to the training data such that $\pt(\xv) = \Tr[\rho_{\tilde{p}}|\xv\rangle\langle \xv|]$.
    The sum in Eq.~\eqref{eq:MMD_cross_terms_variance_with_Pauli_Decomposition} is over all possible subsets of $\NC$ excluding the empty set $\{ \}$. 
\end{supplemental_proposition}

We remark that $B_\sigma$ and $C_\sigma(\Tilde{P})$ are the variances of the first term $\KC_{q,q}(\thv)$ and the second term $\KC_{p,q}(\thv)$ in the MMD, respectively, while we found the covariance term to vanish. The dependence on the training data is encoded in $z_A(\Tilde{P})$ which ranges between $-1$ and $1$.
Lastly, the exact formula of the MMD variance has been found to be consistent with the numerical simulation up to $n=1000$ in Fig.~\ref{fig:mmd_var_training}.

\begin{proof}
There are three main steps in our proof. (i) computing the variance of the first term $\KC_{q,q}(\thv)$, (ii) computing the variance of the second term $\KC_{p,q}(\thv)$ and, lastly, (iii) showing that the covariance between the two terms is zero.
    
\medskip

(i) \underline{Computing the variance of $\KC_{q,q}(\thv)$:}
\begin{align}
    \Var_{\thv}[\KC_{q,q}(\thv)] =& \Ebb_{\thv}[\KC^2_{q,q}(\thv)] - (\Ebb_{\thv}[\KC_{q,q}(\thv)])^2 \\
    = & \Ebb_{\thv}\left[ \left( \sum_{\vec{x},\vec{y} } \qth(\vec{x})\qth(\vec{y}) K(\vec{x},\vec{y})\right)^2\right] -   \left(\Ebb_{\thv}\left[  \sum_{\vec{x},\vec{y} } \qth(\vec{x})\qth(\vec{y}) K(\vec{x},\vec{y})\right]\right)^2 \\
    = & \sum_{\xv,\yv,\vec{x'},\vec{y'}} \Ebb_{\thv}[\qth(\xv)\qth(\yv)\qth(\vec{x'})\qth(\vec{y'})]K(\xv,\yv)K(\vec{x'},\vec{y'}) - \left( \sum_{\xv,\yv} \Ebb_{\thv}[\qth(\xv)\qth(\yv)] K(\xv,\yv)\right)^2 \;. \label{eq:pf-var-1-term}
\end{align}
We now can express each individual model probability as in Eq.~\eqref{eq:prob-model-product} and then average over the parameters $\thv$. This requires us to perform Haar integration for the first and second terms in Eq.~\eqref{eq:pf-var-1-term}, respectively. 

First, consider
\begin{align}
    \sum_{\xv, \yv} \Ebb_{\thv} [\qth(\xv) \qth(\yv)] K(\xv,\yv) 
    & =  \sum_{\xv, \yv} \prod_{i=1}^n  \int dU_i(\thv_i) \Tr\left[ (U_i(\thv_i))^{\otimes 2}|0_i\rangle\langle 0_i|^{\otimes 2} (U_i^\dagger(\thv_i))^{\otimes 2} (|x_i\rangle\langle x_i|\otimes |y_i\rangle\langle y_i| ) \right]  K(\xv,\yv) \\
    & =  \sum_{\xv, \yv} \prod_{i=1}^n   \Tr\left[ \left(\int dU_i(\thv_i) (U_i(\thv_i))^{\otimes 2}|0_i\rangle\langle 0_i|^{\otimes 2} (U_i^\dagger(\thv_i))^{\otimes 2}\right) (|x_i\rangle\langle x_i|\otimes |y_i\rangle\langle y_i| ) \right]  K(\xv,\yv) \\
    & = \sum_{\xv, \yv} \prod_{i=1}^n \Tr\left[ \left( \frac{\id \otimes \id + S_{12}}{6}\right)(|x_i\rangle\langle x_i|\otimes |y_i\rangle\langle y_i| )  \right] e^{- \frac{(x_i - y_i)^2}{2\sigma}} \\
    & = \sum_{\xv, \yv} \prod_{i=1}^n \left( \frac{1 + \delta_{x_i,y_i}}{6}\right)e^{- \frac{(x_i - y_i)^2}{2\sigma}} \\
    & = \sum_{\xv} \prod_{i=1}^n \left[\left( \frac{1 + \delta_{x_i,0}}{6}\right)e^{- \frac{(x_i)^2}{2\sigma}} +  \left( \frac{1 + \delta_{x_i,1}}{6}\right)e^{- \frac{(x_i - 1)^2}{2\sigma}}\right] \\
    & =  \left( \frac{2 + e^{-1/2\sigma}}{3}\right)^n \;,
\end{align}
where, in the third equality, we use Eq.~\eqref{eq:2-design-int-1-qubit}, and, in the fifth equality as well as in the last equality, we use the identity $\sum_{\xv}\prod_{i=1}^n h_i (x_i) = \prod_{i=1}^n(h_i (0) + h_i (1))$. 

\medskip

Similarly, the first term in Eq.~\eqref{eq:pf-var-1-term} can be computed via Haar integration using Eq.~\eqref{eq:4-design-int-1-qubit} and repeatedly applying the identity $\sum_{\xv}\prod_{i=1}^n h_i (x_i) = \prod_{i=1}^n(h_i (0) + h_i (1))$, leading to
\begin{align}
    \sum_{\xv,\yv,\vec{x'},\vec{y'}} \Ebb_{\thv}[\qth(\xv)\qth(\yv)\qth(\vec{x'})\qth(\vec{y'})]K(\xv,\yv)K(\vec{x'},\vec{y'})  = \left( \frac{7+ 6 e^{-1/2\sigma} + 2 e^{-1/\sigma}}{15}\right)^n \;.
\end{align}

Altogether, we have the variance of the first MMD term as
\begin{align}
    \Var_{\thv}[\KC_{q,q}(\thv)] = \left( \frac{7+ 6 e^{-1/2\sigma} + 2 e^{-1/\sigma}}{15}\right)^n - \left( \frac{4+ 4 e^{-1/2\sigma} + e^{-1/\sigma}}{9}\right)^n \;. \label{eq:var-mmd-tensor-first}
\end{align}

(ii) \underline{Computing the variance of $\KC_{p,q}(\thv)$:}
There are two alternative ways of doing this, leading to two equivalent expressions of the variance of $\KC_{p,q}(\thv)$. First is the same approach used in (i). Alternatively, we can interpret the middle term as an expectation value of an observable $O^{(\sigma)}_{\Tilde{p}} (\Tilde{P}) = \sum_{\xv} \lambda_{\xv}(\Tilde{P}) |\xv\rangle\langle\xv|$ with $\lambda_{\xv}(\Tilde{P}) = \sum_{\yv} p(\yv) K(\xv, \yv)$ and then use Lemma~\ref{lemma:var-o-pauli}. That is, $\KC_{p,q}(\thv) = \Tr [U(\thv) |\vec{0} \rangle\langle \vec{0} | U^\dagger(\thv) O^{(\sigma)}_{\Tilde{p}}(\Tilde{P}) ]$. Transforming $O^{(\sigma)}_{\Tilde{p}}$ into the Pauli basis with $|\xv\rangle\langle\xv| = \bigotimes_{i=1}^n |x_i\rangle\langle x_i | = \bigotimes_{i=1}^n \frac{1}{2}(\id_i + (-1)^{x_i} Z_i)$ leads to
\begin{align}\label{eq:obs-mmd-middle-pauli}
    O^{(\sigma)}_{\Tilde{p}} = \sum_{A \subseteq \NC }(1-p_\sigma)^{n - |A|}p_\sigma^{|A|} z_A(\Tilde{P}) \bigotimes_{i\in A}Z_i \;,
\end{align}
where $\NC=\{1,2,...,n\}$, $p_\sigma = (1 - e^{-\frac{1}{2\sigma}})/2$ and we denote 
\begin{align} \label{eq:z_a_def}
    z_A(\tilde{P}) & = \Tr[\left(\bigotimes_{i\in A} Z_i \right) \rho_{\tilde{p}}] \\
    & = \sum_{\yv} \pt(\yv) (-1)^{\sum_{i \in A}y_i } \;,
\end{align}
where $\rho_{\tilde{p}}$ is the quantum state associated with the training data with $\pt(\xv) = \Tr[\rho_{\tilde{p}} |\xv\rangle\langle\xv|]$~\footnote{Equivalently, one can get to this reduced MMD observable by tracing out half of the qubits that $\rho_{\tilde{p}}$ acts on in $O^{(\sigma)}_{\rm MMD}$ in Eq.~\eqref{eq:mmd-obs-binom}. That is, $O^{(\sigma)}_{\Tilde{p}} = \Tr_1\left[(\id \otimes \rho_{\Tilde{p}})O^{(\sigma)}_{\rm MMD}\right]$.}. 

By using Lemma~\ref{lemma:var-o-pauli}, the variance of the middle term with $O^{(\sigma)}_{\Tilde{p}}$ expressed in the Pauli basis is of the form
\begin{align}
    \Var_{\thv}[\KC_{p,q.}(\thv)] = \sum_{\substack{A \subseteq \NC \\  A\neq \{\}}} (1-p_\sigma)^{2(n-|A|)}\left( \frac{p^2_\sigma}{3}\right)^{|A|} z^2_A(\tilde{P}) \;. \label{eq:var-mmd-tensor-second}
\end{align}
Notice that the sum now excludes the empty set $\{ \}$. 
We can see that $z_A(\tilde{P})$ encodes information about the target distribution.

\medskip

(iii) \underline{Computing the covariance between $\KC_{p,p}(\thv)$ and $\KC_{p,q}(\thv)$:} By direct computation as in (i), we have

\begin{align}
    \Cov_{\thv}[\KC_{q,q}(\thv), \KC_{p,q}(\thv)] & = \Ebb_{\thv}[\KC_{q,q}(\thv)\KC_{p,q}(\thv)] - \Ebb_{\thv}[\KC_{q,q}(\thv)] \Ebb_{\thv}[\KC_{p,q}(\thv)] \\
    & = \sum_{\xv,\yv,\vec{x'},\vec{y'}} p(\vec{y'}) \bigg( \Ebb_{\thv}[\qth(\xv)\qth(\yv)\qth(\vec{y'})] - \Ebb_{\thv}[\qth(\xv)\qth(\yv)] \Ebb_{\thv}[\qth(\vec{x'})] \bigg) K(\xv,\yv)K(\vec{x'},\vec{y'})  \\
    & = 0 \;, \label{eq:cov-mmd-tensor}
\end{align}
where the last equality follows from
\begin{align}
   \sum_{\xv,\yv}\bigg( \Ebb_{\thv}[\qth(\xv)\qth(\yv)\qth(\vec{y'})] - \Ebb_{\thv}[\qth(\xv)\qth(\yv)] \Ebb_{\thv}[\qth(\vec{x'})] \bigg) K(\xv,\yv) =0
\end{align}
which holds for any $\vec{x'}$ and $\vec{y'}$ from Eq.~\eqref{eq:1-design-integration-p(x)}, Eq.~\eqref{eq:2-design-int-1-qubit} and Eq.~\eqref{eq:3-design-int-1-qubit}.
\medskip

By substituting Eq.~\eqref{eq:var-mmd-tensor-first}, Eq.~\eqref{eq:var-mmd-tensor-second} and Eq.~\eqref{eq:cov-mmd-tensor} back into the MMD variance expression in Eq.~\eqref{eq:var-mmd-general-form}, the proof is completed.
\end{proof}

\subsubsection{Variance scaling and trainability of MMD}\label{app:mmd_var_theorem}

We now analyze how the scaling of the variance depends on the bandwidth $\sigma$. 
To demonstrate the presence of loss concentration, it is sufficient to show that the variance of the whole MMD loss has an exponentially small upper bound. We show that this happens when the bandwidth is constant and independent of the number of qubits, i.e., $\sigma \in \OC(1)$.
On the other hand, to establish trainability it is crucial to accurately measure all individual terms in the MMD loss. More precisely, we require both first and second MMD terms to have at least a polynomially large variance. We argue that this can be achieved by using bandwidth that scales as $\sigma \in \Theta(n)$. 

\medskip

A formal version of Theorem~\ref{thm:mmd-sigma} is stated below. 
\begin{theorem}[Product ansatz trainability of MMD, formal]\label{thm:mmd-sigma-formal}
Consider the MMD loss function $\LC^{(\sigma)}_{\rm MMD}(\thv)$ as defined in Eq.~\eqref{eq:mmd-loss-implicit}, which uses the classical Gaussian kernel as defined in Eq.~\eqref{eq:gaussian-kernel} with the bandwidth $\sigma>0$, and a quantum circuit generative model that is comprised of a tensor-product ansatz $U = \bigotimes_i^n U_i(\theta_i)$ with $\{ U_i(\theta_i) \}_{\theta_i}$ being single-qubit (Haar) random unitaries. Given a training dataset $\Tilde{P}$, the asymptotic scaling of the variance of the MMD loss depends on the value of $\sigma$. 

\medskip
    \noindent For $\sigma \in \OC(1)$, we have
    \begin{align}
        \Var_{\thv}[\LC^{(\sigma)}_{\rm MMD}(\thv)] \in \OC(1/b^n) \;,
    \end{align}
    with some $b>1$.

\medskip
    
    \noindent On the other hand, according to Supplemental Proposition~\ref{sup-prop:mmd-var-tensor} and for $\sigma \in \Theta(n)$, we have
    \begin{align}
        \Var_{\thv}[\LC^{(\sigma)}_{\rm MMD}(\thv)] = B_\sigma + 4C_\sigma(\Tilde{P})  \;, 
    \end{align}
    where the variance of the first term $B_{\sigma}$ is lower-bounded as 
    \begin{align}
        B_{\sigma} \in \Omega(1/n) \;,
    \end{align}
    as well as, the variance of the second term $C_{\sigma}(\Tilde{P})$ is lower-bounded as 
    \begin{align}
        C_{\sigma}(\Tilde{P}) \in \Omega(1/\poly(n)) \;,
    \end{align}
    provided that
    \begin{align}
        \sum_{\substack{A \subseteq \NC \\  A\neq \{\}, |A| \leq k }} z^2_A(\Tilde{P}) \in \Omega(1/\poly(n)) \;,
    \end{align}
    with $k \in \OC(1)$, $\NC = \{1,...,n\}$ and $z_{A}(\Tilde{P}) = \sum_{\yv} \pt(\yv) (-1)^{\sum_{i \in A}y_i}$ which encodes the information about the training data. 
\end{theorem}

\begin{proof}
We consider the scaling of the MMD variance in Eq.~\eqref{eq:var-mmd-tensor} for two scenarios of the bandwidth values. 

\medskip

(i) \underline{For $\sigma\in \OC(1)$:} We show that both $B_\sigma$ and $C_\sigma(\Tilde{P})$ in the MMD variance are exponentially small. First, from Eq.~\eqref{eq:MMD_First_Term_Variance} we know that
\begin{align}
     B_\sigma \leq  \left[ \frac{7+6 e^{-1/2\sigma}+2e^{-1/\sigma}}{15}\right]^n \;.
\end{align}
If $\sigma\in\OC(1)$, then $\tilde{B}_\sigma$ is exponentially decreasing as $e^{-1/2\sigma}, e^{-1/\sigma} \in \OC(1)$.

For the second MMD term, Eq.~\eqref{eq:MMD_cross_terms_variance_with_Pauli_Decomposition}, we upper bound as
\begin{align}
     C_\sigma(\tilde{P}) &\leq \sum_{\substack{A \subseteq \NC \\  A\neq \{\}}} (1-p_\sigma)^{2(n-|A|)}\left( \frac{p^2_\sigma}{3}\right)^{|A|} \\
     & =  \sum_{|A|=1}^n  \binom{n}{|A|} (1-p_\sigma)^{2(n-|A|)}\left( \frac{p^2_\sigma}{3}\right)^{|A|}\\
     & \leq \sum_{|A|= 0}^n  \binom{n}{|A|} (1-p_\sigma)^{2(n-|A|)}\left( \frac{p^2_\sigma}{3}\right)^{|A|}\\
     & = \left( (1-p_\sigma)^2 + \frac{p_\sigma^2}{3}\right)^n \\
     & = \left[ \frac{1+e^{-1/2\sigma} + e^{-1/\sigma}}{3}\right]^n \;,
\end{align}
where the first inequality is due to $z^2_A(\tilde{P}) \leq 1$ , the first equality leverages that the expression only depends on the support of $A$ rather than $A$ itself, the final inequality is due to including the empty set in the sum (i.e., $|A| = 0$), and in the second equality we use the binomial sum formula. In the final line we use the definition of $p_\sigma$ from Eq.~\eqref{eq:p_sigma_appx}. Similarly to the first term, if $\sigma \in \OC(1)$, the upper bound decays exponentially with $n$. 

Therefore, when $\sigma \in \OC(1)$, the variance of the MMD loss scales as
\begin{align}
    \Var_{\thv} [\LC_{\rm MMD}(\thv)] \in \OC(1/b^n) \; ,
\end{align}
for some $b > 0$.

\medskip

(ii) \underline{For $\sigma\in \Theta(n)$:}  
Consider the variance of the first MMD term
\begin{align}
    B_\sigma = & \left[ \frac{7+ 6 e^{-1/2\sigma} + 2 e^{-1/\sigma}}{15}\right]^n - \left[ \frac{4+ 4 e^{-1/2\sigma} + e^{-1/\sigma}}{9}\right]^n \\
    \geq & \left[ \frac{7+ 6 \left( 1 - \frac{1}{2\sigma} + \frac{1}{8 \sigma^2} - \frac{1}{48 \sigma^3}\right) + 2 \left( 1 - \frac{1}{\sigma} + \frac{1}{2\sigma^2} - \frac{1}{6\sigma^3}\right)}{15}\right]^n 
    - \left[ \frac{4+ 4 \left( 1 - \frac{1}{2\sigma} + \frac{1}{8 \sigma^2} \right) + \left( 1 - \frac{1}{\sigma} + \frac{1}{2\sigma^2} \right)}{9}\right]^n \\
    = & \left[ 1 - \frac{1}{3\sigma} +\frac{7}{60 \sigma^2} - \frac{11}{360 \sigma^3} \right]^n - \left[ 1 - \frac{1}{3\sigma} +\frac{1}{9 \sigma^2}  \right]^n  \\
    = & \left(1 - \frac{1}{3\sigma} \right)^n \left[ \left[ 1 + \frac{\frac{7}{60 \sigma^2} - \frac{11}{360 \sigma^3}}{1 - \frac{1}{3\sigma}} \right]^n -  \left[ 1 + \frac{\frac{1}{9 \sigma^2}}{1 - \frac{1}{3\sigma}} \right]^n\right] \\
    \geq & \left(1 - \frac{1}{3\sigma} \right)^{n-1} \left( \frac{n}{180\sigma^2}\right) \left( 1 - \frac{11}{2\sigma} \right) \\
    = &  \left( \frac{n}{180\sigma^2}\right) \left( 1 - \frac{11}{2\sigma} \right) \left[\left(1 - \frac{1}{3\sigma} \right)^{3\sigma}\right]^{(n-1)/3\sigma} \;,
\end{align}
where in the first inequality we use $1 - x + x^2/2 - x^3/6 \leq e^{-x} \leq 1 - x + x^2/2 $, the second inequality is due to $(1 + a)^n - (1 + b)^n \geq n(a-b)$ for positive $a,b$ and $a>b$. This is satisfied when $\sigma > 5.5$, which is the case for sufficiently large $n$. We note that $(n-1)/3\sigma \in \OC(1)$. To proceed further, we consider the following lemma.
\begin{lemma}\label{lemma:bound-1/e}
    The lower bound of $f(x) = (1 - 1/x)^x$ with $1 < |x|$ is given by
    \begin{align}
        f(x) \geq \frac{1}{e} \left( 1 - \sum_{j = 1}^\infty \frac{1}{(j+1)x^j} \right) \;.
    \end{align}
\end{lemma}
\begin{proof}
We consider
    \begin{align}
    f(x) & = \exp(x \ln(1 - 1/x)) \\
    & = \exp\left(x \left(\sum_{j=1}^\infty -\frac{1}{j x^j}\right)\right)\\
    & = \frac{1}{e} \cdot  \exp\left(\sum_{j=1}^\infty -\frac{1}{(j+1) x^j}\right) \\
    & \geq \frac{1}{e} \left( 1 - \sum_{j=1}^\infty \frac{1}{(j+1) x^j}\right) \;,
\end{align}
where the second equality is due to the Taylor expansion of $\ln(1 - 1/x)$ which converges for $1 < |x|$, the inequality is by using $e^{-y} \geq 1 - y$.
\end{proof}
By using Lemma~\ref{lemma:bound-1/e}, we have the following lower bound
\begin{align}
    B_\sigma & \geq \left( \frac{n}{180\sigma^2}\right) \left( 1 - \frac{11}{2\sigma} \right) \left[\frac{1}{e}\left( 1 - \sum_{j=1}^\infty \frac{1}{(j+1) (3\sigma)^j}\right)\right]^{(n-1)/3\sigma} \;,
\end{align}
which implies that $B_\sigma \in \Omega(1/n)$ for $\sigma \in \Omega(n)$

\medskip

Similarly, for the second term, we have
\begin{align}
    C_\sigma(\tilde{P}) & = \sum_{\substack{A \subseteq \NC \\ ; A\neq \{\}}} (1-p_\sigma)^{2(n-|A|)}\left( \frac{p^2_\sigma}{3}\right)^{|A|} z^2_A(\tilde{P}) \\
    & = \left[\frac{1+2e^{-1/2\sigma}+e^{-1/\sigma}}{4}\right]^n \sum_{\substack{A \subseteq \NC \\ ; A\neq \{\}}}\left(\frac{\tanh^2(1/4\sigma)}{3}\right)^{|A|}z^2_A(\tilde{P}) \\
    & \geq  \left[\frac{1+2\left( 1 - \frac{1}{2\sigma} \right)+\left( 1 - \frac{1}{\sigma}\right)}{4}\right]^n \sum_{\substack{A \subseteq \NC \\ ; A\neq \{\} }} \left( \frac{\frac{1}{16\sigma^2}\left(1- \frac{1}{24\sigma^2}\right)}{3}\right)^{|A|}  z_A^2(\Tilde{P}) \\
    & \geq \left[ \left( 1 - \frac{1}{2\sigma} \right)^{2\sigma}\right]^{\frac{n}{2\sigma}} \left( \frac{1}{48\sigma^2}\left(1- \frac{1}{24\sigma^2}\right)\right) ^{k} \sum_{\substack{A \subseteq \NC \\ ; A\neq \{\}, |A| \leq k }} z^2_A(\Tilde{P}) \\
    & \geq \left[ \frac{1}{e}\left( 1 - \sum_{j=1}^\infty \frac{1}{(j+1) (2\sigma)^j}\right)\right]^{\frac{n}{2\sigma}} \left( \frac{1}{48\sigma^2}\left(1- \frac{1}{24\sigma^2}\right)\right) ^{k} \sum_{\substack{A \subseteq \NC \\ ; A\neq \{\}, |A| \leq k }} z^2_A(\Tilde{P}) \; ,
\end{align}
where $\tanh(1/4\sigma) = p_\sigma/(1 - p_\sigma)$, the first inequality is due to $e^{-x} \geq 1 - x$ and $\tanh(x) \geq x - x^3/3$ (for positive $x$), and in the second inequality we truncate the sum at $|A| = k$ and taking $|A| = k$ for all terms within the truncated sum. Finally, in the last inequality, we note that $\frac{n}{2\sigma} \in \OC(1)$ and use Lemma~\ref{lemma:bound-1/e}.    

Altogether, taking $k \in \OC(1)$ and assuming 
\begin{align}
    \sum_{\substack{A \subseteq \NC \\ ; A\neq \{\}, |A| \leq k }} z^2_A(\Tilde{P}) \in \Omega(1/\poly(n)) \; ,
\end{align}
the MMD variance is lower bounded as
\begin{align}
    \Var_{\thv} [\LC_{\rm MMD}(\thv)] \in \Omega(1/n)\;,
\end{align}
with the desired scaling of $B_\sigma$ and $C_\sigma(\Tilde{P})$. This completes the proof of the theorem.
    
\end{proof}

\begin{figure}[h]
    \centering
    \includegraphics[width=0.55\linewidth]{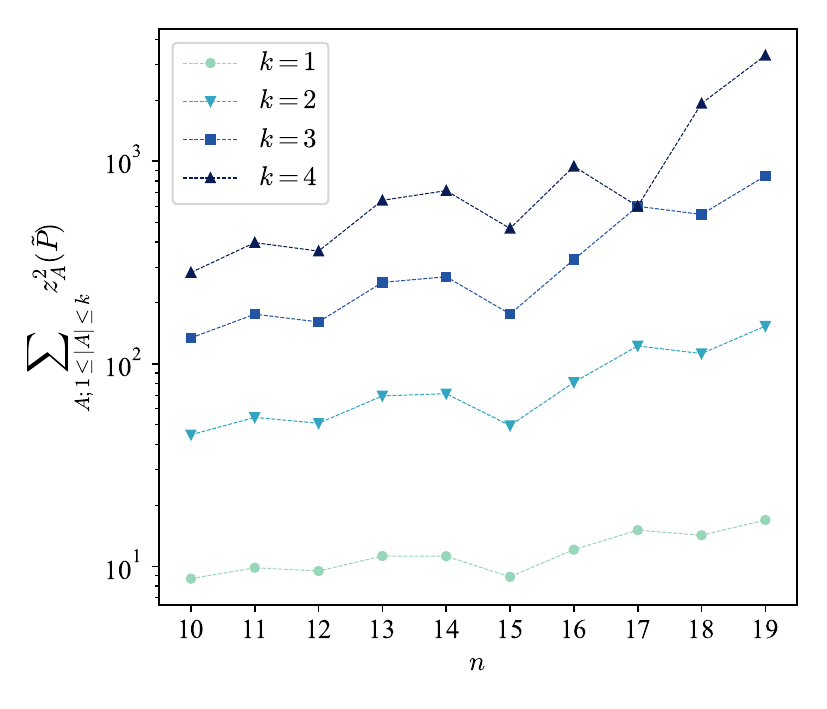}
    \caption{We plot $\sum_{A \subseteq \NC} z_A^2(\Tilde{P})$ such that $|A| \leq k$ and $A\neq \{ \}$ as a function of $n$ and shows different $k$ values. The datasets considered here corresponds to the real high energy physics dataset used in Sec.~\ref{sec:train-hea-dataset}. Values of $n=10,11,...,19$ and $k=1,2,3,4$ are presented. The quantity does not vanish with the increasing number of qubits, satisfying the data-dependence assumption made in Theorem~\ref{thm:mmd-sigma-formal}.}
    \label{fig:avg_za_cern}
\end{figure}

We showed that the variance of the MMD cross terms $C_\sigma(\Tilde{P})$ decays at most polynomially in $n$ provided that the sum of $z_A^2(\tilde{P})$ terms for $|A|\in \OC(1)$ and $A\neq \{\,\}$ is at least polynomially small in $n$, i.e., not exponentially small. Here, we comment on this assumption.

First, we recall the definition of $z_{A}(\Tilde{P})$ from Eq.~\eqref{eq:z_a_def}
\begin{align} 
    z_A(\tilde{P}) & = \Tr[\left(\bigotimes_{i\in A} Z_i \right) \rho_{\tilde{p}}] \\
    & = \sum_{\yv} \pt(\yv) (-1)^{\sum_{i \in A}y_i } \;,
\end{align}
with $\rho_{\tilde{p}}$ being the quantum state associated with the training data distribution $\pt(\xv) = \Tr[\rho_{\tilde{p}} |\xv\rangle\langle\xv|]$. 
Importantly, $z_{A}(\Tilde{P})$ encodes the correlation of training data on the subset $A$ of the bitstring (with $A$ defined in Eq.~\eqref{eq:A}), which can be interpreted as an average parity over $A$. In other words, its purpose is for the model to learn the same expectation on the operator $Z_A$ as the dataset.

The magnitude of $\sum_{A} z^2_{A}(\Tilde{P})$ depends on the provided training dataset, and if the variance of the cross term vanishes due to the data-dependence, this is because it is required for a faithfulness of the loss function. We note, however, that we do not expect this to occur in practice because partical datasets are likely exhibit significant correlations, and also due to the assumption of polynomial dataset sizes, i.e., $\pt(\xv) \in \Omega(1/\poly(n))$ for $\xv \in \Tilde{P}$. 
To emphasize this point, we analyzed the dataset from HEP colliders experiments used throughout Section~\ref{sec:train-hea-dataset}. Fig.~\ref{fig:avg_za_cern} numerically shows the sum of the first $z^2_{A}(\Tilde{P})$ terms for $1\leq |A|\leq k$ as a function of $n$ and for $k=1,\, 2,\, 3$ and $4$. We see that these terms for the low-body interactions do not disappear as the number of qubits increases and, in fact, they increase with $n$ instead. This provides a practical example that the data-dependence assumption used in Theorem~\ref{thm:mmd-sigma-formal} is satisfied in practice.

\subsection{MMD variance for a generic Pauli rotation ans\"{a}tze}\label{appx:mmd-general-pauli}
In this sub-section, we present a formal version of Theorem~\ref{thm:mmd-train-general} and provide details of the proof. 
First, we recall that the MMD observable in Eq.~\eqref{eq:mmd-obs-binom-0} is of the form
\begin{align}
    O^{(\sigma)}_{\rm MMD} & = \sum_{A\subseteq \NC}(1-p_\sigma)^{n-|A|}p_\sigma^{|A|}\bigotimes_{i\in A}(Z_i\otimes Z_{n+i}) \\
    & = \sum_{A\subseteq \NC} c_{\sigma}(A) Z_A \otimes  Z_{A} \label{eq:reminder-MMD-observable}
\end{align}
where we denote $A$ as the set of all possible
subsets of the indices $1, ..., n$ as defined in Eq.~\eqref{eq:A}, and in the second equality we introduce some shorthand notations $c_{\sigma}(A) = (1-p_\sigma)^{n-|A|}p_\sigma^{|A|}$ and $Z_A = \bigotimes_{i\in A}Z_i$. 
Also, each term in the MMD loss function can be expressed as some expectation of the form
\begin{align} 
    \MC(\rho, \rho') = \Tr[  O^{(\sigma)}_{\rm MMD} (\rho \otimes \rho')] \;,
\end{align}
and the MMD loss function is of the form
\begin{align}\label{eq:mmd-reminder}
     \LC_{\rm MMD}(\thv) =  \MC(\rho_{\thv}, \rho_{\thv}) - 2  \MC(\rho_{\thv}, \rho_{\pt}) + \MC(\rho_{\pt}, \rho_{\pt})  \;,
\end{align}
where $\rho_{\thv} = U(\thv) \rho_0 U^\dagger(\thv)$ and $\rho_{\pt}$ is a quantum state such that $\pt(\xv) = \Tr[\rho_{\pt} |\xv\rangle\langle \xv| ]$. For those who want more refresher on the MMD observable, we refer to Supplementary Note~\ref{app:mmd-k-body-observable}. 

We consider a generic Pauli rotation ans\"{a}tze of the form 
\begin{align}\label{eq:mmd-generic-ansatze}
    U_{\rm PQC}(\thv) = U(\alv) U_{\rm tensor} (\btv) \;,
\end{align}
where $U(\alv) = \prod_{k=1}^M e^{- i \alpha_k G_k /2} V_k$ such that the generators $\{ G_k \}_{k=1}^M$ are some $n$-qubit Pauli strings $G_k \in  \pauli := \{ \id, X, Y, Z \}^n$ and $\{ V_k\}_{k=1}^M$ is a set of non-parametrized Clifford gates, and $ U_{\rm tensor} (\btv) = \bigotimes_{i=1}^n U_i(\beta_i)$ is a layer of Haar random single qubit rotations $U_i(\beta_i)$. In addition, we assume all parameters are uncorrelated.

Before showing the main theoretical results, we remarks that any arbitrary Pauli-string $B \in \pauli := \{\id, X, Y, Z\}^{\otimes n}$ remains a single Pauli-string when being back-propagated with $U(\alv_{c})$ where $\alv = \alv_{c} \in \{ 0, \pi/2\}^M$ is a Clifford angle. That is, we have
\begin{align}
    B_{\alv_{c}} = U^\dagger(\alv_{c}) B U(\alv_{c}) \in \pauli = \{\id, X, Y, Z\}^{\otimes n} \;.
\end{align}
We refer to Lemma~\ref{lemma:remain-pauli-clifford} for more details and proof. 

Now, denote $\sup(B_{\alv_{c}})$ as the set of qubits that $B_{\alv_{c}}$ acts non-trivially on. The light cone of $B_{\alv_{c}}$ is defined as the cardinality of $\sup(B_{\alv_{c}})$ i.e., $\Delta_{B}(\alv_{c}) = | \sup(B_{\alv_{c}}) |$. For example, if $B_{\alv_{c}}$ corresponds to $Z\otimes \id^{n-4}\otimes X \otimes Y \otimes Z$, then we have $\Delta_{B}(\alv_{c}) = 4$. Notice the light cone defined this way depends on the Clifford angle $\alv_{c}$. We then introduce the concept of an average light cone
\begin{align}\label{eq:avg-light-cone}
    \Delta_{B}^{\rm (avg)} = \Ebb_{\alv_{c}\sim\DC_{c}} [\Delta_{B}(\alv_{c})] \;, 
\end{align}
where $\DC_{c}$ is a uniform distribution over the Clifford parameters. That is, $\Ebb_{\alv_{c} \sim \DC_{c}}[\cdot] = \frac{1}{2^M} \sum_{\alv_{c} \in \{0, \pi/2\}^M} [\cdot]$. The average light cone intuitively tells us on average how many qubits are involved by the end of the evolution of the circuit $U(\alv)$. 

\subsubsection{Summary of the key technical results}
We now present the generic expression of the variance of the MMD loss with general Pauli rotation ans\"{a}tze $U_{\rm PQC}(\thv)$ for an arbitrary band-width.
\begin{supplemental_proposition}\label{sup-prop:mmd-variance-pauli-rotations}
Consider the MMD loss function as defined in Eq.~\eqref{eq:mmd-loss-implicit} using the Gaussian kernel in Eq.~\eqref{eq:gaussian-kernel} and a QCBM composed of a parametrized Pauli rotation circuit as defined in Eq.~\eqref{eq:mmd-generic-ansatze}. Then, we have the variance of the MMD loss scales as
\begin{align}
    \Var_{\thv}[\LC^{(\sigma)}_{\rm MMD}(\thv)] = \Var_{\thv} \left[\MC(\rho_{\thv},\rho_{\thv})\right] + 4  \Var_{\thv} \left[\MC(\rho_{\thv},\rho_{\pt})\right] - 4 \Cov_{\thv}\left[\MC(\rho_{\thv},\rho_{\thv}), \MC(\rho_{\thv},\rho_{\pt}) \right]\;, \label{eq:sup-prop-mmd-variance-pauli-rotations}
\end{align}
with
\begin{align}
    &\Var_{\thv} \left[\MC(\rho_{\thv},\rho_{\pt})\right] = \sum_{\substack{A \subseteq \NC \\ A\neq \{\; \} }} \left( c_{\sigma}(A) z_A(\tilde{P}) \right)^2 \Ebb_{\alv_{c} \sim \DC_c}\left( \frac{1}{3}\right)^{\Delta_{Z_A} (\alv_{c})} \; ,\\
    &\Cov_{\thv}\left[\MC(\rho_{\thv},\rho_{\thv}), \MC(\rho_{\thv},\rho_{\pt}) \right]  = 0  \;, \\
    &\Var_{\thv} \left[\MC(\rho_{\thv},\rho_{\thv})\right]  \geq \Ebb_{(A, \alv_{c}), (A', \alv'_{c}) \sim \tilde{\DC}_{c}(\sigma)} \left(\frac{1}{3}\right)^{\Delta_{Z_{A}}(\alv_{c}) + \Delta_{Z_{A'}}(\alv'_{c})}\left[\left(\frac{9}{5}\right)^{\left| \sup(Z_{A, \alv_{c}}) \cap \sup(Z_{A', \alv'_{c}})\right|}\left(\frac{1}{3}\right)^{n(Z_{A,\alv_{c}},Z_{A',\alv'_{c}})}-1\right] \geq 0
\end{align}
where $z_A(\tilde{P}) = \Tr\left[Z_A \rho_{\pt}\right]$, $\alv_{c}$ are Clifford angles i.e., $\alv_{c} \in \{0, \pi/2\}^M$, $\DC_c$ is a uniform distribution over these Clifford angles, $Z_{A, \alv_{c}} = U^{\dagger}(\alv_{c}) Z_A U(\alv_{c}) $ is guaranteed to be a single Pauli-string i.e., $Z_{A, \alv_{c}} = \bigotimes_{i=1}^n Z^{(i)}_{A, \alv_{c}} \in \{ \id, X ,Y ,Z\}^{\otimes n}$ (which is $Z_A$ back-propagated with $U(\alv_{c})$) together with $\sup(Z_{A, \alv_{c}})$ as a set of qubits that $Z_{A, \alv_{c}}$ acts non-trivially on and $\Delta_{Z_A}(\alv_{c})$ as the light cone. 

In addition, $\tilde{\DC}_{c}(\sigma)$ is the probability distribution that $A$ is chosen with probability $c_{\sigma}(A)$ (from the set $\NC$) and independently $\alv_{c}$ is uniformly chosen from $\{0,\pi/2 \}^{M}$, as well as $n(Z_{A, \alv_{c}}, Z_{A', \alv'_{c}})$ is the number of qubits acted by different non-trivial single-qubit Pauli observables $Z^{(i)}_{A, \alv_{c}}$ and $Z^{(i)}_{A', \alv'_{c}}$ such that $Z^{(i)}_{A, \alv_{c}} \neq Z^{(i)}_{A', \alv'_{c}} \neq \id$.
\end{supplemental_proposition}

We remark that although some proof techniques to compute the variance of the overlap term are similar to Ref.~\cite{letcher2023tight}, the main technical challenges are to analytically show that the covariance between two terms vanish and compute the exact form of the variance lower bound of the purity term. Now, we recall that for the linear bandwidth $\sigma \in \Theta(n)$ the MMD observable behaves effectively as a few-body observable (see Proposition~\ref{prop:mmd-k-body}) which allows for some initial trainability guarantee of the QCBM. The formal version of Theorem~\ref{thm:mmd-train-general} is presented below.

\begin{theorem}[Pauli rotation ansatz MMD trainability, formal]\label{thm:mmd-train-general-appx}
Under the same assumptions as in Supplemental Proposition~\ref{sup-prop:mmd-variance-pauli-rotations} together with the linear bandwidth $\sigma \in \Theta(n)$, as long as there exists some $A \in \NC$ with $|A| \in \OC(\log(n))$ such that the average light cone of some $Z_A$ back-propagated with $U(\vec{\alpha}_{c})$ as defined in Eq.~\eqref{eq:avg-light-cone} remains at most in the order of $\log(n)$ qubits, i.e.,
\begin{align}
   \Delta_{Z_A}^{\rm (avg)} \in \OC(\log(n))  \;,
\end{align}
and as long as $z^2_A(\tilde{P})$, which encodes the information about the training distribution, decays at worst polynomially, i.e., 
\begin{align}
    z^2_A(\Tilde{P}) \in \Omega(1/\poly(n)) \;,
\end{align}
then the QCBM is trainable in the sense that
\begin{align}
    \Var_{\thv} [\LC^{(\sigma)}_{\rm MMD}(\thv)] \in \Omega(1/\poly(n)) \;.
\end{align}
\end{theorem}
We note that, similar to the discussion at the end of the previous sub-section, the condition on $z_{A}(\tilde{P})$ is mild and is expected to hold generally in practice as long as few-body marginals of the training distribution does not look too uniform.  

The theorem provides a generic criteria in practice for determining when an ansatz does not exhibit a barren plateau.  An example of an ansatz that satisfies the few-body light cone condition is a shallow depth circuit with nearest neighbour connectivity (which could be either hardware efficient or problem inspired). Crucially, we emphasise that this light cone argument goes beyond the 2-design assumption and is expected to work even more generally to any ansatz that may not even be in the Pauli rotation form.

Now, we provide details on how to numerically and efficiently estimate this average light cone of the Pauli rotations in practice. For a given Pauli $Z_A$, we can apply Monte Carlo sampling to estimate the average light cone over the uniform distribution of the Clifford angles. Particularly, we do the following 
\begin{enumerate}
    \item Uniformly sample a Clifford angle $\alv_{c}$,
    \item Work out the light cone of $Z_A$ back-propagated with $U(\alv_{c})$.
\end{enumerate}
This latter step is efficient since for each layer of Pauli-rotation and a fixed Clifford gate it always map a Pauli string to only a single Pauli string (i.e., Lemma~\ref{lemma:remain-pauli-clifford}) and the computational cost of this step scales at worst polynomially with $M$. With $N_{\rm sample}$ repetitions of steps $1.$ and $2.$, the empirical estimate of the average light converges at the rate $\OC\left(1/\sqrt{N_{\rm sample}}\right)$. 

Lastly, to derive the polynomial guarantee on the MMD variance, we simply throw away the purity term which is always guaranteed to be non-negative.  However, in practice, we expect the term to contribute to the whole MMD variance and does not generally vanish. Particularly, since the key quantity $z^2_{A}(\thv)$ of the purity also has the polynomially large variance with the few-body light cone, this intuitively implies the non-vanishing contribution of the purity term to the whole MMD variance. Thus, the theoretical lower bound here is expected to be sub-optimal, yet it is sufficient to show the trainability guarantee.

\subsubsection{Preliminaries: useful identities}
In this subsection, we introduce the following useful mathematical statements. We encourage the readers to only skim through this part and revisit when going through the proof of the main theoretical statements in the next subsections.  

\begin{lemma}\label{lemma:remain-pauli-clifford}
Consider a parametrized ansatz of the form $U(\alv) = \prod_{k=1}^M e^{- i \alpha_k G_k /2} V_k$ such that the generators $\{ G_k \}_{k=1}^M$ are some $n$-qubit Pauli strings $G_k \in \pauli = \{ \id, X, Y, Z \}^{\otimes n}$ and $\{ V_k\}_{k=1}^M$ is a set of non-parametrized Clifford gates, as well as all parameters are uncorrelated. Then, for the Clifford angles $\alv_c := \alv \in \{ 0, \pi/2 \}^M$ and a Pauli string $B$, the effective observable from back propagating $U(\alv_c)$ remains a single Pauli string (up to a sign factor). That is, we have
\begin{align}
    B_{\alv_{c}} = U^\dagger(\alv_{c}) B U(\alv_{c}) \in \pm\pauli \;.
\end{align}
\end{lemma}
\begin{proof}
Let us first consider a simpler ansatz of the form
\begin{align}
    U(\alpha) = e^{- i \alpha G /2} V \;,
\end{align}
where the generator is some Pauli-string $G \in \{ \id, X, Y, Z \}^{\otimes n}$ and $V$ is some Clifford gate. Then, we have
\begin{align}
    B_{\alpha} =& U^\dagger(\alpha) B U(\alpha) \\
    =& V^\dagger \left[ \cos(\alpha/2) \id + i \sin(\alpha/2) G \right] B \left[ \cos(\alpha/2) \id - i \sin(\alpha/2) G \right] V  \\
    =& \cos^2(\alpha/2) V^\dagger B V +  \sin^2(\alpha/2) V^\dagger G B G V + \sin(\alpha) V^\dagger \left( \frac{i}{2} [G,B]\right)V \;,
\end{align}
where we use the identity $e^{-i\alpha G/2} = \cos(\alpha/2) \id - i \sin(\alpha/2) B$. Since both $G$ and $B$ are Pauli strings, they either commute or anti-commute with one another, which leads to
\begin{align} \label{eq:lemma3-Pauli-first-order}
B_\alpha  = \left\{
\begin{array}{ll}
    V^\dagger B V & \mbox{ if } \comm{B}{G}=0 \;, \\ [.161cm]
    \cos(\alpha) V^\dagger B V +   \sin(\alpha) V^\dagger \left( \frac{i}{2} [G,B]\right)V & \mbox{ if } \{B,G\}=0\;.
\end{array}
\right.
\end{align}
From this, we can see that averaging over $\alpha$ leads to a vanishing expectation (i.e. first order moment vanishes if the Pauli generator of the rotation anti-commutes with the Pauli it acts on).
Now, in the case of $\{\sigma,P\}=0$, we have the following for the Clifford angles i.e., $\alpha \in \{ 0, \pi/2 \}$
\begin{align} 
B_{\alpha_c}  = \left\{
\begin{array}{ll}
    V^\dagger B V & \mbox{ if } \alpha_c =0 \;, \\ [.161cm]
    V^\dagger \left( \frac{i}{2} [G,B]\right)V & \mbox{ if }  \alpha_c = \pi/2\;.
\end{array}
\right.
\end{align}
Crucially, $\frac{i}{2} [G,B]$ remains a Pauli string. Notice that this is true up to a sign factor, but it will not impact our results in the following. In addition, since $V$ is a Clifford gate, it maps a Pauli-string to another Pauli-string, which implies 
\begin{align}
     V^\dagger B V \;,\;  V^\dagger \left( \frac{i}{2} [G,B]\right)V \in  \pm \{ \id, X, Y, Z \}^{\otimes n} \;.
\end{align}
That is, for a Clifford angle, the effective observable remains a single Pauli string. We can iteratively apply the result obtained from this simple ansatz to a more general ansatz of the form $U(\alv) = \prod_{k=1}^M e^{- i \alpha_k G_k /2} V_k$ when the parameter settings are all Clifford angles $\alv_c := \alv \in \{ 0, \pi/2 \}^M$, since the effective observable for each layer is a single Pauli string. This completes the proof of the lemma. 
\end{proof}

\begin{supplemental_proposition}\label{sup-prop:second-third-moments}
Consider a parametrized ansatz of the form $U(\alv) = \prod_{k=1}^M e^{- i \alpha_k G_k /2} V_k$ such that the generators $\{ G_k \}_{k=1}^M$ are some $n$-qubit Pauli strings $G_k \in \{ \id, X, Y, Z \}^n$ and $\{ V_k\}_{k=1}^M$ is a set of non-parametrized Clifford gates, as well as another tensor product ansatz of the form $U_{\rm tensor}(\btv) = \bigotimes_{i=1}^n U_i(\beta_i)$. Also, consider Pauli strings $B$ and $B'$ which are not identities i.e., $B, B' \subseteq \{ \id, X, Y, Z\}^{\otimes n} \backslash \id^{\otimes n}$. Then, if all parameters are uncorrelated, we have the following quantities related to higher moments expressed as
\begin{align}
\Ebb_{\alv} (U^\dagger(\alv))^{\otimes 2} B^{\otimes 2}(U(\alv))^{\otimes 2} = & \Ebb_{\alv_{c}\sim \DC_{c}} \left(B_{\alv_{c}} \otimes B_{\alv_{c}}\right) \label{eq:sup-prop-second-moment-clifford-only} \\
    \Ebb_{\alv, \btv} \Tr\left[ (U_{\rm tensor}(\btv))^{\otimes 2} \rho_0^{\otimes 2} (U_{\rm tensor}^\dagger(\btv))^{\otimes 2} (U^\dagger(\alv))^{\otimes 2} (B \otimes B')(U(\alv))^{\otimes 2}\right] = & \delta_{B, B'} \Ebb_{\alv_{c}\sim\DC_{c}}\left( \frac{1}{3} \right)^{\Delta_B(\alv_{c})} \label{eq:sup-prop-second-moment}\\
    \Ebb_{\alv,\btv} \Tr\left[ (U_{\rm tensor}(\btv))^{\otimes 3} \rho_0^{\otimes 3} (U_{\rm tensor}^\dagger(\btv))^{\otimes 3} (U^\dagger(\alv))^{\otimes 3} (B^{\otimes 2} \otimes  B')(U(\alv))^{\otimes 3}\right] = & 0 \label{eq:sup-prop-third-moment}\;, 
\end{align}
where $\delta_{B,B'}$ is the Kronecker delta of $B$ and $B'$, $\alv_{c} \in \{0, \pi/2\}^M$ is a Clifford angle, $B_{\alv_{c}} = U^\dagger(\alv_{c}) B U(\alv_{c})$, $\Delta_{B}(\alv_{c})$ is a light cone associated with $B_{\alv_{c}}$, and $\DC_{c}$ is a uniform distribution over all possible $\alv_{c}$.
\end{supplemental_proposition}

\begin{proof}
In the proof, we use a vectorization formalism and refer the readers who are not familiar to, for example, Sec.~5 of Ref.~\cite{mele2023introduction}. 

\bigskip

\noindent\textbf{1.~Second moment.} 

\medskip

\noindent\textit{\underline{1.1~Simple ansatz.}} We begin with a simpler ansatz of the form 
\begin{align}\label{eq:simple-ansazte-simple}
    U(\alpha) = e^{- i \alpha G /2} V \;,
\end{align}
where the generator is some Pauli-string $G \in \{ \id, X, Y, Z \}^{\otimes n}$ and $V$ is some Clifford gate. Denote the following notation
\begin{align}\label{eq:super-operator-t-design}
    \UC_G^{(t)}(\alpha) = & \left(e^{i\alpha G/2}\right)^{\otimes t} \otimes \left(e^{- i\alpha G^T/2}\right)^{\otimes t} \;. 
\end{align}

The vectorized form of $\Ebb_{\alpha} (U^\dagger(\alpha))^{\otimes 2} (B \otimes B') (U(\alpha))^{\otimes 2}$ can be expressed as
\begin{align}
    \Ebb_{\alpha} |(U^\dagger(\alpha))^{\otimes 2} (B \otimes B') (U(\alpha))^{\otimes 2} \rrangle = &  \Ebb_{\alpha}(U^\dagger(\alpha))^{\otimes 2} \otimes (U^T(\alpha))^{\otimes 2} | B \otimes B' \rrangle \\
    = & \Ebb_{\alpha} \left(\left(V^\dagger\right)^{\otimes 2} \otimes \left(V^T\right)^{\otimes 2}   \right)  \left( \left(e^{i\alpha G/2}\right)^{\otimes 2} \otimes \left(e^{- i\alpha G^T/2}\right)^{\otimes 2} \right) | B \otimes B' \rrangle \\
    = & \left(\left(V^\dagger\right)^{\otimes 2} \otimes \left(V^T\right)^{\otimes 2}   \right) \Ebb_{\alpha}\left[ \UC_G^{(2)}(\alpha) \right]| B \otimes B' \rrangle \;,
\end{align}
where the first equality is due to the transpose-trick $|ABC \rrangle = A\otimes C^T |B\rrangle$, and the second equality is due to $(A \otimes B)(C \otimes D) = (AC) \otimes (BD)$ and linearity of the expectation. To proceed, there are three different cases to consider together with using $e^{-i\alpha G/2} = \cos(\alpha/2) \id - i \sin(\alpha/2) G$, leading to

\medskip

\noindent \textit{\underline{1.1.(i)~If $G$ does not commute with $B$ and $B'$ i.e. $\{G,B\} = \{ G, B'\} = 0$}}, we have
\begin{align}
    \Ebb_{\alpha} |(U^\dagger(\alpha))^{\otimes 2} (B \otimes B') (U(\alpha))^{\otimes 2} \rrangle  & = \frac{| V^\dagger B V \otimes V^\dagger B' V \rrangle  + | V^\dagger (i G B) V \otimes V^\dagger (i G B') V \rrangle }{2} \\
    & = \left(\left(V^\dagger\right)^{\otimes 2} \otimes \left(V^T\right)^{\otimes 2}   \right)\left(\frac{\UC^{(2)}_G(0) + \UC^{(2)}_G(\pi/2)}{2}\right) |B \otimes B' \rrangle \;. \label{eq:proof-sup-proop-simple-ansatze1}
\end{align}

\medskip

\noindent \textit{\underline{1.1.(ii)~If $G$ commutes with both $B$ and $B'$ i.e., $[G, B] = [G, B'] = 0$}}, we have
\begin{align}
    \Ebb_{\alpha} |(U^\dagger(\alpha))^{\otimes 2} (B \otimes B') (U(\alpha))^{\otimes 2} \rrangle & = | V^\dagger B V \otimes V^\dagger B' V \rrangle \\
    & =  \left(\left(V^\dagger\right)^{\otimes 2} \otimes \left(V^T\right)^{\otimes 2}   \right)\left(\frac{\UC^{(2)}_G(0) + \UC^{(2)}_G(\pi/2)}{2} \right) |B \otimes B' \rrangle \;,\label{eq:proof-sup-proop-simple-ansatze2}
\end{align}
where, to express in the same form in the previous case, we add the action of $\UC^{(2)}_G(\pi/2)$ on $|B \otimes B' \rrangle$ which simply acts as an identity in this case (due to the commutation assumption).

\medskip

Crucially, note that $iGB$ and $iGB'$ are simply Pauli-strings i.e., the result of applying $\UC^{(2)}_G(\pi/2)$ on to $|B\otimes B' \rrangle$ in the 1.1.(i) case. Since $V$ is a Clifford gate which transforms a Pauli-string to another Pauli-string, we have that Eq.~\eqref{eq:proof-sup-proop-simple-ansatze1} results in an equally weighted sum of two Pauli-strings, and Eq.~\eqref{eq:proof-sup-proop-simple-ansatze2} results in a single Pauli-string.

\medskip

\noindent \textit{\underline{1.1.(iii)~If $G$ commutes with only either $B$ or $B'$ e.g., $\{G, B\} = 0$ and $[G, B'] = 0$ (or vice versa)}}, we have
\begin{align}
    \Ebb_{\alpha} |(U^\dagger(\alpha))^{\otimes 2} (B \otimes B') (U(\alpha))^{\otimes 2} \rrangle  & =   \Ebb_{\alpha} |(U^\dagger(\alpha)\otimes \id) (B \otimes B') (U(\alpha)\otimes \id) \rrangle\\
    & = 0 \;. \label{eq:proof-sup-proop-simple-ansatze3}
\end{align}
where the first equality is by the assumption that $[G, B'] = 0$ which simplifies the calculation to compute the first moment with respect to $B$, resulting in zero when averaging over the full angle range (see Eq.~\eqref{eq:lemma3-Pauli-first-order}). The same calculation can be done for the case of $[G,B]=0$ and $\{ G, B'\} \neq 0$.

\bigskip

\noindent\underline{\textit{1.2~Generic Pauli ansatz.}} Now, we extend these findings to an ansatz (with uncorrelated parameters) of the form 
\begin{align} \label{eq:proof-sup-prop-generic-ansatze}
    U(\alv) U_{\rm tensor}(\btv) = \left(\prod_{k=1}^M e^{- i \alpha_k G_k /2} V_k\right) \left(\bigotimes_{i=1}^n U_i(\beta_i)\right) \;,
\end{align}
such that the generators $\{ G_k \}_{k=1}^M$ are some $n$-qubit Pauli strings $G_k \in \{ \id, X, Y, Z \}^n$ and $\{ V_k\}_{k=1}^M$ is a set of non-parametrized Clifford gates, as well as all $U_i(\beta_i)$ is a parametrized single-qubit Haar random on the qubit $i$.

\medskip

\noindent \underline{\textit{1.2.(i)~The scenario when $B = B'$.}} The vectorized form of $\Ebb_{\alv}(U^\dagger(\alv))^{\otimes 2} (B \otimes B) (U(\alv))^{\otimes 2}$ can be expressed as
\begin{align}
    \Ebb_{\alv}| (U^\dagger(\alv))^{\otimes 2} (B \otimes B) (U(\alv))^{\otimes 2} \rrangle = &  \Ebb_{\alv} \left[(U^\dagger(\alv))^{\otimes 2} \otimes (U^T(\alv))^{\otimes 2} \right] | B \otimes B \rrangle \\
    = & \Ebb_{\alv}\left[ \prod_{k=M}^1 \left(\left(V^\dagger_k\right)^{\otimes 2} \otimes \left(V^T_k\right)^{\otimes 2}  \right) \left( \left(e^{i\alpha_k G_k/2}\right)^{\otimes 2} \otimes \left(e^{- i\alpha_k G^T_k/2}\right)^{\otimes 2} \right) \right]| B \otimes B \rrangle \\
    = & \prod_{k=M}^1\left(V^\dagger_k\right)^{\otimes 2} \otimes \left(V^T_k\right)^{\otimes 2}   \Ebb_{\alpha_k} \left[ \UC^{(2)}_{G_k}(\alpha_k)\right]  | B \otimes B \rrangle \\
    = & \left(\prod_{k=M-1}^1 \left(V^\dagger_k\right)^{\otimes 2} \otimes \left(V^T_k\right)^{\otimes 2}   \Ebb_{\alpha_k} \left[ \UC^{(2)}_{G_k}(\alpha_k)\right]\right) \cdot \left(\frac{\UC^{(2)}_{G_M}(0) + \UC^{(2)}_{G_M}(\pi/2)}{2} \right) |B \otimes B \rrangle\\
    = & \frac{1}{2^M} \sum_{\alv_{c} \in \{ 0,\pi/2 \}^M} \left[ \prod_{k=M}^1 \left(V^\dagger_k\right)^{\otimes 2} \otimes \left(V^T_k\right)^{\otimes 2} \UC^{(2)}_{G_k}(\alpha^{(k)}_{c})\right] |B\otimes B \rrangle \\
    = & \frac{1}{2^M}  \sum_{\alv_{c} \in \{ 0,\pi/2 \}^M}  |B_{\alv_{c}} \otimes B_{\alv_{c}} \rrangle \\
    = & \Ebb_{\alv_{c}\sim \DC_{c}}  |B_{\alv_{c}} \otimes B_{\alv_{c}} \rrangle
\end{align}
where in the second line we introduce the notation $\prod_{k=M}^1 a_k = a_1 a_2 ... a_{M-1}a_M$, the third equality is due to parameters being uncorrelated and by using the notation $\UC^{(2)}_{G_{k}}(\alpha_k)$ introduced in Eq.~\eqref{eq:super-operator-t-design}. The fourth equality is by either Eq.~\eqref{eq:proof-sup-proop-simple-ansatze1} if $\{B, G_M\} = 0$ or Eq.~\eqref{eq:proof-sup-proop-simple-ansatze2} if $[B,G_M] = 0$, which happily results in the same expression. The fifth equality is the result of recursively applying what we just do for all layers. Crucially, thanks to Lemma~\ref{lemma:remain-pauli-clifford}, each term in the sum correspond to two copies of a single Pauli string. In the last equality, we simply use the definition of uniformly average over the Clifford angles i.e., $\Ebb_{\alv_{c} \sim \DC_{c}} [\cdot] = \frac{1}{2^M} \sum_{\alv_{c} \in \{ 0, \pi/2\}^M} [\cdot]$. We note that this proves Eq.~\eqref{eq:sup-prop-second-moment-clifford-only} in the supplemental proposition.

Now, we can do the average over the final tensor-product layer $U_{\rm tensor}(\btv)=\bigotimes_{i=1}^nU_i(\beta_i)$  by expressing $B_{\alv_{c}} = \bigotimes_{i=1}^n B^{(i)}_{\alv_{c}}$ with $B^{(i)}_{\alv_{c}} \in \pm\{\id, X, Y, Z\}$ and performing Haar integration qubit-wise, leading to
\begin{align}
     \Ebb_{\alv, \btv} \Tr\left[ (U_{\rm tensor}(\btv))^{\otimes 2} \rho_0^{\otimes 2} (U_{\rm tensor}^\dagger(\btv))^{\otimes 2} (U^\dagger(\alv))^{\otimes 2} B^{\otimes 2}(U(\alv))^{\otimes 2}\right] = & \Ebb_{\btv} \Tr\left[ (U_{\rm tensor}(\btv))^{\otimes 2} \rho_0^{\otimes 2}  (U_{\rm tensor}^\dagger(\btv))^{\otimes 2} \Ebb_{\alv_{c}\sim \DC_{c}} \left( B_{\alv_{c}} ^{\otimes 2} \right)\right] \\
     = & \Ebb_{\alv_{c} \sim \DC_{c}} \prod_{i=1}^n \Ebb_{\beta_i} \Tr\left[ (U_i(\beta_i))^{\otimes 2} \rho_0^{(i)} (U_i^\dagger(\beta_i))^{ \otimes2}B_{\alv_{c}}^{(i)\otimes 2} \right] \\
     = & \Ebb_{\alv_{c} \sim \DC_{c}} \left(\frac{1}{3}\right)^{\Delta_B(\alv_{c})} \;,
\end{align}
where in the last equality, the Haar integration from Eq.~\eqref{eq:2-design-int-1-qubit} gives $1/3$ if $B^{(i)}_{\alv_{c}} \neq \id$ (otherwise, it gives $1$).

\medskip

\noindent\underline{\textit{1.2.(ii)~The scenario when $B \neq B'$.}} 
For two initial Pauli strings $B$ and $B'$ such that $B\neq B'$, they evolve into two different Pauli-strings $B_{\alv_{c}}$ and $B'_{\alv_{c}}$ with the same $U(\alv_{c})$. This can be trivially seen since $B_{\alv_{c}} = U^\dagger(\alv_{c}) B U(\alv_{c})$ and $B'_{\alv_{c}} = U^\dagger(\alv_{c}) B' U(\alv_{c})$ they can evolve to the same Pauli-string only if $B = B'$. Hence, the average over $\alv$ in the vectorized form can be expressed as
\begin{align}
    \Ebb_{\alv}| (U^\dagger(\alv))^{\otimes 2} (B \otimes B') (U(\alv))^{\otimes 2} \rrangle = & \frac{1}{2^M} \sum_{\alv_{c} \in \{0, \pi/2\}^M} c_{B,B'}(\alv_{c}) |B_{\alv_{c}} \otimes B'_{\alv_{c}} \rrangle
\end{align}
with $B_{\alv_{c}} \neq B'_{\alv_{c}}$. The coefficient $c_{B, B'}(\alv_{c})$ takes into account the fact that we could have some generator that commutes with only one of the two evolving Pauli-string (i.e., the 1.1.(iii) case described in the simple ansatz scenario above). $c_{B, B'}(\alv_{c})$ 
takes the value of either $1$ if there is no generator that commutes with only one evolving Pauli-string, or $0$ if there exists a generator that commutes with only one evolving Pauli-string (due to Eq.~\eqref{eq:proof-sup-proop-simple-ansatze3}). 

Crucially, averaging over the tensor-product ansatz $U_{\rm tensor}(\btv)$ ensures that the contribution of any pair of different Pauli-strings vanishes. That is, for any two different Pauli-strings $B_1$ and $B_2$ such that $B_1 \neq B_2$, we have
\begin{align}
    \Ebb_{\btv} \Tr\left[ (U_{\rm tensor}(\btv))^{\otimes 2} \rho_0^{\otimes 2}  (U_{\rm tensor}^\dagger(\btv))^{\otimes 2}  \left( B_1 \otimes B_2 \right)\right] &=  \prod_{i=1}^n \Ebb_{\beta_i}\Tr\left[ (U_i(\beta_i))^{\otimes 2} \rho_0^{(i) \otimes 2} (U_i^\dagger(\beta_i))^{\otimes2} \left( B_1^{(i)} \otimes B_2^{(i)} \right) \right] \\
    &= 0 \;,
\end{align}
which is from the direct computation of Haar integration with Eq.~\eqref{eq:2-design-int-1-qubit}.

From the 1.2.(ii) case, we have that for different initial Pauli-strings $B$ and $B'$ we have
\begin{align}
    \Ebb_{\alv, \btv} \Tr\left[ (U_{\rm tensor}(\btv))^{\otimes 2} \rho_0^{\otimes 2} (U_{\rm tensor}^\dagger(\btv))^{\otimes 2} (U^\dagger(\alv))^{\otimes 2} (B\otimes B')(U(\alv))^{\otimes 2}\right] = 0 \;.
\end{align}
\bigskip
Together from the 1.2.(i) and 1.2.(ii) cases, the proof of Eq.~\eqref{eq:sup-prop-second-moment} is complete.

\bigskip

\noindent\textbf{2.~Third moment.} We take the same strategy as in the second moment.

\medskip

\noindent\textit{\underline{2.1~Simple ansatz.}} We consider the same simple ansatz of the form in Eq.~\eqref{eq:simple-ansazte-simple}. Then, the vectorized form of $\Ebb_{\alpha} (U^\dagger(\alpha))^{\otimes 3} (B^{\otimes 2} \otimes B') (U(\alpha))^{\otimes 3}$ can be expressed as
\begin{align}
    \Ebb_{\alpha} |(U^\dagger(\alpha))^{\otimes 3} (B^{\otimes 2} \otimes B') (U(\alpha))^{\otimes 3} \rrangle = &  \Ebb_{\alpha}(U^\dagger(\alpha))^{\otimes 3} \otimes (U^T(\alpha))^{\otimes 3} | B^{\otimes 2} \otimes B' \rrangle \\
    = & \Ebb_{\alpha} \left(\left(V^\dagger\right)^{\otimes 3} \otimes \left(V^T\right)^{\otimes 3}   \right)  \left( \left(e^{i\alpha G/2}\right)^{\otimes 3} \otimes \left(e^{- i\alpha G^T/2}\right)^{\otimes 3} \right) | B^{\otimes 2} \otimes B' \rrangle \\
    =  & \left(\left(V^\dagger\right)^{\otimes 3} \otimes \left(V^T\right)^{\otimes 3}   \right) \Ebb_{\alpha}\left[ \UC_G^{(3)}(\alpha) \right]| B^{\otimes 2} \otimes B' \rrangle \;,
\end{align}
where the first equality is due to the transpose-trick $|ABC \rrangle = A\otimes C^T |B\rrangle$, and the second equality is due to $(A \otimes B)(C \otimes D) = (AC) \otimes (BD)$ and linearity of the expectation. To proceed, there are three different cases to consider together with using $e^{-i\alpha G/2} = \cos(\alpha/2) \id - i \sin(\alpha/2) G$, which lead to

\medskip

\noindent\textit{\underline{2.1.(i)~If $G$ does not commute with $B'$ i.e. $\{G,B'\} = 0$}}, we further have two sub-scenarios including (i).a for $[G, B] = 0$ and (i).b for $\{G, B\} = 0$. Either sub-scenario results in zero when averaging over i.e.,
\begin{align}
     \Ebb_{\alpha} |(U^\dagger(\alpha))^{\otimes 3} (B^{\otimes 2} \otimes B') (U(\alpha))^{\otimes 3} \rrangle = 0 \;.
\end{align}
Particularly, in the (i).a sub-scenario with $[G, B] = 0$, the computation reduces to compute the first moment which results in zero. On the other hand, in the (i).b sub-scenario with $\{ G, B\} = 0$, we have
\begin{align}
    \Ebb_{\alpha} |(U^\dagger(\alpha))^{\otimes 3} (B^{\otimes 2} \otimes B') (U(\alpha))^{\otimes 3} \rrangle
    =\; &  \Ebb_{\alpha}[\cos^3(\alpha)]|B^{\otimes 2}\otimes B'\rrangle + \Ebb_{\alpha}[\sin^2(\alpha)\cos(\alpha)] |(iG B)^{ \otimes 2} \otimes B' \rrangle \\ \nonumber
    &+ \Ebb_{\alpha}[\cos^2(\alpha)\sin(\alpha)]\left( | B \otimes iG B \otimes B' \rrangle + | i G B \otimes B \otimes B' \rrangle \right) \\ \nonumber
    &+ \Ebb_{\alpha}[\cos^2(\alpha)\sin(\alpha)] | B^{\otimes 2} \otimes iGB \rrangle  
    + \Ebb_{\alpha}[\sin^3(\alpha)]| (iGB)^{\otimes2} \otimes iGB' \rrangle  \\ \nonumber
    &+\Ebb_{\alpha}[\cos(\alpha)\sin^2(\alpha)]
    \left( | B\otimes iGB \otimes iGB' \rrangle + | iGB \otimes B \otimes iGB' \rrangle \right) \\
    = \;& 0 \;.
\end{align}

Indeed, this (i).b sub-scenario is the only scenario that is of particular interest regarding the third moment. Other scenarios reduce themselves to either compute the first or second moments over the parameters $\alv$, as we already saw in the (i).a sub-scenario.

\medskip

\noindent\textit{\underline{2.1.(ii)~If $G$ commutes with $B'$ and does not commute with $B$ i.e., $[G,B'] = 0$ and $\{G, B\} = 0$}}, this reduces to compute the second moment (see the case 1.1.(i) above i.e., Eq.~\eqref{eq:proof-sup-proop-simple-ansatze1}), leading to

\begin{align}
     \Ebb_{\alpha} |(U^\dagger(\alpha))^{\otimes 3} (B^{\otimes 2} \otimes B') (U(\alpha))^{\otimes 3} \rrangle =  \left(\left(V^\dagger\right)^{\otimes 2} \otimes \left(V^T\right)^{\otimes 2}   \right)\left(\frac{\UC^{(2)}_G(0) + \UC^{(2)}_G(\pi/2)}{2}\right) |B \otimes B \rrangle \otimes \left( V^\dagger \otimes V^T \right) |B'\rrangle \;.
\end{align}

\medskip

\noindent\textit{\underline{2.1.(iii)~If $G$ commutes with $B'$ and commutes with $B$ i.e., $[G, B'] = 0$ and $[G, B] = 0$,}} this reduces to the parameterised part of the circuit completely commuting with the observables, left them unchanged. The result can also be written in the same form as in the previous case (see the case 1.1.(ii) above i.e., Eq.~\eqref{eq:proof-sup-proop-simple-ansatze2})

\begin{align}
     \Ebb_{\alpha} |(U^\dagger(\alpha))^{\otimes 3} (B^{\otimes 2} \otimes B') (U(\alpha))^{\otimes 3} \rrangle & = \left(V^\dagger\right)^{\otimes 3} \otimes \left(V^T\right)^{\otimes 3}  | B^{\otimes 2} \otimes B'\rrangle \\
     & =  \left(\left(V^\dagger\right)^{\otimes 2} \otimes \left(V^T\right)^{\otimes 2}   \right)\left(\frac{\UC^{(2)}_G(0) + \UC^{(2)}_G(\pi/2)}{2}\right) |B \otimes B \rrangle \otimes \left( V^\dagger \otimes V^T \right) |B'\rrangle \;,
\end{align}
where $\UC^{(2)}_G(\pi/2)$ simply commutes through and effectively acts as an identity. 

\bigskip

\noindent\underline{\textit{2.2~Generic Pauli ansatz.}} Now, we generalises our findings to Pauli rotation ans\"{a}tze (with uncorrelated parameters) of the form in Eq.~\eqref{eq:proof-sup-prop-generic-ansatze}. 

To analytically show that the third moment vanishes, we can keep applying the results in the simple ansatz scenario to each layer. We notice that if the 2.1.(i) case in the simple ansatz above happens the third moment of the generic Pauli rotation ansatz simply vanishes. For example, this happens when $B'$ does not commute with $G_M$. More generally, if there exists at least one generator that does not commute with the evolving Pauli-string with the initial $B'$ at that particular layer, then this implies the 2.1.(i) case happens. 

In the event that we never satisfy the 2.1.(i) case this means the evolving $B'$ always commutes with the generator. In other words, the evolving $B'$ does not contribute at all to the third moment regarding $U(\alv)$. By repeatedly applying the 2.1.(ii) and 2.1.(iii) cases, the evolution of $B^{\otimes 2}$ can be obtained to be identical to the 1.2.(i) case in the second moment above. Thus, the vectorized form of the average over $U(\alv)$ is
\begin{align}\label{eq:proof-sup-third-moment-gone0}
    \Ebb_{\alv} |(U^\dagger(\alv))^{\otimes 3} (B^{\otimes 2} \otimes B') (U(\alv))^{\otimes 3} \rrangle &= \Ebb_{\alv_{c}\sim \DC_{c}}  |B_{\alv_{c}} \otimes B_{\alv_{c}} \rrangle \otimes |B'_{V} \rrangle \;,
\end{align}
with $|B'_{V}\rrangle = \prod_{k=M}^1 (V^\dagger_{k} \otimes V^T_{k}) |B'\rrangle$ which is still a single Pauli-string due to the Clifford gates $\{V_{k}\}_{k=1}^M$. 

Now, the average over the tensor-product ansatz $U_{\rm tensor}(\btv)$ ensures the third moment disappears. Since for any two different Pauli-strings $B_1$ and $B_2$ that are not identity such that $B_1 \neq B_2$, we have 
\begin{align}
    \Ebb_{\btv} \Tr\left[ (U_{\rm tensor}(\btv))^{\otimes 3} \rho_0^{\otimes 3}  (U_{\rm tensor}^\dagger(\btv))^{\otimes 3}  \left( B_1^{\otimes 2} \otimes B_2 \right)\right] &=  \prod_{i=1}^n \Ebb_{\beta_i}\Tr\left[ U_i(\beta_i)^{\otimes 3} \rho_0^{(i) \otimes 3} U_i^{\dagger \otimes 3}(\beta_i) \left( B_1^{(i) \otimes 2} \otimes B_2^{(i)} \right) \right] \\
    &= 0 \;, \label{eq:proof-sup-third-moment-gone}
\end{align}
which is due to the direct computation of Haar integration using Eq.~\eqref{eq:3-design-int-1-qubit}. Due to the linearity of average and by applying Eq.~\eqref{eq:proof-sup-third-moment-gone} to each term in Eq.~\eqref{eq:proof-sup-third-moment-gone0}, we have the third moment to vanish in this scenario. 

Hence, altogether the third moment vanishes in all scenarios
\begin{align}
    \Ebb_{\alv,\btv} \Tr\left[ (U_{\rm tensor}(\btv))^{\otimes 3} \rho_0^{\otimes 3} (U_{\rm tensor}^\dagger(\btv))^{\otimes 3} (U^\dagger(\alv))^{\otimes 3} (B^{\otimes 2} \otimes  B')(U(\alv))^{\otimes 3}\right] =  0  \;.
\end{align}
This completes the proof of the supplemental proposition.

\end{proof}

\begin{lemma}\label{lemma:decompose-general-variance}
Consider a function $f(\thv) := f(\vec{\alpha}, \vec{\beta})$ with $\thv := (\vec{\alpha}, \vec{\beta})$. The variance of the function can be decomposed as
\begin{align}
    \Var_{\thv}[f(\thv)] =   \Ebb_{\alv}[\Var_{\btv} f(\alv, \btv) ] + \Var_{\vec{\alpha}}[\Ebb_{\vec{\beta}} f(\alv, \btv)]\;.
\end{align}
\end{lemma}
\begin{proof} 
We can express the variance of $f(\thv)$ as
\begin{align}
\Var_{\thv} [f(\thv)] = & \Var_{\alv, \btv} [f(\alv, \btv)] \\
= & \Ebb_{\alv, \btv} [f^2(\alv, \btv)] - (\Ebb_{\alv, \btv} [f(\alv, \btv)])^2 \\
= & \Ebb_{\alv, \btv} [f^2(\alv, \btv)] + \left( \Ebb_{\alv} \left[  \Ebb_{\btv} \left[f(\alv, \btv)\right]^2\right] - \Ebb_{\alv} \left[  \Ebb_{\btv} \left[f(\alv, \btv)\right]^2\right] \right) - (\Ebb_{\alv, \btv} [f(\alv, \btv)])^2 \\
= & \Ebb_{\alv}\left(\Ebb_{\btv}[f^2(\alv,\btv)] - \Ebb_{\btv}[f(\alv,\btv)]^2 \right) + \left( \Ebb_{\alv}\left[\Ebb_{\btv}[f(\alv,\btv)]^2\right] - (\Ebb_{\alv}[\Ebb_{\btv}f(\alv, \btv)])^2 \right) \\
= & \Ebb_{\alv} [\Var_{\btv} f(\alv,\btv)] + \Var_{\alv}[\Ebb_{\btv} f(\alv, \btv)] \;,
\end{align}
where in the third equality we add and subtract $\Ebb_{\alv} \left[  \Ebb_{\btv} \left[f(\alv, \btv)\right]^2\right]$ and to reach the last equality we simply use the definition of the variance. This completes the proof of the lemma.
\end{proof}

\subsubsection{Proof of Supplemental Proposition~\ref{sup-prop:mmd-variance-pauli-rotations}}

\begin{proof}
The variance of the MMD loss in Eq.~\eqref{eq:mmd-reminder} can be written as
\begin{align}
    \Var_{\thv} [\LC_{\rm MMD}(\thv)] = & \Var_{\thv} \left[ \MC(\rho_{\thv}, \rho_{\thv}) - 2  \MC(\rho_{\thv}, \rho_{\pt}) + \MC(\rho_{\pt}, \rho_{\pt}) \right] \\
     = & \Var_{\thv} [ \MC(\rho_{\thv}, \rho_{\thv}] + 4 \Var_{\thv} [\MC(\rho_{\thv}, \rho_{\pt})] - 4 \Cov_{\thv} [ \MC(\rho_{\thv}, \rho_{\thv}), \MC(\rho_{\thv}, \rho_{\pt})] \;, \label{eq:proof-general-mmd-1}
\end{align}
where we have used $\Var[X + Y] = \Var[X] + \Var[Y] + 2 \Cov[X,Y]$ and  $\Var[X + c] = \Var[X]$ for any random variables $X, Y$ and some constant $c$.

\medskip

Our strategy is to tackle Eq.~\eqref{eq:proof-general-mmd-1} term by term. 

\medskip

\noindent\textit{\underline{(i) the variance of the cross term $\MC(\rho_{\thv}, \rho_{\pt})$.}} The cross term can be expressed as  
\begin{align}
    \MC(\rho_{\thv}, \rho_{\pt}) = &  \Tr[  O^{(\sigma)}_{\rm MMD} (\rho_{\thv} \otimes \rho_{\pt})] \\
    = & \Tr\left[ \sum_{A\subseteq \NC} c_{\sigma}(A) Z_A \otimes Z_{A}  (\rho_{\thv} \otimes \rho_{\pt}) \right] \\
    = &  \sum_{A\subseteq \NC} c_{\sigma}(A) z_A(\tilde{P}) z_A(\thv) \;,
\end{align}
where $z_A(\tilde{P}) = \Tr[\rho_{\pt}  Z_A ]$ and $z_A(\thv) = \Tr[U_{\rm PQC}(\thv) \rho_0 U_{\rm PQC}^\dagger(\thv) Z_A]$. Then, we can compute the variance as
\begin{align}
    \Var_{\thv}\left[ \MC(\rho_{\thv}, \rho_{\pt}) \right] = & \Var_{\thv} \left[ \sum_{A\subseteq \NC } c_{\sigma}(A) z_A(\tilde{P}) z_A(\thv) \right] \\
    = & \sum_{A \subseteq \NC} (c_{\sigma}(A) z_A(\tilde{P}))^2 \Var_{\thv}\left[ z_A(\thv) \right] + \sum_{\substack{A, A' \subseteq \NC \\  A \neq A'  }}  c_{\sigma}(A) c_{\sigma} (A') z_A(\tilde{P}) z_{A'}(\tilde{P}) \Cov_{\thv}[z_{A}(\thv), z_{A'}(\thv)] \;, \label{eq:proof-sup-key-var-cross}
\end{align}
which is due to a standard formula of variance of a sum i.e., $\Var\left[ \sum_{i} X_i \right] = \sum_{i,j} \Cov[X_i, X_j]$. Notice that the identity term, $A=\{\;\}$, does not contribute to the variance as it is a constant. Indeed, in this case we have $z_A(\thv)=1$. So, we are only considering $A,A'\neq \{\;\}$ in the following (i.e. we can remove cases where $Z_A=\id$ or $Z_{A'}=\id$ in previous equation). First, we consider the average term
\begin{align}
   \Ebb_{\thv}[z_A(\thv)] =&  \Ebb_{\alv, \btv}\Tr\left[ U_{\rm tensor}(\btv) \rho_0 U_{\rm tensor}^\dagger(\btv) \left(  U^\dagger(\alv) Z_A U(\alv)\right) \right] \\
   = &  \Ebb_{\btv}\Tr\left[ U_{\rm tensor}(\btv) \rho_0 U_{\rm tensor}^\dagger(\btv) \Ebb_{\alv}\left(  U^\dagger(\alv) Z_A U(\alv)\right) \right] \\
   =& 0 \;,
\end{align}
where we average over the full angle range and hence the first moment disappears in this case. 

Next, we consider the covariance term when $A \neq A'$
\begin{align}
    \Cov_{\thv}[z_A(\thv), z_{A'}(\thv)] =& \Ebb_{\thv} [z_A(\thv) z_{A'}(\thv)] \\
    = & \Ebb_{\alv, \btv} \Tr\left[ (U_{\rm tensor}(\btv))^{\otimes 2} \rho_0^{\otimes 2} (U_{\rm tensor}^\dagger(\btv))^{\otimes 2} (U^\dagger(\alv))^{\otimes 2} (Z_A \otimes Z_{A'})(U(\alv))^{\otimes 2}\right] \\
    = & 0 \;,
\end{align}
where to reach the last equality we apply Eq.~\eqref{eq:sup-prop-second-moment} from Supplemental Proposition~\ref{sup-prop:second-third-moments}. 

Lastly, we compute the variance terms.
\begin{align}
    \Var_{\thv}[z_A(\thv)] = & \Ebb_{\thv} [z^2_A(\thv)] \\
    = & \Ebb_{\alv, \btv} \Tr\left[ (U_{\rm tensor}(\btv))^{\otimes 2} \rho_0^{\otimes 2} (U_{\rm tensor}^\dagger(\btv))^{\otimes 2} (U^\dagger(\alv))^{\otimes 2} Z_A ^{\otimes 2}(U(\alv))^{\otimes 2}\right] \\
    = & \Ebb_{\alv_{c} \sim \DC_{c}} \left( \frac{1}{3}\right)^{\Delta_{Z_A}(\alv_{c})} \;,
\end{align}
where similarly to reach the last equality we apply Eq.~\eqref{eq:sup-prop-second-moment} from Supplemental Proposition~\ref{sup-prop:second-third-moments}. 

By substituting these back into Eq.~\eqref{eq:proof-sup-key-var-cross}, we obtain the variance of the cross term as
\begin{align}
    \Var_{\thv} \left[\MC(\rho_{\thv},\rho_{\pt})\right] = \sum_{\substack{A \subseteq \NC \\ A\neq \{\;\}}} \left( c_{\sigma}(A) z_A(\tilde{P}) \right)^2 \Ebb_{\alv_{c} \sim \DC_M}\left( \frac{1}{3}\right)^{\Delta_{Z_A} (\alv_{c})} \;.
\end{align}

\bigskip

\noindent\textit{\underline{(ii) The covariance between the purity term $\MC(\rho_{\thv}, \rho_{\thv})$ and the cross term $\MC(\rho_{\thv}, \rho_{\pt})$}} is of the form
\begin{align}
    \Cov_{\thv} [ \MC(\rho_{\thv}, \rho_{\thv}), \MC(\rho_{\thv}, \rho_{\pt})] &= \Cov_{\thv}\left[ \sum_{A\subseteq \NC} c_{\sigma}(A) z^2_A(\thv), \sum_{A' \subseteq \NC} c_{\sigma}(A') z_{A'}(\tilde{P}) z_{A'}(\thv)\right] \\
    & = \sum_{A,A' \subseteq \NC} c_{\sigma}(A) c_{\sigma}(A') z_{A'}(\tilde{P}) \Cov_{\thv}\left[ z_A^2(\thv), z_{A'}(\thv) \right] \;,
\end{align}
There are 4 different cases to consider. We note that when $Z_A = Z_{A'} = \id$ we have that $z_A(\thv) = z_{A'}(\thv) = \Tr[U_{\rm PQC}(\thv) \rho_0 U^\dagger(\thv) \id] = 1$. 

For $Z_A \neq \id$ and $Z_{A'} \neq \id$, we have
\begin{align}
    \Cov_{\thv}\left[ z_A^2(\thv), z_{A'}(\thv) \right] & = \Ebb_{\thv}\left[ z_A^2(\thv) z_{A'}(\thv) \right] - \Ebb_{\thv}\left[ z_A^2(\thv)\right] \Ebb_{\thv}\left[ z_{A'}(\thv) \right] \\
    & = \Ebb_{\thv}\left[ z_A^2(\thv) z_{A'}(\thv) \right] \\
    & = \Ebb_{\thv}\Tr\left[(U_{\rm tensor}(\btv))^{\otimes 3} \rho_0^{\otimes 3}  (U^\dagger_{\rm tensor}(\btv))^{\otimes 3} (U^{\dagger}(\alv))^{\otimes 3}  (Z_A^{\otimes 2} \otimes Z_{A'}) (U(\alv))^{\otimes 3}   \right] \\
    & = \Ebb_{\btv}\Tr\left[(U_{\rm tensor}(\btv))^{\otimes 3} \rho_0^{\otimes 3}  (U^\dagger_{\rm tensor}(\btv))^{\otimes 3} \Ebb_{\alv} \left( (U^{\dagger}(\alv))^{\otimes 3} (Z_A^{\otimes 2} \otimes Z_{A'}) (U(\alv))^{\otimes 3}  \right) \right] \\
    & = 0 \;,
\end{align}
where the last equality is due to Eq.~\eqref{eq:sup-prop-third-moment} from Supplemental Proposition~\ref{sup-prop:second-third-moments}.

For $Z_A = \id$ and $Z_{A'} \neq \id$, we have
\begin{align}
    \Cov_{\thv}\left[ z_A^2(\thv), z_{A'}(\thv) \right] & = \Ebb_{\thv}\left[ z_A^2(\thv) z_{A'}(\thv) \right] - \Ebb_{\thv}\left[ z_A^2(\thv)\right] \Ebb_{\thv}\left[ z_{A'}(\thv) \right] \\
    & = \Ebb_{\thv}\left[ z_{A'}(\thv) \right] - \Ebb_{\thv}\left[ z_{A'}(\thv) \right]  \\
    & = 0\;. 
\end{align}

For $Z_A \neq \id$ and $Z_{A'} = \id$, we have
\begin{align}
    \Cov_{\thv}\left[ z_A^2(\thv), z_{A'}(\thv) \right] & = \Ebb_{\thv}\left[ z_A^2(\thv) z_{A'}(\thv) \right] - \Ebb_{\thv}\left[ z_A^2(\thv)\right] \Ebb_{\thv}\left[ z_{A'}(\thv) \right] \\
    & = \Ebb_{\thv}\left[ z^2_{A}(\thv) \right] - \Ebb_{\thv}\left[ z^2_{A}(\thv) \right]  \\
    & = 0 \;.
\end{align}

For $Z_A = \id$ and $Z_{A'} = \id$, we have 
\begin{align}
    \Cov_{\thv}\left[ z_A^2(\thv), z_{A'}(\thv) \right] & = \Ebb_{\thv}\left[ z_A^2(\thv) z_{A'}(\thv) \right] - \Ebb_{\thv}\left[ z_A^2(\thv)\right] \Ebb_{\thv}\left[ z_{A'}(\thv) \right] \\
    & = 1 - 1  \\
    & = 0 \;.
\end{align}

Hence, altogether we conclude in all cases that the covariance between the cross term and the purity term vanishes
\begin{align}
    \Cov_{\thv} [ \MC(\rho_{\thv}, \rho_{\thv}), \MC(\rho_{\thv}, \rho_{\pt})] = 0 \;.
\end{align}

\bigskip

\noindent\textit{\underline{(iii) The variance of the purity term $\MC(\rho_{\thv}, \rho_{\thv})$.}} We remind that $\thv = (\alv, \btv)$ and apply Lemma~\ref{lemma:decompose-general-variance} to obtain the following lower bound
\begin{align}
    \Var_{\thv} [\MC(\rho_{\thv}, \rho_{\thv})] & =  \Var_{\btv} \left[ \Ebb_{\alv} \MC(\rho_{\thv}, \rho_{\thv})  \right] + \Ebb_{\btv}\left[ \Var_{\alv}  \MC(\rho_{\thv}, \rho_{\thv}) \right]  \\
    & \geq  \Var_{\btv} \left[ \Ebb_{\alv} \MC(\rho_{\thv}, \rho_{\thv})  \right] \;, \label{eq:proof-sup-prop-purity0}
\end{align}
where we throw away the second non-negative term.

\medskip

Now, we consider the central object $\Ebb_{\alv} \MC(\rho_{\thv}, \rho_{\thv})$ which can be expressed as
\begin{align}
    \Ebb_{\alv} \MC(\rho_{\thv}, \rho_{\thv}) & = \Ebb_{\alv} \sum_{A \subseteq \NC} c_{\sigma}(A) (z_A(\thv))^2 \\
    & = \sum_{A \subseteq \NC} c_{\sigma}(A) \Tr\left[ U^{\otimes 2}_{\rm tensor}(\btv) \rho_0^{\otimes 2} U^{\dagger \otimes 2}_{\rm tensor}(\btv) \Ebb_{\alv} (U^{\dagger \otimes 2}(\alv) Z_A^{\otimes 2} U^{\otimes 2}(\alv)) \right] \\
    & = \sum_{A \subseteq \NC} c_{\sigma}(A)\Tr\left[  U^{\otimes 2}_{\rm tensor}(\btv) \rho_0^{\otimes 2} U^{\dagger \otimes 2}_{\rm tensor}(\btv) \Ebb_{\alv_{c} \sim \DC_{c}}[ Z_{A, \alv_{c}} \otimes Z_{A, \alv_{c}} ] \right] \\
    & = \sum_{A \subseteq \NC} c_{\sigma}(A) \Ebb_{\alv_{c} \sim \DC_{c}} \Tr\left[ \rho_{\btv}^{\otimes 2}  Z^{\otimes 2}_{A, \alv_{c}} \right] \\
    & = \Ebb_{(A , \alv_{c}) \sim \tilde{\DC}_{c}(\sigma)} \Tr\left[ \rho_{\btv}^{\otimes 2}  Z^{\otimes 2}_{A, \alv_{c}}  \right] \;, \label{eq:proof-sup-prop-purity1}
\end{align}
where the third equality is due to Eq.~\eqref{eq:sup-prop-second-moment-clifford-only} in Supplemental Proposition~\ref{sup-prop:second-third-moments} and in the fourth equality we pull out the expectation and introduce the shorthand $\rho_{\btv} = U_{\rm tensor}(\btv)\rho_0 U^\dagger_{\rm tensor}(\btv)$. To reach the last equality, we notice that $c_{\sigma}(A) = (1-p_\sigma)^{n-|A|}p_\sigma^{|A|}$ (see e.g., Eq.~\eqref{eq:reminder-MMD-observable})) can be interpreted as a probability of sampling $A$, and $\sum_{A \subseteq \NC} c_{\sigma}(A) = 1$. Hence, $\sum_{A\subseteq \NC} c_{\sigma}(A) \Ebb_{\alv_{c} \sim \DC_{c}} [\cdot] = \Ebb_{(A , \alv_{c}) \sim \tilde{\DC}_{c}(\sigma)}$ where $\tilde{\DC}_{c}(\sigma)$ is the probability distribution that $A$ is chosen with probability $c_{\sigma}(A)$ (from the set $\NC$) and independently $\alv_{c}$ is uniformly chosen from $\{0,\pi/2 \}^{M}$.

By plugging Eq.~\eqref{eq:proof-sup-prop-purity1} back into Eq.~\eqref{eq:proof-sup-prop-purity0}, we have
\begin{align}
    \Var_{\thv} [\MC(\rho_{\thv}, \rho_{\thv})] \geq\; & \Var_{\btv} \left[ \Ebb_{(A , \alv_{c}) \sim \tilde{\DC}_{c}(\sigma)}  \Tr\left[ \rho_{\btv}^{\otimes 2}  Z^{\otimes 2}_{A, \alv_{c}}  \right]\right] \\
    =\;&  \Ebb_{(A, \alv_{c}), (A', \alv'_{c}) \sim \tilde{\DC}_{c}(\sigma)}   \Cov_{\btv}\left[  \Tr\left[ \rho_{\btv}^{\otimes 2}  Z^{\otimes 2}_{A, \alv_{c}}  \right] ,  \Tr\left[ \rho_{\btv}^{\otimes 2}  Z^{\otimes 2}_{A', \alv'_{c}}  \right]\right] \\
    =\;& \Ebb_{(A, \alv_{c}), (A', \alv'_{c}) \sim \tilde{\DC}_{c}(\sigma)} \left[ \Ebb_{\btv}\Tr\left[ \rho_{\btv}^{\otimes 4}  (Z^{\otimes 2}_{A, \alv_{c}} \otimes  Z^{\otimes 2}_{A',\alv'_{c}})\right]  - \Ebb_{\btv} \Tr\left[ \rho_{\btv}^{\otimes 2}  Z^{\otimes 2}_{A, \alv_{c}}  \right]\Ebb_{\btv} \Tr\left[ \rho_{\btv}^{\otimes 2}  Z^{\otimes 2}_{A', \alv'_{c}}  \right]\right] \\
    = \;&  \Ebb_{(A, \alv_{c}), (A', \alv'_{c}) \sim \tilde{\DC}_{c}(\sigma)} \left[  \prod_{i=1}^n \Ebb_{\beta_i}\Tr\left[ \rho_{\beta_i}^{\otimes 4}  (Z^{(i) \otimes 2}_{A, \alv_{c}} \otimes  Z^{(i) \otimes 2}_{A',\alv'_{c}})\right]- \prod_{i=1}^n  \Ebb_{\beta_i} \Tr\left[ \rho_{\beta_i}^{\otimes 2}  Z^{(i)\otimes 2}_{A, \alv_{c}}  \right]\Ebb_{\beta_i} \Tr\left[ \rho_{\beta_i}^{\otimes 2}  Z^{(i)\otimes 2}_{A, \alv_{c}}  \right] \right] \;, \label{eq:proof-sup-purity-key} 
\end{align}
where the second line is due to $\Var[\sum_i X_i] = \sum_{i,j} \Cov[X_i, X_j]$.

To further proceed, we have to carry out these Haar integration in Eq.~\eqref{eq:proof-sup-purity-key}.

For the second order moment, we remind that the following relation holds through the Haar integration.
\begin{align}
    \Ebb_{\beta_i} \Tr\left[ \rho_{\beta_i}^{\otimes 2}  Z^{(i)\otimes 2}_{A, \alv_{c}}  \right] = \left\{
\begin{array}{ll}
    1/3 & \mbox{ if } Z^{(i)}_{A,\alv_{c}} \neq \id \;, \\ [.161cm]
   1 & \mbox{ if } Z^{(i)}_{A,\alv_{c}} = \id \;.
\end{array}
\right. \label{eq:proof-purity-start}
\end{align}

For the fourth order moment, we first notice that if either $Z^{(i)}_{A,\alv_{c}}=\id$ or $Z^{(i)}_{A',\alv'_{c}}=\id$, we do not need to compute 4-design integration. Indeed, we have that $\Tr[\rho_{\beta_i}^{\otimes 4} Z^{(i)\otimes 2}_{A,\alv_{c}}\otimes \id^{\otimes 2} ]=\Tr[\rho_{\beta_i}^{\otimes 2} Z^{(i) \otimes 2}_{A,\alv_{c}} ]$. Now, for the cases where neither $Z^{(i)}_{A,\alv_{c}}$ nor $Z^{(i)}_{A',\alv'_{c}}$ is $\id$, we use Eq.~\eqref{eq:4-design-int-1-qubit} which leads to
\begin{align}
  \Ebb_{\beta_i}\Tr[\rho_{\beta_i}^{\otimes 4} (Z^{(i) \otimes 2}_{A,\alv_{c}} \otimes Z^{(i)}_{A',\alv'_{c}}) ]= \left\{
\begin{array}{ll}
    1/5 & \mbox{ if } Z^{(i)}_{A,\alv_{c}} = Z^{(i)}_{A',\alv'_{c}} \;, \\ [.161cm]
    1/15 & \mbox{ if } Z^{(i)}_{A,\alv_{c}} \neq  Z^{(i)}_{A',\alv'_{c}}\;.
\end{array}
\right.
\end{align}

\medskip

Given Pauli-strings $Z_{A,\alv_{c}} = \bigotimes_{i=1}^n Z^{(i)}_{A,\alv_{c}}$ and $Z_{A',\alv'_{c}} = \bigotimes_{i=1}^n Z^{(i)}_{A',\alv'_{c}}$, we do a qubit-wise comparison to see which case for each qubit in these following four different cases it is in
\begin{enumerate}
    \item  $ Z^{(i)}_{A,\alv_{c}} =  Z^{(i)}_{A',\alv'_{c}} = \id \;.$
    \item $ Z^{(i)}_{A,\alv_{c}} =  Z^{(i)}_{A,\alv_{c}} \neq \id \;.$
    \item $ Z^{(i)}_{A,\alv_{c}} \neq  Z^{(i)}_{A',\alv'_{c}} \neq \id \;.$ 
    \item  $ Z^{(i)}_{A,\alv_{c}} \neq  Z^{(i)}_{A',\alv'_{c}}$ and one of them is an identity.
\end{enumerate}
Let us denote $n_{11}$ as the number of qubits in the 1. case, $n_{zz}$ as the number of qubits in the 2. case, $n_{zz'}$ as the number of qubits in the 3. case, and $n_{1z}$ as the number of qubits in the 4. case. Notice that $n = n_{11} + n_{zz} + n_{zz'} + n_{1z}$. In addition, the relations to the sets of qubits that $Z_{A,\alv_{c}}$ and $Z_{A',\alv'_{c}}$ act non-trivially on i.e., $\sup(Z_{A,\alv_{c}})$ and $\sup(Z_{A',\alv'_{c}})$, as well as their associated light cones are as follow:
\begin{align}
    n_{zz} + n_{zz'}& = | \sup(Z_{A, \alv_{c}}) \cap \sup(Z_{A', \alv'_{c}})| \;, \\ 
    n_{1z} & = |\sup(Z_{A, \alv_{c}}) \cup \sup(Z_{A', \alv'_{c}})| - | \sup(Z_{A, \alv_{c}}) \cap \sup(Z_{A', \alv'_{c}})| \\
    &= \Delta_{Z_{A}}(\alv_{c}) + \Delta_{Z_{A'}}(\alv'_{c}) - 2 | \sup(Z_{A, \alv_{c}}) \cap \sup(Z_{A', \alv'_{c}})| \;.\label{eq:proof-purity-end}
\end{align}

Using these equations from Eq.~\eqref{eq:proof-purity-start} to Eq.~\eqref{eq:proof-purity-end}, we can compute the coveriance term in Eq.~\eqref{eq:proof-sup-purity-key} as follows
\begin{align}
    \Cov_{\btv}\left[  \Tr\left[ \rho_{\btv}^{\otimes 2}  Z^{\otimes 2}_{A, \alv_{c}}  \right] ,  \Tr\left[ \rho_{\btv}^{\otimes 2}  Z^{\otimes 2}_{A', \alv'_{c}}  \right]\right] 
 & = \prod_{i=1}^n \Ebb_{\beta_i}\Tr\left[ \rho_{\beta_i}^{\otimes 4}  (Z^{(i) \otimes 2}_{A, \alv_{c}} \otimes  Z^{(i) \otimes 2}_{A',\alv'_{c}})\right]- \prod_{i=1}^n  \Ebb_{\beta_i} \Tr\left[ \rho_{\beta_i}^{\otimes 2}  Z^{(i)\otimes 2}_{A, \alv_{c}}  \right]\Ebb_{\beta_i} \Tr\left[ \rho_{\beta_i}^{\otimes 2}  Z^{(i)\otimes 2}_{A, \alv_{c}}  \right] \\
 & =  \left(\frac{1}{5}\right)^{n_{zz}}\left(\frac{1}{15}\right)^{n_{zz'}}\left(\frac{1}{3}\right)^{n_{1z}} - \left(\frac{1}{3}\right)^{\Delta_{Z_A}(\alv_{c}) + \Delta_{Z_{A'}}(\alv'_{c})} \\
 & = \left(\frac{1}{5}\right)^{n_{zz}+n_{zz'}}\left(\frac{1}{3}\right)^{n_{zz'}}\left(\frac{1}{3}\right)^{n_{1z}} - \left(\frac{1}{9}\right)^{n_{zz}+n_{zz'}}\left(\frac{1}{3}\right)^{n_{1z}} \\
 & = \left(\frac{1}{3}\right)^{n_{1z}}\left(\frac{1}{9}\right)^{n_{zz}+n_{zz'}}\left(\left(\frac{9}{5}\right)^{n_{zz}+n_{zz'}}\left(\frac{1}{3}\right)^{n_{zz'}}-1\right)  \\
 & =  \left(\frac{1}{3}\right)^{\Delta_{Z_{A}}(\alv_{c}) + \Delta_{Z_{A'}}(\alv'_{c})}\left[\left(\frac{9}{5}\right)^{\left| \sup(Z_{A, \alv_{c}}) \cap \sup(Z_{A', \alv'_{c}})\right|}\left(\frac{1}{3}\right)^{n_{zz'}}-1\right] \;.
\end{align}

Hence, the lower-bound on the purity term can be expressed as
\begin{align}
     \Var_{\thv} [\MC(\rho_{\thv}, \rho_{\thv})] \geq \Ebb_{(A, \alv_{c}), (A', \alv'_{c}) \sim \tilde{\DC}_{c}(\sigma)}\left(\frac{1}{3}\right)^{\Delta_{Z_{A}}(\alv_{c}) + \Delta_{Z_{A'}}(\alv'_{c})}\left[\left(\frac{9}{5}\right)^{\left| \sup(Z_{A, \alv_{c}}) \cap \sup(Z_{A', \alv'_{c}})\right|}\left(\frac{1}{3}\right)^{n(Z_{A,\alv_{c}},Z_{A',\alv'_{c}})}-1\right] \geq 0 \;,
\end{align}
where we introduce a new notation $n(Z_{A,\alv_{c}},Z_{A',\alv'_{c}}) := n_{zz'}$ to emphasise the dependence on $Z_{A,\alv_{c}}$ and $Z_{A',\alv'_{c}}$. Note that since we begin with $\Var_{\btv} \left[ \Ebb_{\alv} \MC(\rho_{\thv}, \rho_{\thv})  \right]$ in Eq.~\eqref{eq:proof-sup-prop-purity0} and we do not do any further lower bound, our final expression is guaranteed to be non-negative. This completes the proof of the supplemental proposition.
\end{proof}

\subsubsection{Proof of Theorem~\ref{thm:mmd-train-general-appx}}
\begin{proof}
We consider the generic form of the MMD variance in Eq.~\eqref{eq:sup-prop-mmd-variance-pauli-rotations} from Supplemental Proposition~\ref{sup-prop:mmd-variance-pauli-rotations}. The lower bound of the MMD variance can be written as 
\begin{align}
    \Var_{\thv}[\LC^{(\sigma)}_{\rm MMD}(\thv)] & \geq 4   \sum_{\substack{A \subseteq \NC \\ A\neq \{\;\}}} \left( c_{\sigma}(A) z_A(\tilde{P}) \right)^2 \Ebb_{\alv_{c} \sim \DC_M}\left( \frac{1}{3}\right)^{\Delta_{Z_A} (\alv_{c})} \\
    & \geq 4   \sum_{\substack{A \subseteq \NC \\ A\neq \{\;\}}} \left( c_{\sigma}(A) z_A(\tilde{P}) \right)^2 \left( \frac{1}{3}\right)^{\Ebb_{\alv_{c} \sim \DC_M}[\Delta_{Z_A} (\alv_{c})]} \\
    & =  4   \sum_{\substack{A \subseteq \NC \\ A\neq \{\;\}}} \left( c_{\sigma}(A) z_A(\tilde{P}) \right)^2 \left( \frac{1}{3}\right)^{\Delta_{Z_A}^{\rm (avg)} } \;,
\end{align}
where the first inequality is obtained by throwing away the non-negative purity term, and in the second inequality is by applying Jensen's inequality, leading to $ \Ebb_{\alv_{c} \sim \DC_M}\left( \frac{1}{3}\right)^{\Delta_{Z_A} (\alv_{c})} \geq \left( \frac{1}{3}\right)^{\Ebb_{\alv_{c} \sim \DC_M}[\Delta_{Z_A} (\alv_{c})]}$. The last equality is by using the definition of the average light cone in Eq.~\eqref{eq:avg-light-cone}. 

Now, from Proposition~\ref{prop:mmd-k-body} the MMD observable acts as few-body for the linear bandwidth $\sigma \in \Theta(n)$ with $c_{\sigma}(A) \in \Omega(1/\poly(n))$ for $|A| \in \OC(\log(n))$. Hence, by the assumption that there exists $A$ with $|A|\in\OC(\log(n))$ and the average light cone on at most $\log(n)$ qubits together with the polynomial scaling of $z^2_A(\tilde{P})$, that is
\begin{align}
   \Delta_{Z_A}^{\rm (avg)} \in \OC(\log(n)) \;\;\, \;\; \exists A \;{\rm such\;that}\; |A| \in \OC(\log(n)) \;{\rm and \;}z^2_A(\Tilde{P}) \in \Omega(1/\poly(n))  \;,
\end{align}
we have that the MMD variance lower bound scales polynomially with the number of qubits
\begin{align}
    \Var_{\thv} [\LC^{(\sigma)}_{\rm MMD}(\thv)] \in \Omega(1/\poly(n)) \;.
\end{align}
This completes the proof of the theorem.
\end{proof}

\subsection{Beyond loss gradients - resolving high-order correlations with the MMD}\label{app:mmd-faithfulness}

Our results so far (Theorem~\ref{thm:mmd-sigma} and Theorem~\ref{thm:mmd-train-general}) indicate that picking a single bandwidth $\sigma \in \Theta(n)$ maximizes the expected magnitude of gradients for a randomly initialized QCBM. However, while non-vanishing gradients are necessary, they are not sufficient to guarantee reliable training performance. 

As discussed in the main text and Supplementary Note~\ref{app:mmd-k-body-observable}, the MMD observable can be decomposed into a weighted sum of Pauli-Z strings ranging from low-body to global interaction terms. 
For $\sigma \in \Theta(n)$, Proposition~\ref{prop:mmd-k-body} ensures that the MMD observable is largely composed of low-body terms, with the contribution from global terms negligible. While this leads to substantial cost gradients, we will argue that losses composed purely of low-body terms struggle to learn global properties of the target distribution. In particular, we will argue that an MMD-type loss that is at most $2k$ bodied cannot distinguish between two distributions with the same marginals on $k$-qubits but which differ on higher-order marginals. In Supplementary Note~\ref{app:faithful_arbitrary} we generalise this argument to a broader family of losses for generative modelling. 

For a given subset of bits $A \subseteq \NC=\{1, 2, ..., n\}$, denote $\xv_A$ as a part of the bitstring $\xv$ on that subset $A$ and $\xv_{\bar{A}}$ as the rest of the bitstring $\xv$. The full bitstring $\xv$ can be expressed (not in the right bit order) as $\xv = (\xv_A, \xv_{\bar{A}})$.
Then, the marginal probability of the training distribution on $A$ can be expressed as
\begin{align}\label{eq:marginal_prob}
    \pt(\xv_A) = & \sum_{\xv_{\bar{A}} \in \{0,1\}^{\otimes (n - |A|)}} \pt(\xv_A, \xv_{\bar{A}})  \\
    = & \Tr\left[\rho_{\tilde{p}} \left( |\xv_A\rangle\langle\xv_A| \otimes \id_{\bar{A}} \right)\right]\label{eq:p_xA_marginal} \;,
\end{align}
where the sum is over all constellations of the bits that are not in $A$. In the second line, $\pt(\xv_A)$ is equivalently expressed as the expectation value of a projector onto a computational basis of the subsystem $A$ with the training quantum state $\rho_{\tilde{p}}$. 
Similarly, the marginal distribution of the model on $A$ is of the form
\begin{align}\label{eq:marginal_prob_model}
    \qth(\xv_A) = &\sum_{\xv_{\bar{A}} \in \{0,1\}^{\otimes (n - |A|)}} \qth(\xv_A, \xv_{\bar{A}}) \\
    = &  \Tr\left[\rho_{\thv} \left( |\xv_A\rangle\langle\xv_A| \otimes \id_{\bar{A}} \right)\right]\label{eq:q_xA_marginal} \;.
\end{align}
Physically, we note that, when marginals of two distributions agree up to $k$-bits, this implies that the diagonal elements of the reduced density matrices on any subsets of $k$ qubits are also identical. That is, for all $A\subseteq \NC$ such that $|A|\leq k$, if $\pt(\xv_A) = \qth(\xv_A)$, we have
\begin{align}
    Diag(\Tr_{\bar A} [\rho_{\Tilde{p}}]) =  Diag(\Tr_{\bar A} [\rho_{\thv}]) \;,
\end{align}
where $\bar{A}$ is a complementary of $A$.

In addition, we recall that the truncated version of the MMD observable defined in Eq.~\eqref{eq:mmd-obs-binom} is of the form
\begin{align}
    \Tilde{O}^{(\sigma, k)}_{\rm MMD} = &  \sum_{l=0}^k {n \choose l} (1-p_\sigma)^{n-l} p_\sigma^l \, D_{2l} \\
    = &  \sum_{\substack{A\subseteq \NC \\ |A|\leq k}}(1-p_\sigma)^{n-|A|}p_\sigma^{|A|}\bigotimes_{i\in A}(Z_i\otimes Z_{n+i})\;,
\end{align}
where in the second line the observable is re-written explicitly as the sum over $A$ (with $l = |A|$). Then, the truncated version of the MMD loss can be expressed as
\begin{align}
    \Tilde{\LC}^{(\sigma, k)}_{\rm MMD}(\thv) = & \Tr\left[ \tilde{O}^{(\sigma,k)}_{\rm MMD}  (\rho_{\thv} \otimes \rho_{\thv})\right] - 2 \Tr\left[\tilde{O}^{(\sigma,k)}_{\rm MMD}  (\rho_{\thv} \otimes \rho_{\tilde{p}})\right]+ \Tr\left[  \tilde{O}^{(\sigma,k)}_{\rm MMD} (\rho_{\tilde{p}} \otimes \rho_{\tilde{p}})\right]  \\
    = & \Tr\left[ \tilde{O}^{(\sigma,k)}_{\rm MMD}  (\rho_{\thv} - \rho_{\tilde{p}})^{\otimes 2}\right] \label{eq:C120}
    \; . 
\end{align}

\noindent We are now ready to state and prove the following proposition.
\begin{proposition}[The truncated MMD loss is not faithful]
   Consider a distribution $\qth(\xv)$ that agrees with the training distribution $\pt(\xv)$ on all the marginals up to $k$ bits, but disagrees on higher-order marginals. The distribution $\qth(\xv)$ minimizes the truncated MMD loss. That is, suppose
   \begin{align}
       \qth(\xv_A) = \pt(\xv_A) \;,
   \end{align}
    for all $A\subseteq\{1,2,...,n\}$ with $|A| \leq k$, then
    \begin{align}
       \tilde{\LC}^{(\sigma, k)}_{\rm MMD}(\thv) = 0 \;.
   \end{align}
    Crucially, this is true even if for some $B\subseteq\{1,2,...,n\}$ with $|B| > k$
   \begin{align}
       \qth(\xv_B) \neq \pt(\xv_B) \;.
   \end{align}
\end{proposition}
\begin{proof} Our proof idea is to express the truncated MMD loss in terms of the marginals up to $k$ bits and show that the loss is minimized when the marginals up to $k$ bits match. First, we explicitly expand $ \Tilde{O}^{(\sigma, k)}_{\rm MMD}$ in the truncated MMD loss function in Eq.~\eqref{eq:C120} leading to
\begin{align} \label{eq:mmd-loss-za}
    \Tilde{\LC}^{(\sigma, k)}_{\rm MMD}(\thv) = &  \sum_{\substack{A\subseteq \NC \\ |A|\leq k}}(1-p_\sigma)^{n-|A|}p_\sigma^{|A|} \left[\langle Z_A \rangle_{\thv} - \langle Z_A \rangle_{\Tilde{p}} \right]^2 \;,
\end{align}
where we introduce the shorthand notations $Z_A =\bigotimes_{i\in A} Z_i$, $\langle Z_A \rangle_{\thv} = \Tr[ Z_A \rho_{\thv}]$ and $\langle Z_A \rangle_{\Tilde{p}} = \Tr[ Z_A \rho_{\tilde{p}}]$. 
By expressing $Z_A$ in the computational basis, we have 
\begin{align}
Z_A=\sum_{\xv_A}(-1)^{\sum_{i\in A}x_i}|\xv_A\rangle\langle\xv_A|\otimes\id_{\bar{A}}\;.    
\end{align}
So, the expectation of $Z_A$ can be written as a sum of the marginals probabilities on $A$ with the definition in Eq.~\eqref{eq:marginal_prob} as follows,
\begin{align}
    \langle  Z_A \rangle_{\Tilde{p}} = \sum_{\xv_A}(-1)^{\sum_{i\in A}x_i}\pt(\xv_A) \;,
\end{align}
and 
\begin{align}
    \langle  Z_A \rangle_{\thv} = \sum_{\xv_A}(-1)^{\sum_{i\in A}x_i}\qth (\xv_A) \;.
\end{align}
Then, we have 
\begin{align}\label{eq:marginal-from-z-model}
    \Tilde{\LC}^{(\sigma, k)}_{\rm MMD}(\thv) = &  \sum_{\substack{A\subseteq \NC \\|A|\leq k}}(1-p_\sigma)^{n-|A|}p_\sigma^{|A|} \left[ \sum_{\xv_A}(-1)^{\sum_{i\in A}x_i} (\pt(\xv_A) - \qth (\xv_A)) \right]^2\\
    = & 0 \;,
\end{align}
which completes the proof. Note that we do not need information of the marginals beyond $k$ bits and therefore this leads to the unfaithfulness in the sense that higher-order marginals can disagree even with the truncated loss being minimized.

\end{proof}

In Proposition~\ref{prop:mmd_not_faithful}, we show that if the marginals between the model and training distributions match up to $k$ bits, then the truncated loss of order $k$ is minimized with the model distribution. We now show that the inverse direction also holds. That is, minimizing the truncated loss means learning the marginals of the training distribution.

To show this, we consider again the truncated MMD loss in Eq~\eqref{eq:mmd-loss-za} and notice that the loss is minimized  and equals to $0$ if and only if
\begin{align} \label{eq:z_p_z_q_equal}
    \langle Z_A \rangle_{\thv} = \langle Z_A \rangle_{\Tilde{p}} \;,
\end{align}
for all $A \in \NC$ such that $|A| \leq k$. Concerning the marginal probabilities, we now decompose the projector $|\xv_A\rangle\langle\xv_A|\otimes\id_{\bar{A}}$ in the Pauli basis
\begin{align} 
    |\xv_A\rangle\langle\xv_A|\otimes\id_{\bar{A}} &=\bigotimes_{i\in A} \frac{1}{2}\left(\id_i+(-1)^{x_i}Z_i\right) \\ \label{eq:partial_projector_as_Pauli}
    & =\frac{1}{2^{|A|}}\sum_{B\subseteq A} (-1)^{\sum_{i\in B}x_i} Z_B \;,
\end{align}
where we expanded the product fully with $Z_B = \bigotimes_{i\in B} Z_i$, where $B$ are all possible subsets of $A$. Thus any $k$ bit marginal can be computed from a sum of the average parities of all subsets up to $k$ bits via 
\begin{align}\label{eq:marginalsparity}
    \pt(\xv_A)  = & \frac{1}{2^{|A|}}\sum_{B\subseteq A} (-1)^{\sum_{i\in B}x_i}   \langle Z_B \rangle_{\Tilde{p}} \, .
\end{align}
It is clear from Eq.~\eqref{eq:mmd-loss-za} that training on the $k$-truncated MMD learns all average parities of the target distribution up to $k$ bits, Eq.~\eqref{eq:z_p_z_q_equal}, and hence Eq.~\eqref{eq:marginalsparity} implies we also learn all marginals up to and including $k$ bits. Put another way,
the difference between model and training marginal probabilities is given by
\begin{align}
    \pt(\xv_A) - \qth(\xv_A) = & \frac{1}{2^{|A|}}\sum_{B\subseteq A} (-1)^{\sum_{i\in B}x_i}  \left( \langle Z_B \rangle_{\Tilde{p}} - \langle Z_B \rangle_{\thv} \right)\;, \\
    = & 0 \; ,
\end{align}
for all $A$ such that $|A| \leq k$.

\subsection{Distinguishing marginals using an arbitrary loss}\label{app:faithful_arbitrary}

In Supplementary Note~\ref{app:mmd-faithfulness}, we have shown that a truncated MMD loss operator cannot distinguish between model distributions that agree with the data distribution up until a certain order of marginals, but disagree beyond. In this section, we show that this phenomenon can be extended to general generative losses for classical data that can be formulated as the expectation value of some observable. As a key example, we first consider loss functions $ \LC(\thv) = \Tr[O \rho_{\thv}]$ with an observable in the following form,
\begin{align}
    O = \sum_{\xv \in \XC} D_{\vec{\alpha}}(\xv) |\xv \rangle\langle \xv| \;.
\end{align}
Here, $D_{\vec{\alpha}}(\xv)$ is the eigenvalue of the operator corresponding with the computational basis sample $\xv$, which could additionally be parametrized by $\vec{\alpha}$. Notably, the loss for the Generator in a quantum GAN can be expressed in this form, where $D_{\vec{\alpha}}(\xv)$ is the classification output of the Discriminator. 

The truncated version of $O$ in the Pauli basis up to $k$-body terms can be expressed as
\begin{align}
   O^{(k)} = \sum_{\substack{A\subseteq \NC \\|A|\leq k}} c_A Z_A \;,
\end{align}
where $c_A = \frac{1}{2^n} \sum_{\xv} D_{\vec{\alpha}}(\xv)(-1)^{\sum_{i\in A} x_i}$, and $Z_A = \bigotimes_{i \in A} Z_i$ are Pauli operators acting non-trivially on qubits $i$ in a subset of qubits $A$.

Now, we show that the loss assigned by the truncated loss function is the same between that the model state $\rho_{\thv_1}$, which matches the training distribution exactly (i.e., it is the global optimum of the full loss but not necessarily of the truncated one), and a state  $\rho_{\thv_2}$, matches the training distribution up to $k$-bit marginals but disagrees beyond. Both are characterized by the property
\begin{equation}\label{eq:marginals_match}
    \Tr\left[\rho_{\thv_1} \left( |\xv_A\rangle\langle\xv_A| \otimes \id_{\bar{A}} \right)\right] = \Tr\left[\rho_{\thv_2} \left( |\xv_A\rangle\langle\xv_A| \otimes \id_{\bar{A}} \right)\right] = \pt(\xv_A) \, .
\end{equation}
for all $A \subseteq \NC$ such that $|A| \leq k$.  
This implies that the reduced states $\rho_{\thv_1,A} = \text{Tr}_{\bar{A}}[\rho_{\thv_1}]$, $\rho_{\thv_2, A} = \text{Tr}_{\bar{A}}[\rho_{\thv_2}]$ and $\rho_{\tilde{p}, A} = \text{Tr}_{\bar{A}}[\rho_{\Tilde{p}}]$ have the same diagonal, that is $Diag(\rho_{\thv_1,A}) = Diag(\rho_{\thv_2,A}) =  Diag(\rho_{\Tilde{p},A})$. Consequently, the expectations of any diagonal Pauli operator $Z_A$ on $A$ are the same,
\begin{equation}\label{eq:diags_match}
    \langle Z_A \rangle_{\thv_1} = \langle Z_A \rangle_{\thv_2} =  \langle Z_A \rangle_{\Tilde{p}}\,.
\end{equation}
Now consider the difference in the loss between any of these two states i.e., $\rho, \rho' \in \{ \rho_{\thv_1}, \rho_{\thv_2}, \rho_{\Tilde{p}}\}$
\begin{align}
    d(\rho,\rho') &= \mathcal{L}(\rho) - \mathcal{L}(\rho') \\
    &= \sum_{\substack{A\subseteq \NC \\|A|\leq k}} c_A \left(\Tr_{\bar A} [Z_A \rho] -\Tr_{\bar A} [Z_A \rho']\right)\\
    & = 0 \;.
\end{align}

This shows that any loss composed exclusively of low-body terms cannot distinguish between distributions with the same low-order marginals (but potentially different higher-order marginals). Beyond losses that are natively composed entirely of low-body operators, this result becomes of practical importance when the loss $ \LC(\thv) = \sum_{\xv \in \XC} \Tr[ D_{\vec{\alpha}}(\xv) |\xv \rangle\langle \xv| \rho_{\thv}]$ is effectively composed entirely of low-body terms, i.e., if the global contributions to $\LC(\thv)$ are too small to be resolved using the available shot budget. In that case, $\LC(\thv)$ is well approximated by its truncated version and our argument applies. Whether this is or is not the case is determined by the choice in $D_{\vec{\alpha}}(\xv)$.

We note that in this derivation we left the structure of $D_{\vec{\alpha}}$ entirely general up to the constraint that it is diagonal in the computational basis and acts only on a single copy of the model distribution at a time. 
But our proof can be directly applied to general generative losses $\LC_{gen}(\thv, \{C_O\})$ for classical data which take expectation values of several observables $C_O = \Tr[O \rho_{\thv}^{\otimes m}]$ with operators $O$ acting on $m$ different sub-systems, each of which contains up to $k$-body terms in the Pauli basis. That is, we have
the general operator of the form
\begin{align}
    O = \sum_{\substack{A_1,...,A_m \subseteq \NC \\|A_1|,...,|A_m|\leq k}} c_{A_1,...,A_m}(\boldsymbol\alpha, \tilde{P}) \left( Z_{A_1} \otimes Z_{A_2} \otimes ... \otimes Z_{A_m}\right) \;, 
\end{align}
where, $A_1, A_2, ..., A_m$ are subsets of $\NC = \{1,...,n\}$, $c_{A_1,...,A_m}(\boldsymbol\alpha, \tilde{P})$ are real coefficients that can depend on training data, and $Z_{A_j} = \bigotimes_{i\in A_j} Z_{(j-1)n + i}$ acting non-trivially on the qubits of the $j^{\rm th}$ subsystem. 

This form of the general loss covers loss functions for quantum circuit Born machines (in particular the MMD, as outlined in Sec.~\ref{sec:mmd}), quantum GANS, quantum Boltzmann machines~\cite{QBM_amin}, and any other proposed (quantum) generative model on classical discrete data. 
Therefore, any diagonal loss operator that implements a generative modelling loss and contains only low-bodied operators cannot be used to reliably learn global probability marginals.

\section{Supplementary Note - Analysis on the quantum fidelity loss}\label{ap:localcost}

In this section, we present an approach that can be used to estimate the local fidelity quantity $\LC_{QF}^{(L)}(\thv)$ using a series of Hadamard tests without explicitly loading the training data into a quantum state or requiring a quantum oracle.

We begin by introducing a pure quantum state corresponding to the training dataset $\ket{\phi} = \sum_{\xv} \sqrt{\pt(\xv)} \ket{\xv}$. Learning the training distribution is equivalent to the state learning task with $\ket{\phi}$ as the target state. We note that our choice of having the coefficient as $\sqrt{\pt(\xv)}$ is arbitrary and, generally, any target quantum state with the probabilities corresponding to the training probabilities would be valid candidates for the task.

As discussed in the main text, the quantum fidelity can be used as a cost function in this learning task
\begin{align}
    \LC_{QF}(\thv) = 1 - |\langle \phi | \psi(\thv) \rangle|^2 \;.
\end{align}
However, the globality of the loss leads to barren plateaus, which in turn leads to the untrainability of the loss. 
However, the local version of the quantum fidelity has been shown to be both trainable with the shallow depth circuits and faithful to the original global version~\cite{khatri2019quantum}. Specifically, the local quantum fidelity is of the form 
\begin{align}
    \LC_{QF}^{(L)}(\thv) = 1 - \bra{\phi} U(\thv) H_L U^{\dagger}(\thv) \ket{\phi} \;,
\end{align}
with
\begin{align}
    H_L = \frac{1}{n} \sum_{i=1}^n |0_i\rangle\langle 0_i | \otimes \id_{\bar{i}} \;,
\end{align}
where $|0_i\rangle\langle 0_i |$ are now single-qubit projectors.
In the context of the generative modeling with classical, there is an additional challenge as only the training dataset $\Tilde{P}$ is given to us and not the state $\ket{\phi}$. Using Hadamard tests, we now show that measuring the local quantum fidelity can be achieved efficiently without loading the training data into the quantum state.
To see this, we first express $H_L$ in the Pauli basis
\begin{align}
    H_L = &  \frac{1}{n} \sum_{i = 1}^n \left(\frac{\id_i + Z_i}{2} \right) \otimes \id_{\bar{i}} \\
    = & \frac{\id}{2} + \frac{1}{2n} \sum_i Z_i \;.
\end{align}
We then expand $\LC_{QF}^{(L)}(\thv)$ as
\begin{align}
    \LC_{QF}^{(L)}(\thv) = & 1 - \bra{\phi} U(\thv) \left( \frac{\id}{2} + \frac{1}{2n} \sum_i Z_i \right) U^{\dagger}(\thv) \ket{\phi} \\ \label{eq:local_fidelity_cost}
    = & \frac{1}{2} - \frac{1}{2n} \sum_i  \bra{\phi} U(\thv) Z_i U^{\dagger}(\thv) \ket{\phi}  \\
    =  & \frac{1}{2} - \frac{1}{2n} \sum_{i,\xv, \vec{x'}} \sqrt{\pt(\xv) \pt(\vec{x'})} \bra{\xv} U(\thv) Z_i U^\dagger(\thv) \ket{\vec{x'}} \\
    =  & \frac{1}{2} - \frac{1}{2n} \sum_{i,\xv, \vec{x'}} \sqrt{\pt(\xv) \pt(\vec{x'})} \bra{\xv} U(\thv) Z_i U^\dagger(\thv) U_{\vec{x'},\xv}\ket{\xv} \\
    =  & \frac{1}{2} - \frac{1}{2n} \sum_{i,\xv, \vec{x'}} \sqrt{\pt(\xv) \pt(\vec{x'})} \bra{\xv}  \Tilde{U} (\thv, i, \xv, \vec{x'}) \ket{\xv} \;,
\end{align}
where in the third line we explicitly expand $\ket{\phi} = \sum_{\xv} \sqrt{p(\xv)} \ket{\xv}$, and in the fourth line we introduce $U_{\vec{x'}, \xv}$ which is the unitary mapping from the computational basis $\xv$ to $\vec{x'}$ by applying the necessary single-qubit flips. Finally, in the last line, we introduce $ \Tilde{U} (\thv, i, \xv, \vec{x'})= U(\thv) Z_i U^\dagger(\thv) U_{\vec{x'},\xv}$ which summarizes the full unitary to be implemented for any pair $\xv$ and $\vec{x'}$. Each term can now be estimated using two Hadamard tests, i.e., one for the real part and the one for the imaginary part of each overlap. If $N_p \in \OC(\poly(n))$ is the number of unique bitstrings in the training dataset $\Tilde{P}$, one thus requires $2nN_p^2 \in \OC(
\poly(n))$ Hadamard tests to evaluate all terms in the loss. The number of Hadamard tests can be reduced by a factor of $2$ if the quantum model is constructed to span only either real or imaginary subspace. The number of controlled unitaries can also be reduced to simple control phase gates by using a diagonal ansatz~\cite{commeau2020variational}. 
While the number of Hadamard tests naively scales with the number of training bitstrings, this overhead is expected to be significantly reduced by employing techniques such as stochastic gradient descent~\cite{sweke2020stochastic} which allows us to stochastically optimize the loss in an unbiased manner.

\section{Supplementary Note - Additional training with exponential support}
\label{app:training}
As stated in the main text, QCBMs are not expected to be able to learn distribution with exponential support when the training data are stored in classical computers. This is since even storing so much data is unrealistic beyond a few dozens of qubits. However, due to the focus on systems with very few qubits, most applications taken from the literature utilize on distribution with non-zero probabilities on a macroscopic number of bitstrings. Such cases can still be trainable with the KL divergence and finite shots if the number of qubits is small enough ($n\leq 12$). We perform the same training procedure as in the main text in Sec.~\ref{sec:train-hea-dataset}, with the additions of gradient batching over $k=10$ iterations, and gradient clipping with a threshold of $\tau = 0.1$. These details aim at stabilising the optimisation and improve the performance, and follow best practices. Fig.~\ref{training_6_app} shows this numerically on the ECAL dataset with $n=6$ qubits. More particularly, we train on a transformed version of the dataset that follows $-\log(\pt(\xv))$, which exhibits exponential support (top three rows) and additionally on the original probabilities (bottom three rows), which have only a polynomial support, for different number of shots ($10^2,\,10^3,\,10^4$ and $\infty$) and layers ($1,\,8$ and $16$). Each column is divided into two blocks of three plots each, where the first shows the TVD during the training, the second the generated histograms with the best ansatz against the target distribution, and the third the absolute error between the two distributions. We observe that the three loss functions are comparable for the exponential support case, while KLD breaks down when the support is polynomial and using fewer than $10^3$ shots. 

\begin{figure*}
    \centering
    \includegraphics[scale=0.4]{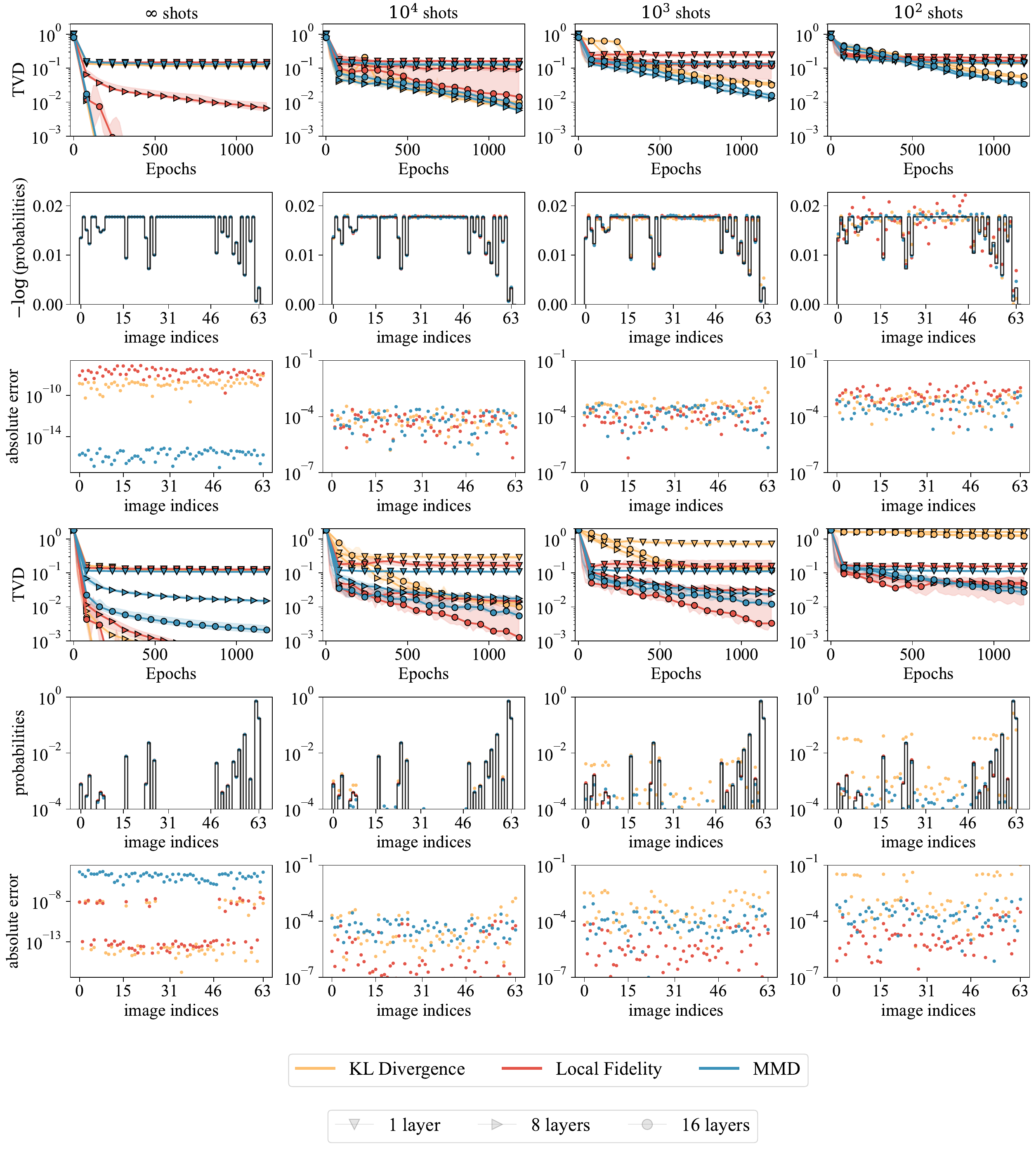}
    \caption{Training on the ECAL dataset with $n=6$ qubits on a dataset with exponential support (top three rows) and one with polynomial support (last three rows). The first rows shows the TVD during training, while the second displays the generated distribution against the target one (black) and the third the absolute error between the two. }
    \label{training_6_app}
\end{figure*}

We recall that these numerics are not a contradiction with the message of this paper, since learning a distribution with exponential support is not scalable using QCBMs, but may explain why the exponential concentration issue has not been discussed for quantum generative models before.

\end{document}